\documentclass[11pt,letterpaper]{article} 
\usepackage{amsmath,amsthm,verbatim,amssymb,amsfonts,amscd,graphicx,graphics,physics,bm}
\usepackage[latin2]{inputenc}
\usepackage{t1enc}
\usepackage[mathscr]{eucal}
\usepackage{indentfirst}
\usepackage{booktabs}
\usepackage{pict2e}
\usepackage{epic}
\numberwithin{equation}{section}
\usepackage[footnotesize,bf]{caption}
\usepackage{epstopdf}
\usepackage{mathtools}
\usepackage[dvipsnames]{xcolor}
\usepackage{algorithmic}
\usepackage{setspace}
\usepackage[ruled,vlined]{algorithm2e}
\usepackage{authblk}
\usepackage{natbib}
\usepackage{comment}

\usepackage[colorlinks,linkcolor = black, citecolor=blue]{hyperref} 
\usepackage{cleveref}

\definecolor{mypink1}{rgb}{0.858, 0.188, 0.478}
\definecolor{mypink2}{RGB}{219, 48, 122}
\definecolor{mypink3}{cmyk}{0, 0.7808, 0.4429, 0.1412}
\definecolor{mygray}{gray}{0.6}

\theoremstyle{plain}
\newtheorem{theorem}{Theorem}

\newtheorem{lemma}{Lemma}
\newtheorem{proposition}{Proposition}

\theoremstyle{definition}
\newtheorem{definition}{Definition}
\newtheorem{example}{Example}
\newtheorem{assumption}{Assumption}
\newtheorem*{assumption-non}{Assumption}

\newtheorem{innercustomthm}{Assumption}
\newenvironment{customthm}[1]
{\renewcommand\theinnercustomthm{#1}\innercustomthm}
{\endinnercustomthm}

\theoremstyle{remark}
\newtheorem{remark}{Remark}

\def\A{{\mathcal A}}
\def\B{{\mathcal B}}

\def\E{{\mathbb E}}
\def\EE{{\mathcal{E}}}

\def\N{{\mathbb N}}
\def\P{{\mathbb P}}

\def\R{{\mathbb R}}
\def\S{{\mathcal S}}
\def\V{{\mathcal V}}

\def\K{{\bm K}_n}
\def\Q{{\bm Q}_n}

\def\HH{{\mathcal H}}
\def\VV{{\mathcal V}}
\def\EE{{\mathcal E}}
\def\T{{e}}
\def\e{{\varepsilon}}
\def\t{{\bm{\theta}}}

\def\bx{{\bm{x}}}

\def\z{{\bm{z}}}
\def\bu{{\bm{u}}}
\def\u{{\bm{u}}}
\def\v{{\bm{v}}}

\def\uu{{\widehat{\bm{u}}}}
\def\vv{{\widehat{\bm{v}}}}
\def\tt{{\widehat{\bm{\theta}}}}
\def\deg{{\mathrm{deg}}}
\def\br{{\mathrm{br}}}
\def\curl{{\mathsf{curl}}}

\def\diam{{\mathrm{diam}}}

\def\PL{{\mathrm{PL}}}
\def\PLDC{{\mathrm{PlusDC}}}
\def\range{{\mathrm{range}}}

\def\AA{{{\bm W}_n}}
\def\M{{N_{\br}}}
\def\mM{{M}}
\def\muu{{m}}
\def\ttt{{t}} 

\def\HHn{{\mathcal H}}
\def\VVn{{\mathcal V}}
\def\EEn{{\mathcal E}}

\def\tauf{\tau_1}
\def\rhof{\tau_2} 
\def\lambdaf{\rho} 
\def\mathW{\Phi}

\newcommand{\ER}{Erd\H{o}s--R\'{e}nyi }

\renewcommand{\bar}{\overline}

\newcommand{\fgamma}{\iota}
\newcommand{\fGamma}{\mathcal{I}}

\newcommand{\checkedd}[1]{\textcolor{black}{{#1}}}

\renewcommand{\d}[1]{\ensuremath{\operatorname{d}\!{#1}}}

\DeclareMathOperator*{\argmax}{arg\,max}

\newcommand{\CR}[1]{\textcolor{black}{{#1}}}
\def\spacingset#1{\renewcommand{\baselinestretch}%
	{#1}\small\normalsize} \spacingset{1}

\begin{document}

\spacingset{1.17}
\title{\bf Statistical ranking with dynamic covariates\footnotetext{Authorships are ordered alphabetically. Email: ruijian.han@polyu.edu.hk; yiming.xu@uky.edu.}}
\date{}
\author[1]{Pinjun Dong}
\author[2]{Ruijian Han}
\author[2]{Binyan Jiang}
\author[3]{Yiming Xu}
\affil[1]{\small Department of Mathematics, Zhejiang University, Hangzhou, China}
\affil[2]{\small Department of Data Science and Artificial Intelligence, 
	
	\small	 The Hong Kong Polytechnic University, Hong Kong, China}
\affil[3]{\small Department of Mathematics, University of Kentucky, Lexington, USA}

\maketitle

\begin{abstract}
	We introduce a general covariate-assisted statistical ranking model within the Plackett--Luce framework. Unlike previous studies focusing on individual effects with fixed covariates, our model allows covariates to vary across comparisons. This added flexibility enhances model fitting yet brings significant challenges in analysis. This paper addresses these challenges in the context of maximum likelihood estimation (MLE). We first provide sufficient and necessary conditions for both model identifiability and the unique existence of the MLE. Then, we develop an efficient alternating maximization algorithm to compute the MLE. Under suitable assumptions on the design of comparison graphs and covariates, we establish a uniform consistency result for the MLE, with convergence rates determined by the asymptotic graph connectivity. We also construct random designs where the proposed assumptions hold almost surely. Numerical studies are conducted to support our findings and demonstrate the model's application to real-world datasets, including horse racing and tennis competitions. \\ 
	
\end{abstract}

\tableofcontents

\section{Introduction}
Ranking data generated from comparisons among individuals is ubiquitous. A prevalent approach to analyzing such data involves a latent score model, whose foundation dates back to the early works nearly a century ago \citep{thurstone1927method, MR1545015}. Specifically, \cite{MR1545015} proposed to model the choice probabilities of individuals proportional to their latent scores, which was later shown to be the unique choice model satisfying the axiom of \emph{independence of irrelevant alternatives}.
This axiom is also known as Luce's choice axiom \citep{MR0108411} and has led to the development of the more general Plackett--Luce (PL) model with multiple comparison outcomes \citep{plackett1975analysis}.

In the PL model, the ranking outcome $j_1\succ \cdots \succ j_m$ among objects $j_1, \ldots, j_m$ on an edge $\T=\{j_1, \ldots, j_m\}$ is assumed to be observed with probability
\begin{align}
	\P_{\u^*}\left\{j_1\succ \cdots \succ j_m \mid \T\right\} = \prod_{k=1}^m\frac{\exp\{u^*_{j_k}\}}{\sum_{t = k}^m \exp\{u^*_{j_t}\}},\label{intro:1}
\end{align} 
where $u_{j}^*$ and $\exp\{u_j^*\}$ are referred to as the \emph{utility} and \emph{latent score} of the $j$th object, respectively; see \eqref{eq:PL} for a precise definition. When restricted to pairwise edges, the PL model reduces to the celebrated Bradley--Terry (BT) model \citep{MR0070925}, which has been extensively studied in the literature \citep{chen2021spectral}.

Many works on the PL model focus on the primitive scenario where latent scores are purely determined by individual utilities. Under such circumstances, asymptotic theories have been established for common estimators of the utility vector, including the likelihood-based estimators \citep{fan2022ranking, han2023unified} and estimators based on spectral methods \citep{jang2018top, fan2023spectral}. This line of work generalizes the asymptotic theory in the BT model, where similar results \citep{MR1724040, 
	MR3953449, han2020asymptotic, han2022general, gao2023uncertainty} or refined variants \citep{chen2015spectral, MR3613103, MR3504618,  chen2022optimal, chen2022partial, han2024statistical} have been established. Meanwhile, efficient numerical algorithms have been developed to facilitate scalable estimation and inference in the PL model for relatively large datasets \citep{MR2051012, maystre2015fast, agarwal2018accelerated, qu2023sinkhorn}.

However, a gap exists between theory and practice. Before applying the asymptotic theory, a substantial amount of comparison data needs to be gathered first. The gathered data typically spans a wide time range, as in sports analytics \citep{massey1997statistical, glickman1999rating}, or is collected under varying contexts. The latter scenario is common in recent applications of the PL model in fine-tuning large language models where different attributes are employed when presenting the same items to users for data acquisition \citep{ouyang2022training, rafailov2024direct}. Given the nature of the data obtained, it is crucial to incorporate contextual information or additional dynamic attributes into latent score modeling for realistic prediction. 

Motivated by such considerations, we propose to use a simple linear function to account for the dynamic component of the log scores not explained by individual utility based on the existing PL model. This results in modification of \eqref{intro:1} as follows:
\begin{align}
	\P_{\t^*}\left\{j_1\succ \cdots \succ j_m \mid \T\right\} = \prod_{k=1}^m\frac{\exp\{u^*_{j_k} + X^\top_{\T, j_k}\v^*\}}{\sum_{t = k}^m \exp\{u^*_{j_t}+X^\top_{\T, j_t}\v^*\}},\label{intro:2}
\end{align} 
where $X_{\T, j_k}\in\R^d$ is a $d$-dimensional covariate vector associated with object $j_k$ on $\T$, and $\v^* \in\R^d$ is the model coefficient vector. The covariates here are dynamic through the index $\T$ so that the scores of objects may vary across comparisons. See Figure~\ref{fig:dc} for an example of applying \eqref{intro:2} to obtain the dynamic scores of some top tennis players over the last 50 years; population age effects are taken into account so that the players exhibit different peaks during different periods. More details on Figure~\ref{fig:dc} will be given in Section~\ref{sec:atp}.

\begin{figure}[htbp]
	\centering		
	\includegraphics[width= 1\linewidth,  trim={20cm 3cm 20cm 3cm}, clip]{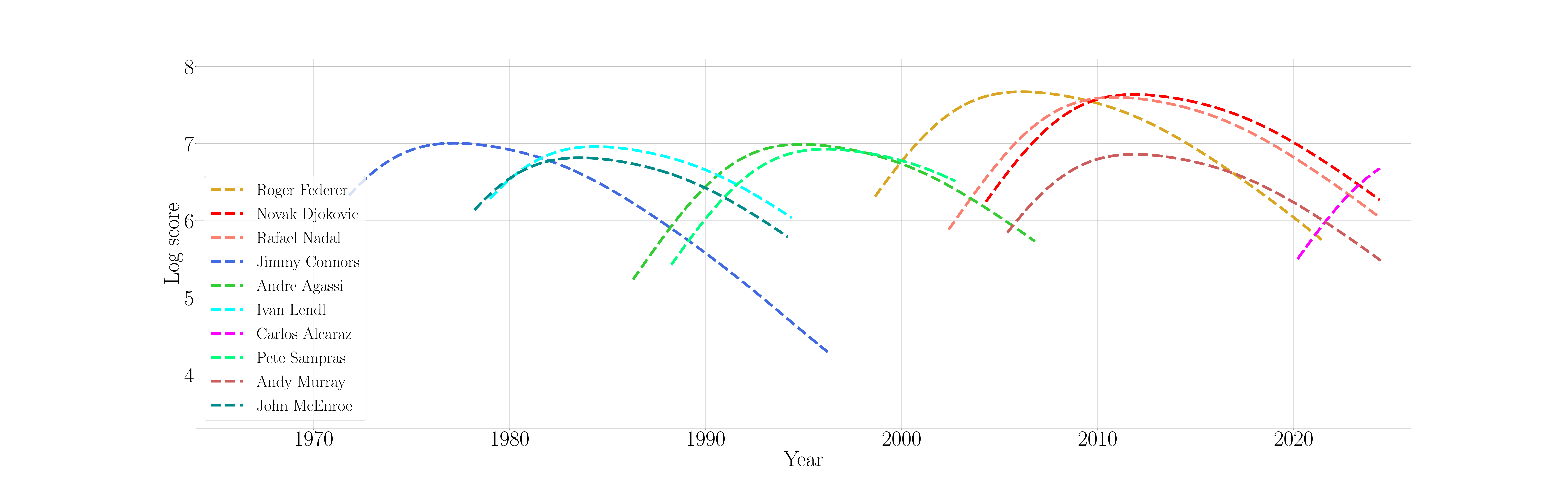}
	\caption{Log score dynamics of the top ten players ranked by model \eqref{intro:2}. For each player, we plot their log scores as a function of time starting from his career and ending at his retirement.}\label{fig:dc}
\end{figure}

\subsection{Related work}

Incorporating covariates into statistical ranking models to improve prediction is not new and has been considered in various settings. In the survey article \citep{MR3012434} and subsequent works \citep{guo2018experimental, schafer2018dyad, zhao2020learning}, the log scores of each individual in the PL model are directly modeled as a linear function of the respective covariates with zero intercepts, based on the premise that the covariates fully explain the scores. In this case, the utility vector $\u^* = \bm 0$, and only $\v^*$ needs to be estimated. The resulting model is similar to a multinomial logistic regression with model coefficients in a fixed dimension. 

Recent research has also explored models incorporating covariates while retaining the individual effects. \cite{li2022bayesian} employed fixed covariates to create informative priors for the utility vector in the Thurstone--Mosteller model \citep{thurstone1927method,mosteller2006remarks}. Meanwhile, they acknowledged the overparameterization issue of using fixed covariates. \cite{fan2022uncertainty} considered a similar setup in the BT model and resolved the issue by imposing additional orthogonal constraints. Under this setup, a systematic study of estimation and inference using MLE was conducted. A follow-up work by \cite{fan2024uncertainty} bridges the gap between the individual-effect and the covariates-only approaches by adding a sparsity constraint on the utility vector. Despite the comprehensive theory developed for these models, the overparameterization issue persists from the model-fitting perspective (Proposition \ref{prop:equiv}). \checkedd{This motivates the consideration of dynamic covariates in the PL model as pursued in this paper.}

Aside from the PL model and its variants, similar ideas of adding covariates also appeared in other network models, such as the $\beta$-type models \citep{graham2017econometric, yan2018statistical} and the cardinal model \citep{singh2024graphical}, among others. Unlike the PL model, the graphs in \cite{graham2017econometric, yan2018statistical} are dense and do not involve model identifiability issues. The model considered in \cite{singh2024graphical} is a cardinal pairwise comparison model, which is mathematically equivalent to a linear model on graphs and admits an explicit solution for the MLE. The theoretical results on models with covariates in \cite{singh2024graphical} primarily focus on the finite-object setting, which differs from ours.

\subsection{Contribution}

The prevailing body of the literature in the previous section considers scenarios where ranking outcomes are either fully explained by covariates or modeled as linear functions of fixed covariates with individual intercepts. Statistical ranking models that simultaneously integrate individual intercepts with dynamic covariates have not been explored. On the other hand, most theoretical work on covariate-assisted ranking models focuses on the pairwise setting. This is somewhat unsatisfactory as multiple comparison data contain richer information despite being challenging to analyze. For example, there is less of a consensus on the choice of operator for hypergraph expansion \citep{friedman1995second, feng1996spectra, lu2013high}. 

Given these reasons, we consider an alternative model defined by \eqref{intro:2}, which jointly accounts for latent individual scores and dynamic covariates in a multiple comparison setting. The intertwinement between hypergraph structures and edge-dependent covariates poses considerable challenges in establishing a theoretical statistical foundation for the model. This paper addresses these challenges, with our contributions summarized as follows:

\begin{itemize}
	\item[{(i)}] We define the model's identifiability and provide a sufficient and necessary condition for it. We also give a sufficient condition based on the non-existence of curl-free flows of covariates, which provides additional interpretation of the results. The corresponding results are given in Propositions~\ref{prop:1} and \ref{prop:curl}. 
	\item[(ii)] For parameter estimation, we consider the MLE approach. We provide a sufficient and necessary condition for the unique existence of the joint MLE $(\uu, \vv)$ for the utility vector $\u^*$ and the linear effect $\v^*$ (Theorem~\ref{thm:existence}). Along the way, we justify the necessity of considering only dynamic covariates in \eqref{intro:2} from a model-fitting perspective (Proposition~\ref{prop:equiv}). Furthermore, we develop an efficient alternating maximization algorithm (Algorithm~\ref{alg:1}) to compute the MLE.
	\item[(iii)] We establish the uniform consistency of the joint MLE for both $\u^*$ and $\v^*$ under appropriate conditions on the graph topology and covariates, with respective convergence rates characterized by the asymptotic connectivity of the graph sequence (Theorem~\ref{thm:main}). \checkedd{Additionally, we provide random hypergraph models and covariate designs for which uniform consistency holds under near-optimal sparsity conditions (Theorem~\ref{thm:randdesign}).}  
\end{itemize}

The crux of our technical contributions is a two-stage analysis, combining ideas of graph-based chaining, empirical processes, and error boosting. \checkedd{In particular, our use of graph-based chaining differs from \cite{han2023unified} and, on its own, is insufficient to establish the uniform consistency result.} We shall also comment on the specific challenges that distinguish \eqref{intro:2} from the fixed-covariate setting. One of our results (Proposition~\ref{prop:equiv}) states that when the covariates are edge-independent, the joint MLE can be obtained by first finding the MLE of a standard PL model and then taking appropriate decomposition. However, such decomposition no longer holds when the covariates become dynamic. \checkedd{Addressing this challenge represents a fundamental contribution of our work that differentiates it from the existing literature.}
\subsection{Organization}

The rest of the paper is organized as follows. In Section~\ref{sec:2}, we introduce the PL model with dynamic covariates, including both the statistical model and the comparison graph design. In Section~\ref{sec:ident}, we define the model's identifiability and discuss conditions under which the proposed model is identifiable. In Section~\ref{sec:est}, we consider a joint MLE approach to estimating the model parameters and provide a sufficient and necessary condition for the unique existence of the MLE; we also show that the static part of the dynamics covariates can be subsumed in the utility vector via reparametrization. An alternating maximization algorithm is also given to find the joint MLE. In Section~\ref{sec:cons}, we establish a uniform consistency result for the MLE with convergence rates characterized using comparison graph topology. \checkedd{In Section~\ref{sec:special case}, we provide specific random hypergraph models and covariates designs where the results in Section~\ref{sec:cons} hold under minimal conditions}. In Section \ref{sec:atp}, we apply the proposed model to analyze an ATP tennis dataset. \checkedd{Additional numerical results on synthetic and horse-racing data and technical proofs are deferred to the supplementary file. The source code to reproduce our numerical results is available at \href{https://github.com/PinjunD/Statistical-ranking-with-dynamic-covariate}{https://github.com/PinjunD/Statistical-ranking-with-dynamic-covariate}.}

\subsection{Notation}

Let $\mathbb N$ and $\mathbb R$ denote the set of natural and real numbers, respectively. For $n\in\mathbb N$, let $[n] = \{1, \ldots, n\}$. For $\T\subseteq [n]$, the power set of $\T$ is denoted by $\mathscr P(\T)$. A permutation $\pi$ on $\T$ is a bijection from $\T$ to itself, and the set of permutations on $\T$ is denoted by $\mathcal S(\T)$. If $|\T| = m$, then $\pi$ can be represented as $[\pi(1), \ldots, \pi(m)]$, where $\pi(i)$ is the element in $\T$ with rank $i$ under $\pi$. Sometimes it is more convenient to work with rank itself, for which we use $r(j)$ to denote the rank of object $j$ in $\pi$. By definition, $r(\pi(j)) = j$ for all $j\in [m]$. In the rest of the paper, we reserve the notation $\u^*\in\R^n$ and $\v^*\in\R^d$ for the true utility vector and the model coefficient vector in the PL model with dynamic covariates. 

We represent a comparison graph as a hypergraph $\HH(\VV, \EE)$, where $\VV=[n]$ is the vertex set and $\EE = \{\T_i\}_{i\in [N]}\subseteq \mathscr{P}(\VV)$ is the edge set. Here, the notation $N$ denotes the number of edges and may depend on $n$. The degree of a vertex $k$ is defined as the number of edges in $\EE$ containing $k$, that is, $\deg(k) = |\{\T_i\in \EE: k\in \T_i\}|$. Following convention, we denote a size-$m$ edge as an increasingly ordered $m$-tuple. We use ${\VV\choose m}= \{e\in\mathscr P(\VV): |e| = m\}=\{(j_1, \ldots, j_m): 1\leq j_1<\cdots < j_m\leq n\}$ to denote the set of hyperedges with fixed size $m$. For $\emptyset\neq U\subset \VV$, the boundary of $U$ is defined as $\partial U = \{e\in \EE: e\cap U\neq\emptyset, e\cap U^\complement\neq\emptyset\}$.

We use $O(\cdot)$ and $o(\cdot)$ for the Bachmann--Landau asymptotic notation.  We use $\lesssim$ and $\gtrsim$ to denote the asymptotic inequality relations, and $\asymp$ if both $\lesssim$ and $\gtrsim$ hold. We use $\mathbb I_{B}(x)$ to denote the indicator function on a set $B$, that is, $\mathbb I_{B}(x)=1$ if $x\in B$ and $\mathbb I_{B}(x)=0$ otherwise. We let $\bm{1}$ and $\bm{I}$ denote the all-ones vector and the identity matrix, respectively, with dimensions that are compatible and often clear from the context.

\section{Methodology}\label{sec:2}
We consider an extended PL model equipped with dynamic covariates. First, we introduce the PL model and then explain how dynamic covariates can be incorporated using a linear function. We also introduce the comparison graph designs used in our analysis, including a deterministic design setting and some common random graph models. In the subsequent analysis, we assume that $\HH(\VV, \EE)$ is simple for notational convenience; multi-edges can be considered similarly with additional indices, which we do not pursue here.

\subsection{PL model with dynamic covariates}
The PL model on $\HH(\VV, \EE)$ with utility vector $\u^*\in\R^n$ is a sequence of probability measures $\P_{\u^*}\{ \ \cdot \mid \T\}$ on $\mathcal S(\T)$, the set of permutations on $\T$, indexed by $\T\in {\mathscr P}(\VV)$ ($\T\neq \emptyset$). For $\T\in {\mathscr P}(\VV)$ with $|\T|=m$ and $\pi\in\mathcal S(\T)$, 
\begin{align}
	\P_{\u^*}\left\{\pi = [\pi(1), \ldots, \pi(m)]\mid \T\right\} = \prod_{j\in [m]}\frac{\exp \big\{u^*_{\pi(j)}\big\}}{\sum_{t=j}^m\exp\big\{u^*_{\pi(t)}\big\}}.\label{eq:PL}
\end{align}
One may alternatively associate $\pi$ with the totally ordered comparison outcome $\pi(1)\succ\cdots\succ\pi(m)$. 
A crucial property of the measures in \eqref{eq:PL} is that they satisfy Luce's choice axiom. In particular, if $\{i_1, \ldots, i_k\}\subset \T_1\cap \T_2$, then $\P_{\u^*}\{i_1\succ\cdots \succ i_k\mid \T_1\} = \P_{\u^*}\{i_1\succ\cdots \succ i_k\mid \T_2\}$ \citep[Section 5]{MR2051012}. Since \eqref{eq:PL} is invariant when $\u^*$ is shifted by a common constant, an additional constraint is needed to ensure model identifiability. A common choice is $\bm 1^\top\u^* = 0$, and more general constraints of type $\bm a^\top\u^* = b$ can be used as long as $\bm a^\top\bm 1\neq 0$.  

To incorporate dynamic covariates in the PL model, we use a linear function to account for the influence of log scores affected by the covariates. Let $\bm v^*\in\R^d$ be a coefficient vector, where $d$ is assumed fixed and does not grow with $n$. For any $\T\in \EE$ and object $j\in \T$, let $X_{\T, j}\in\R^d$ be the covariates vector associated with $j$ on $\T$.  The log score of $j$ in $\T$ is given by
\begin{align*}
	&s_{j}(\t^*; \T) = u^*_j + X_{\T, j}^\top\bm\v^*&\bm\t^* = [(\u^*)^\top, (\bm v^*)^\top]^\top\in\R^{n+d}.
\end{align*}
With slight abuse of notation, the equivalent probability mass function of \eqref{eq:PL} based on $s_{j}(\t^*; \T)$ is given by 
\begin{align}
	\P_{\t^*}\left\{\pi = [\pi(1), \ldots, \pi(m)]\mid \T\right\} = \prod_{j\in [m]}\frac{\exp{s_{\pi(j)}(\t^*; \T)}}{\sum_{t=j}^m\exp{s_{\pi(t)}(\t^*; \T)}}.\label{eq:PLDC}
\end{align}
We refer to \eqref{eq:PLDC} as the {P}lackett--{LU}ce model {S}upplemented with {D}ynamic {C}ovariates (\textbf{PlusDC}), which will be studied in detail in the rest of the paper. We will drop the subscript $\t^*$ and $\T$ in $\P_{\t^*}\{\ \cdot \mid \T\}$ when no ambiguity arises.

The $\PLDC$ model encompasses many existing models as special instances. When $\EE\subseteq{\VV\choose 2}$ and $X_{\T, j}$ are independent of $\T$, the $\PLDC$ model reduces to the CARE model \citep{fan2022uncertainty}. Additionally, it covers other comparison models with global parameters, such as the BT model with home-field advantage \citep{MR3887567}, which we explain in the next example.

\begin{example}[BT model with home-field advantage]\label{BT_home}
	In the BT model with home-field advantage, {the probability of object $j$ beating object $k$ is given by}
	\begin{align*}
		\P\{j\succ k\} = \begin{cases}
			\frac{\vartheta \exp\{u_j^*\}}{\vartheta \exp\{u_j^*\} + \exp\{u_k^*\}}&\quad \text{if $j$ is at home}\\
			\frac{\exp\{u_j^*\}}{ \exp\{u_j^*\} + \vartheta \exp\{u_k^*\}}&\quad \text{if $k$ is at home} 
		\end{cases}
	\end{align*}
	where $\vartheta>0$ measures the strength of the home-field advantage or disadvantage.  Reparametrizing $v = \log\vartheta$ and for $\T = \{j, k\}$,  defining $X_{\T, j} = 1$ if $j$ is at home on $\T$ and $X_{\T, j} = 0$ otherwise, we can write the above probability mass function in the same form as the $\PLDC$ model. 
\end{example}

The home-field advantage model has long been used in sports analytics, but its theoretical foundation remains unclear. The major difficulty arises from the interaction between the global and local parameters, which complicates the analysis. This also suggests that the generality offered by the $\PLDC$ model also brings considerable challenges.

\subsection{Comparison graphs}

\subsubsection{Deterministic graph designs}

In the asymptotic analysis of the $\PLDC$ model, we will work with a sequence of comparison graphs $\HH_n(\VV_n, \EE_n)$. We assume $\HH_n(\VV_n, \EE_n)$ to be deterministic and satisfy certain topological conditions rather than being drawn from a random graph model. This generality allows for a better understanding of the statistical estimation procedure considered in Section~\ref{sec:est}. \checkedd{We will see later in Section \ref{sec:special case} that several commonly used random graph models do produce sequence configurations with the desired properties almost surely.} Moreover, $\HH_n$ are allowed to contain edges of varying sizes but in an asymptotically bounded manner, which we formulate as the following assumption.  
\begin{assumption}[Bounded comparison size]\label{ass:edgesize}
	$\mM:=\sup_{n\in\N}\max_{\T\in \EE_n}|\T|<\infty$. 
\end{assumption}
We need the following definitions to characterize the topological conditions in our analysis.

\begin{definition}[Modified Cheeger constant]\label{def:ch}
	Given a hypergraph $\HH(\VV, \EE)$ with $\VV = [n]$, for non-empty set $U \subset [n]$, let 
	$h_\HH(U) = {|\partial U|}/{\min\{|U|, |U^\complement|\}}.$
	The modified Cheeger constant of $\HH$ is defined as $h_\HH = \min_{U \subset [n]}h_\HH(U)$. In particular, a hypergraph $\HH(\VV, \EE)$ is called connected if and only if $h_\HH>0$.
\end{definition}

\begin{definition}[Weakly admissible sequences]\label{def:ad}
	Given $0<\lambda<1$ and a connected hypergraph $\HH(\VV, \EE)$ with $\EE \neq\emptyset$, a strictly increasing sequence of vertices $\{A_j\}_{j\in [J]}$ is called \emph{$\lambda$-weakly admissible} if 
	\begin{align}
		&\frac{|\{e\in \partial A_{j}: e\cap (A_{j+1}\setminus A_j)\neq\emptyset\}|}{|\partial A_j|}\geq\lambda& 1\leq j <J. \label{was}
	\end{align}
	Denote the set of $\lambda$-weakly admissible sequences of $\HH(\VV, \EE)$ as $\A_\HH(\lambda)$. The diameter of $\A_\HH(\lambda)$ is defined as $\diam(\A_\HH(\lambda))=\max_{\{A_j\}_{j=1}^J\in\A_\HH(\lambda)}J$.
\end{definition}

For any $\lambda$-weakly admissible sequence $\{A_j\}_{j\in [J]}$, a constant proportion of elements in $\partial A_j$ intersect $A_{j+1}$ for each $j<J$. Therefore, $|A_j|$ resembles the growth of the size of the $(j-1)$th graph neighborhood of $A_1$ and thus generalizes the notion of diameter of $\HH$.  \checkedd{For additional intuition and examples, please refer to Section F.2 of the supplementary file.}

\begin{figure}[t]
	\centering 
	\includegraphics[width=0.4\linewidth, trim={0cm 0cm 0cm 0cm},clip]{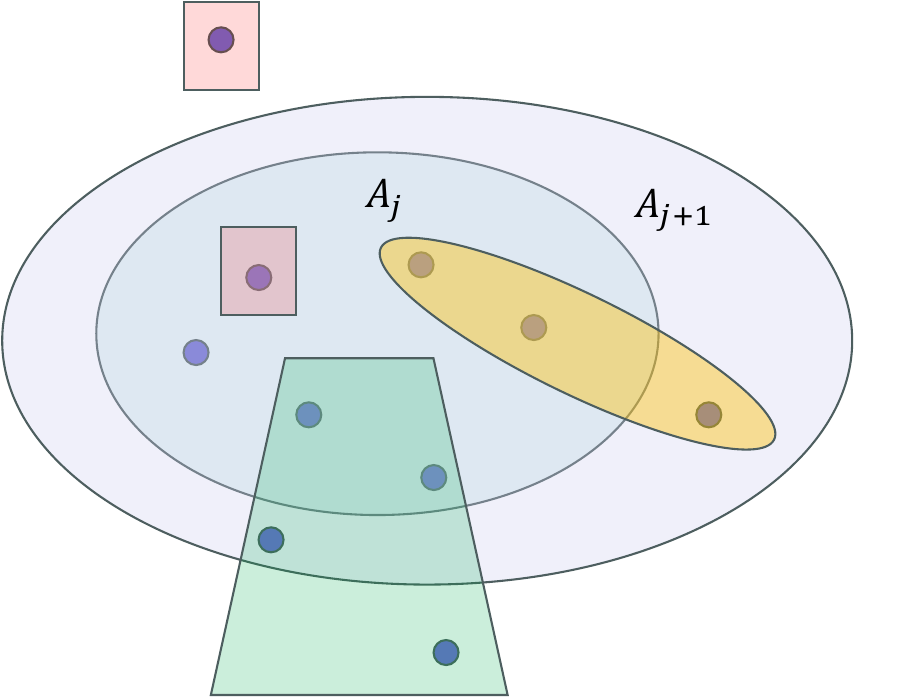}
	\caption{Illustration of different types of boundary edges in a nested sequence $\{A_j\}_{j \in [J]}$. The ellipsoid, the trapezoid, and the rectangle represent three different types of boundary edges of $A_j$: those fully contained in $A_{j+1}$, those intersecting both $A_{j+1}\setminus A_j$ and $A_{j+1}^\complement$, and those not intersecting $A_{j+1}\setminus A_j$, respectively.} \label{fig:ad}
\end{figure}

The definition of weakly admissible sequences generalizes the notion of admissible sequences used in the analysis of the PL model \citep{han2023unified}, which imposes a stronger condition in place of \eqref{was}\footnote{In their original formulation $\lambda = 1/2$ for simplicity.}: ${|\{e \in \partial A_{j}: e \subseteq A_{j+1}\}|} \geq \lambda{|\partial A_j|}.$ The two definitions conincide when $\EE\subseteq{\VV\choose 2}$ but differ in the hypergraph setting. For instance, in Figure~\ref{fig:ad}, while weakly admissible sequences require that the number of boundary edges of the rectangle type not exceed a certain threshold, admissible sequences require both the rectangle and trapezoid types. 

The topological conditions considered in Section~\ref{sec:cons} (when studying uniform consistency) aim to control the asymptotic growth rate of $\diam(\A_{\HH_n}(\lambda))$ for any fixed $\lambda$, which pertains to the asymptotic connectivity of the graph sequence $\HH_n$. Further details will be provided in the discussion in Section~\ref{sec:cons}.

\subsubsection{Random graph models}\label{sec:rg}

We introduce two random hypergraph models commonly used in the literature. The first is called the nonuniform random hypergraph model (NURHM). Let $\mM$ be the same integer in Assumption~\ref{ass:edgesize}. The edge $\EE$ in a NURHM model is defined as $\EE = \sqcup_{\muu=2}^{\mM}\EE^{(\muu)}, \EE^{(\muu)}\subseteq {\V\choose m},$
where $\EE^{(\muu)}$ are independent $\muu$-uniform random hypergraphs generated as follows.
Fixing $\muu$, $\mathbb I_{\{e\in \EE^{(\muu)}\}}$ are independent Bernoulli random variables with parameter $p_{e, n}^{(\muu)}$ for $e\in{\V\choose m}$. 
Denote the lower and upper bound of $\muu$-edge probabilities as 
$p_n^{(\muu)} = \min_{e\in{\V\choose m}}p_{e, n}^{(\muu)}$ and $q_n^{(\muu)} = \max_{e\in{\V\choose m}}p_{e, n}^{(\muu)},$
and the corresponding order of the expected number of edges as
\begin{align}
	&\xi_{n,-}:=\sum_{\muu=2}^\mM n^{\muu-1}p_n^{(\muu)} & \xi_{n,+}:=\sum_{\muu=2}^\mM n^{\muu-1}q_n^{(\muu)}.\label{myxis}
\end{align}
When $p_n^{(\muu)} = q_n^{(\muu)} = p_n^{(\mM)}\mathbb I_{\{\muu=\mM\}}$, NURHM reduces to the $\mM$-way \ER model \citep{MR0125031}. 

The second model is the hypergraph stochastic block model (HSBM) which produces cluster structures. For simplicity, we will focus on an $\mM$-uniform HSBM model with $L$ clusters partitions $\VV$ into $L$ subsets, $\VV = \sqcup_{\ell\in [L]} \VV_\ell$. Here, $L$ is an absolute constant and edges within and across partitioned sets have different probabilities of occurrence: $\P\{e\in \EE\} = \omega_{n, \ell}\mathbb I_{\{e\in{\V\choose M}\}}$ if $e\subseteq \VV_\ell$ for $\ell \in [L]$ and $\P\{e\in \EE\} = \omega_{n, 0}\mathbb I_{\{e\in{\V\choose M}\}}$ otherwise. 
Similar to before, we define the minimal order of the expected number of edges in an HSBM as 
\begin{align}
	\zeta_{n,-}:=n^{\mM-1}\min_{0\leq \ell \leq L}\omega_{n, \ell}.\label{myzeta}
\end{align}
We do not introduce the notation for the upper bound as it is not needed in the subsequent analysis. 
HSBM can be used to simulate heterogeneous networks by allowing $\max_{\ell}\omega_{n, \ell}/\min_{\ell}\omega_{n, \ell}$ to diverge as $n\to\infty$. 

\section{Model identifiability}\label{sec:ident}

We discuss the model identifiability of the $\PLDC$ model. To be self-contained, we define model identifiability first. 

\begin{definition}[Model identifiability]
	Given a comparison graph $\HH(\VV, \EE)$ and covariates $\mathcal{X} =\{X_{\T, j}\}_{j\in \T, \T\in \EE}$, a $\PLDC$ model is said to be identifiable with respect to $(\HH, \mathcal X)$ if $\P_{\t_1}\{\ \cdot  \mid \T\} = \P_{\t_2}\{\ \cdot \mid \T\}$ for all $\T\in\EE$ implies $\t_1 = \t_2$, where $\P_\t\{\ \cdot \mid \T\}$ is defined in \eqref{eq:PLDC}. In other words, the mapping $\t\mapsto\{\P_\t\{\ \cdot \mid \T\}\}_{\T\in \EE}$ is an injection.  \end{definition}

This definition can be applied to the PL model by removing $\mathcal X$ and $\v$, and we will use it without further clarification.

The $\PLDC$ is not identifiable with respect to any $(\HH, \mathcal X)$ without imposing any constraints since $\P_{\t_1}\{\ \cdot \mid \T\} = \P_{\t_2}\{\ \cdot \mid \T\}$ if $\t_1$ and $\t_2$ differ by a multiple constant times $\bm 1$ on the $\u$ component. Similar issues also arise in the PL model. In the latter, additional constraints (e.g., $\bm 1^\top\u = 0$) are often imposed to ensure model identifiability.

In the $\PLDC$ model, however, it is not clear whether a single constraint like $\bm 1^\top\u = 0$ is sufficient to ensure model identifiability. For example, in the CARE model (viewed as a special instance of the $\PLDC$ model), $X_{\T, j}=X_j$ are independent of $\T$. In such situations, the model is not identifiable under $\bm 1^\top\u = 0$, and extra constraints are needed \cite[Proposition 1]{fan2022uncertainty}. We will elaborate on this in Section~\ref{sec:decomp}. Nevertheless, if we allow $\mathcal X$ to be edge-dependent, then a single constraint like $\bm 1^\top\u = 0$ is often sufficient for model identifiability subject to mild non-degeneracy conditions.

To state the result, we need a few more notations. Recall the ordered tuple representation of each edge $\T_i = (j_{i1}, \ldots, j_{im_i})$ with $1\leq j_{i1}<\cdots<j_{im_i}\leq n$. This fixes the orientation of the edge, which is independent of the comparison outcomes. Define
\begin{align}\label{myK}
	\K = \left[\Delta X_1^\top, \cdots, \Delta X_N^\top\right]^\top \in\R^{\M\times d}\quad\quad \M = \bigg(\sum_{i\in [N]}m_i\bigg)-N,
\end{align}
where
\begin{align}\label{DeltaX}
	\Delta X_{i}:=
	\left[X_{\T_i, j_{i2}}-X_{\T_i, j_{i1}},
	\cdots,
	X_{\T_i, j_{i(m_i-1)}}-X_{\T_i, j_{i1}}\right]^\top
	\in\R^{(m_i-1)\times d}, \quad  i\in [N].
\end{align}
\checkedd{Since the $\PLDC$ model depends only on the relative log-scores of the objects, $\Delta X_{i}$ is a minimal breaking (a spanning tree) that preserves the relative covariates information on the hyperedge $\T_i$. Such a breaking also induces a pairwise comparison graph $\HH_\br = (\VV, \EE_\br)$, where $\EE_\br = \{(j_{i1}, j_{it}), t\in [m_i], i\in [N]\}$ consists of all ordered pairwise tuples resulting from breaking $\T_i$ apart.} For ease of illustration, we show how an ordered tuple $(j_{i1}, j_{i2}, j_{i3}, j_{i4})$ is broken in Figure \ref{fig:br}. Note that $\HH_\br$ may not be simple and may contain multi-edges. Denote $\Q\in\R^{n\times \M}$ as the incidence matrix of $\HH_\br$ under the fixed orientation. 

\begin{figure}[t]
	\centering 
	\includegraphics[width=0.5\linewidth, trim={0cm 0cm 0cm 0cm},clip]{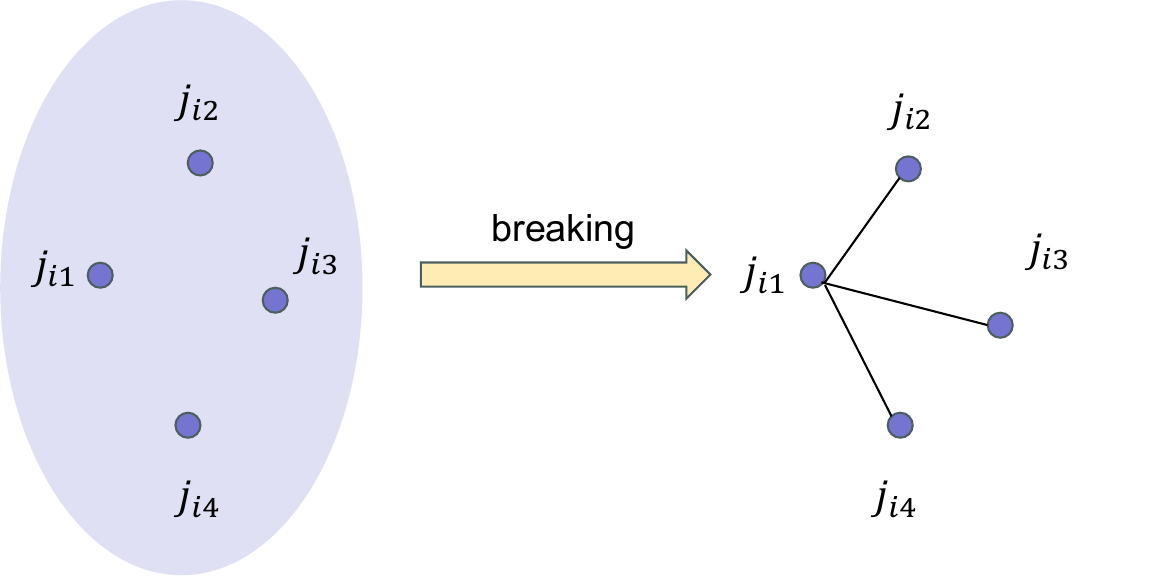}
	\caption{Illustration of edges in $\HH_\br$ obtained from breaking a hyperedge in $\HH$, where a hyperedge $(j_{i1}, j_{i2}, j_{i3}, j_{i4})$, $j_{i1}< j_{i2}< j_{i3}< j_{i4}$ is broken into three edges $(j_{i1}, j_{i2})$, $(j_{i1}, j_{i3})$, and $(j_{i1}, j_{i4})$ in $\HH_\br$.} \label{fig:br}
\end{figure}

\begin{proposition}[Identifiability of $\PLDC$]\label{prop:1}
	Let $\HH = (\VV, \EE)$ be a comparison graph with $\EE=\{\T_i\}_{i\in [N]}$.
	Assuming $\bm 1^\top\u=0$, the $\PLDC$ model is identifiable with respect to $(\HH, \mathcal X)$ if and only if the matrix $\AA=[\Q^\top, \K]\in\R^{\M\times (n+d)}$ satisfies $\rank(\AA) = n+d-1$. 
\end{proposition}

Since $\rank(\Q^\top)\leq n-1$ ($\Q^\top \bm 1 = \bm 0$), the condition $\rank(\AA) = n+d-1$ can be interpreted as requiring: (I) $\rank(\Q^\top) = n-1$; (II) $\rank(\K) = d$; and (III) $\range(\Q^\top)\cap\range(\K) = \{\bm 0\}$. One can verify that (I) holds if and only if $\HH$ (or $\HH_\br$) is connected; (II) holds if and only if $\K$ has full column rank. The last condition can be parsed through the lens of the non-existence of curl-free flows in $\range(\K)$.

\checkedd{To illustrate the geometric intuition behind the last point, we adopt a Hodge decomposition perspective in \cite{jiang2011statistical}; see also \cite{lim2020hodge} for a review of related concepts.} For convenience, we assume $\HH_\br$ is simple. Let $L_\wedge^2(\EE_\br)$ and $L_\wedge^2(\EE_\triangle)$ be the space of alternating functions defined on the edges and triangles in $\HH_\br$, respectively. 
The curl operator $\curl: L_\wedge^2(\EE_\br)\to L_\wedge^2(\EE_\triangle)$ is defined as 
\begin{align}
	\curl(f)(i, j, k) = f(i, j) + f(j, k) + f(k, i),\label{curl}
\end{align}  
where $(i, j, k)$ is a triangle in $\HH_\br$ (that is, $\{i, j\}, \{j, k\}, \{i, k\}\in \EE_\br$). 
Identifying $\range(\Q^\top)$ as a subspace of $L_\wedge^2(\EE_\br)$, $\range(\Q^\top)\subseteq\ker(\curl)$.
Therefore, a sufficient condition for $\range(\Q^\top)\cap\range(\K) = \{\bm 0\}$ is $\ker(\curl)\cap\range(\K) = \{\bm 0\}$, which can be checked by calculating net flows in $\range(\K)$ on the triangles in $\HH_\br$. This yields a sufficient condition for the identifiability of the $\PLDC$ model as follows. 

\begin{proposition}\label{prop:curl}
	Let $f_1, \ldots, f_d\in L^2_\wedge(\EE_\br)$ be the columns of $\K$. 
	If $\HH$ is connected and has $d$ triangles $\mathcal T_{\triangle}=\{(j_{\ell 1}, j_{\ell 2}, j_{\ell 3})\}_{\ell\in [d]}$ such that $\bm T_{\triangle} = (\curl(f_k)(j_{\ell 1}, j_{\ell 2}, j_{\ell 3}))_{\ell, k\in [d]}\in\R^{d\times d}$ has full rank, then $\rank(\AA) = n+d-1$. 
\end{proposition}

\begin{remark}\label{rem:curl}
	\checkedd{Proposition~\ref{prop:curl} offers additional insights into how dynamic covariates contribute to model identifiability. Let $d=1$ and consider a complete pairwise comparison scenario of three objects $1, 2$, and $3$. In this case, $\bm T_{\triangle} = (X_{\{1,2\}, 2}-X_{\{1,2\}, 1}) + (X_{\{2,3\}, 3}-X_{\{2,3\}, 2}) + (X_{\{1,3\}, 1}-X_{\{1,3\}, 3})$, which has full rank if and only if it is nonzero. This holds when the covariates $X_{\{i,j\}, i}$ are sufficiently random. For instance, when $X_{\{i,j\}, i}$ exhibit cyclic behavior, such as $X_{\{1,2\}, 2}>X_{\{1,2\}, 1}, X_{\{2,3\}, 3}>X_{\{2,3\}, 2}$, and $X_{\{1,3\}, 1}>X_{\{1,3\}, 3}$ (i.e., the relative covariate information both helps and hinders in the comparisons an object participates, which differs from the cyclic comparison outcomes \citep{jiang2011statistical, singh2024analysis}). Such fluctuations prevent the covariates' effects from being fully absorbed into the utility parameters. On the other hand, for static covariates, $\bm T_{\triangle}=0$.}
\end{remark} 

\color{black}

\section{Parameter estimation}\label{sec:est}

We now consider estimating the model parameter $\t^*$ using the MLE. We will provide a sufficient and necessary condition for the unique existence of the MLE in the $\PLDC$ model. Moreover, when the covariates in the $\PLDC$ model are edge-independent, we show that the computation of the MLE is equivalent to that in the original PL model plus an additional step of decomposition. This justifies the necessity to incorporate the dynamic part of the covariates in order to improve model prediction. 

\subsection{Maximum likelihood estimation}
{Given a comparison graph $\HHn(\VVn, \EEn)$, edge-dependent covariates $\mathcal{X} =\{X_{\T, j}\}_{j\in \T, \T\in \EEn}$, and independent comparison outcomes $\{\pi_i\}_{i\in [N]}$ with $\pi_i\in\mathcal S(\T_i)$, $|\T_i| = m_i$, the log-likelihood $l$ at $\t = (\u^\top, \v^\top)^\top$ can be written as
\begin{align}
	&l(\t) = \sum_{i\in [N]}\sum_{j\in [m_i]}\left[s_{i j}-\log(\sum_{t=j}^{m_i}\exp{s_{it}})\right] \quad &s_{ij}=s_{\pi_i(j)}(\t; \T_i) = u_{\pi_i(j)} + X_{\T_i, \pi_i(j)}^\top\bm\v,\label{likelihood}
\end{align}
where $\pi_i(j)$ denotes the object in $\T_i$ with rank $j$ in the comparison outcome $\pi_i$. 
Under the model identifiability constraint $\bm 1^\top\u = 0$, we define the MLE $\tt$ for $\bm\t^*$ as 
\begin{align}
	\tt = (\widehat{\bm u}^\top, \widehat{\bm v}^\top)^\top = \argmax_{\t = (\u^\top, \v^\top)^\top: \bm 1^\top\bm u = 0}l(\t).\label{def:mle}
\end{align}
Although $\bm 1^\top\bm u = 0$ alone ensures model identifiability at the population level, it is unclear whether $\tt$ exists. We next provide a sufficient and necessary condition under which the MLE defined in \eqref{def:mle} uniquely exists.  
\begin{customthm}{2$^\dagger$}\label{ass:ue}
	Given $\HHn = (\VVn, \EEn)$, $\mathcal X = \{X_{\T, j}\}_{j\in \T, \T\in\EEn}$, and the comparison outcomes $\{\pi_\T\}_{\T\in\EEn}$, for any nonzero $\t=(\u^\top,\v^\top)^\top \in\R^{n+d}$ with $\bm{1}^\top \u = 0$, there exist objects $k_1, k_2\in\V$ and $\T\in\EEn$ such that $k_1\succ k_2$ on $\T$ and $u_{k_1}+X^\top_{\T,k_1}\v<u_{k_2}+X^\top_{\T,k_2}\v$. 
\end{customthm}

Assumption \ref{ass:ue} may appear exotic at first sight. Without $\v$, Assumption~\ref{ass:ue} is equivalent to the sufficient and necessary condition for the unique existence of the MLE in the PL model \citep{MR0097876, MR2051012}. The additional requirement on $\v$ is weaker than \cite[Assumption 2]{MR2051012} when specialized to the home-field advantage model and holds for general covariate types and dimensions. \checkedd{To the best of our knowledge, Assumption \ref{ass:ue} is the first necessary and sufficient condition to ensure the unique existence of the MLE in the PL model with dynamic covariates. In other covariate-assisted pairwise cardinal models \citep{singh2024graphical}, the MLE coincides with the least-squares estimator, which has a closed-form solution, so only uniqueness needs to be verified to ensure the estimator's well-posedness.} 
	
\begin{theorem}[Unique existence of the MLE]\label{thm:existence}
	The MLE $\tt$ defined in \eqref{def:mle} uniquely exists if and only if Assumption~\ref{ass:ue} holds. In particular, under Assumption~\ref{ass:ue}, the log-likelihood function $l(\t)$ is strictly concave on the constraint set $\bm 1^\top\u = 0$. 
\end{theorem}
	
Unlike the identiability condition in Proposition~\ref{prop:1}, Assumption~\ref{ass:ue} depends also on the comparison outcomes $\{\pi_\T\}_{\T\in\EEn}$. Since we cannot impose assumptions on $\{\pi_\T\}_{\T\in\EEn}$ when discussing the uniform consistency of the MLE, we restate Assumption~\ref{ass:ue} as an alternative condition called \emph{approximate existence profile} (AEP) that depends on $\HH$ and $\mathcal X = \{X_{\T, j}\}_{j\in \T, \T\in\EEn}$ only. The trade-off is that we can only require Assumption~\ref{ass:ue} to hold with high probability. 
	
\begin{assumption}[AEP]\label{ass:aep}
	The hypergraph $\HHn = (\VVn, \EEn)$ (with $|\VVn| = n$) and the corresponding covariates $\mathcal X = \{X_{\T, j}\}_{j\in \T, \T\in\EEn}$ are an AEP, that is, with probability at least $1-n^{-2}$, the comparison outcome $\{\pi_\T\}_{\T\in\EEn}$ satisfies Assumption~\ref{ass:ue} with respect to $\HHn$ and $\mathcal X$. 
\end{assumption}
	
The probability $1-n^{-2}$ is not essential and can be replaced with any sequence $1-\delta_n$ with $\sum_{n\in\N}\delta_n<\infty$ (since we want to apply the Borel--Cantelli lemma later). To show the definition of AEP in Assumption~\ref{ass:aep} is not vacuous, we provide an explicit way to construct such profiles when $\Delta X_i$'s are random and well-balanced and $\HH$ is sufficiently connected. 
	
\begin{customthm}{2$^\ddagger$}\label{ass:avoidpa}
	$\{\Delta X_i\}_{i\in [N]}$ are independent random matrices and there exists $\delta>0$ such that
	\begin{align}
		c_0:=\liminf_{n\to\infty}\min_{i\in [N]}\inf_{\|\v\|_\infty=1}\min_{j\in [m_i-1]}\P\{\Delta X_i[j,:]\v>\delta\}>0, \label{mass}
	\end{align}
	where $\Delta X_i[j,:]$ denotes the $j$th row of $\Delta X_i$. 
\end{customthm}
Assumption~\ref{ass:avoidpa} is a balancing condition on the covariates only; it requires for every $\v$, $\Delta X_i[j,:]\v$ does not over-concentrate on either side of individuals; a further discussion on Assumption~\ref{ass:avoidpa} is given in Proposition~\ref{prop:1+}. 
Moreover, we need the following uniform boundedness assumption, which will be used in the asymptotic results in the remainder of the article.
	
\begin{assumption}\label{ass:unif-bdd}
	The true parameter $\t^*$ and the covariates $X_{\T_i, j}$ are uniformly bounded:
	\begin{align*}
		R:=\sup_{n\in\N}\max\left\{\|\t^*\|_\infty, \max_{i\in [N], j\in \T_i}\|X_{\T_i, j}\|_1\right\}<\infty.
	\end{align*}
\end{assumption}
	
\begin{proposition}\label{prop:at_verify}	
	Under Assumptions~\ref{ass:avoidpa} and \ref{ass:unif-bdd}, if the modified Cheeger constant of $\HH_n$ satisfies ${h_{\HH_n}}/{\log n}\to\infty$, then a.s., $\HH_n$ and $\mathcal X = \{X_{\T, j}\}_{j\in \T, \T\in\EEn}$ are an AEP for all large $n$. 
\end{proposition}

\subsection{A space decomposition perspective on fixed covariates}\label{sec:decomp}
	
Before turning to discuss the computational aspects of \eqref{def:mle}, we observe that the computation of \eqref{def:mle} can be simplified when $X_j=X_{\T, j}$ are independent of $\T$, reducing to the CARE model. As briefly mentioned in Section~\ref{sec:ident}, the CARE model is generally not identifiable and requires additional constraints to ensure a unique solution. \checkedd{To understand why this occurs, observe that for any parameter $\t = (\u^\top, \v^\top)^\top$ with $\v\neq\bm 0$ in the CARE model, the log-likelihood remains unchanged by letting $\bm u' = \u  + [X_1, \ldots, X_n]^\top\v - n^{-1}\bm 1\bm 1^\top[X_1, \ldots, X_n]^\top\v$ and $\v' = \bm 0$. This shows that the part of $\v$ can be perfectly absorbed into $\u$ by an affine transformation.} Consequently, for the CARE model, the optimum value of the log-likelihood function is the same as the PL model. Hence, adding fixed covariates does not enhance model fitting. Proposition~\ref{prop:equiv} formalizes this observation. 
	
\begin{proposition}[Equivalence between PL and CARE]\label{prop:equiv}
	Let $X_j\in\R^d$ be the covariates for the $j$th item and $\bm Z = [X_1, \ldots, X_n]^\top\in\R^{n\times d}$.  
	Let $l(\t)$ be the log-likelihood function defined in \eqref{likelihood}.
	Consider the MLEs of the PL model and the CARE model:
\begin{align*}
	\widetilde{\u}, \bm 0 &= \argmax_{\t=(\u^\top, \bm 0^\top)^\top: \bm 1^\top\u = 0}l(\t)&\mathrm{(PL\  model)},\\
	\uu, \vv &= \argmax_{\t=(\u^\top, \v^\top)^\top: \bm 1^\top\u = 0, \bm Z^\top\u = \bm 0}l(\t)&\mathrm{(CARE \ model).}
\end{align*}
Suppose that $\rank(\bm Z) = d$ and $\bm 1\not\in\range(\bm Z)$. 
If the MLE in either model uniquely exists, then they satisfy the following relations:
\begin{align*}
	\vv = \left(\bm Z^\top\bm Z - \frac{\bm Z^\top\bm 1\bm 1^\top\bm Z}{n}\right)^{-1}\bm Z^\top\widetilde{\u}, \ \ \ 	\uu = \widetilde{\u} - \left(\bm I - \frac{\bm 1\bm 1^\top}{n}\right)\bm Z \vv.
\end{align*}
In particular, $l(\widetilde{\u}, \bm 0) = l(\uu, \vv)$. 
\end{proposition}
	
According to Proposition \ref{prop:equiv}, the PL and CARE models share the same maximum likelihood while the latter contains more parameters. Therefore, the CARE model cannot achieve a better score under frequently used model selection metrics such as the Akaike information criterion (AIC) or Bayesian information criterion (BIC) than the PL model. In contrast, the PlusDC model often improves the AIC or BIC over the PL model. This underscores both the necessity and sufficiency of considering dynamic covariates.
	
\begin{remark}
	Proposition \ref{prop:equiv} offers an intuitive interpretation: a player's utility to win consists of both internal skills and external influences. Fixed covariates, unchanged across games, reflect individual traits and can be absorbed into internal skills. In contrast, dynamic covariates capture varying external impacts, thereby improving prediction. 
\end{remark}
	
\subsection{Alternating maximization algorithm}\label{sec:am}
	
Since $l(\t)$ is a concave function of $\t$, the MLE can be found by applying standard convex optimization methods. Due to the asymmetric roles of $\u$ and $\v$, we consider an alternating maximization (AM) strategy to optimize over $\u$ and $\v$, respectively. Our approach is general and can be extended to settings where the score $s_{ij}$ has a nonlinear dependence on $\v$. The outline of the algorithm is given in Algorithm~\ref{alg:1}. 
	
\begin{algorithm}
	\hspace*{\algorithmicindent} \textbf{Input}: initial guess $\bm\t^{(0)} = (\bm u^{(0)}, \bm v^{(0)})$, tolerance $\e>0$, iterate $\tauf =0$\\
	\hspace*{\algorithmicindent} \textbf{Output}: the joint MLE $\tt$
	\begin{algorithmic}[1]
		\WHILE{not converge}
		\STATE $\bm u^{(\tauf+1)}\gets \argmax_{\bm u: \bm 1^\top \bm u = 0}l(\bm u, \bm v^{(\tauf)})$
		\STATE $\bm v^{(\tauf+1)}\gets \argmax_{\bm v}l(\bm u^{(\tauf+1)}, \bm v)$
		\STATE $\bm\t^{(\tauf+1)} \gets (\bm u^{(\tauf+1)}, \bm v^{(\tauf+1)})$
		\IF{$\frac{1}{N}(l(\t^{(\tauf+1)})-l(\t^{(\tauf)}))\leq \e$}{
			\STATE $\bm\t = \bm\t^{(\tauf+1)}$
			\STATE \textbf{break}}
		\ELSE{
			\STATE $\tauf\gets \tauf+1$}
		\ENDIF
		\ENDWHILE
	\end{algorithmic}
	\caption{Alternating maximization for finding the joint MLE} 
	\label{alg:1}
\end{algorithm}

\checkedd{
	Steps 2 and 3 in Algorithm~\ref{alg:1} are solved using a Minorize--Maximization (MM) algorithm and a Newton--Raphson procedure, respectively. The implementation and analysis of these procedures are standard. When the MLE exists uniquely, all iterates of the AM algorithm remain within a uniformly bounded domain. By Theorem~\ref{thm:existence} and a continuity argument, the normalized log-likelihood function is strongly concave on a common compact domain containing all iterates. Given the smoothness of the function, existing results \citep{luo1993error, beck2015convergence} can be applied to establish the linear convergence of the algorithm. Under suitable assumptions, the total computational complexity of the AM algorithm is $\mathcal O(N\mathrm{poly}(\log (1/\e)))$. More details are given in the supplementary file.} 
\section{Uniform consistency of the MLE}\label{sec:cons}

In this section, we establish the uniform consistency of the MLE in the $\PLDC$ model. For ease of understanding, we add the subscript $n$ in $\HH(\VV, \EE)$ to emphasize the dependence on $n$. Specifically, 
let $\HH_n = (\VV_n, \EE_n)$ be a comparison graph with $\EE_n=\{\T_i\}_{i\in [N]}$. Matrix $\bm Q_n\in\R^{n\times \M}$ denotes the incidence matrix of the induced graph $(\HH_n)_\br$ of $\HH_n$ under the prefixed orientation, and $\bm K_n\in\R^{\M\times d}$ denotes the covariate differences between objects on each edge. To establish the uniform consistency of the MLE, we need additional assumptions on the incoherence between $\range(\bm Q_n^\top)$ and $\range(\bm K_n)$ and non-degeneracy of $\bm K_n$, as well as appropriate topological conditions on the comparison graph sequence $\HH_n$.

\begin{assumption}[Incoherence]\label{ass:cosine}
	There exists $\e>0$ such that
	\begin{align*}
		\sup_{n\in\N}\sup_{\substack{\bm x\in\range(\bm Q_n^\top),  \bm y\in\range(\bm K_n)\\ \|\bm x\|_2=\|\bm y\|_2 = 1}}\bm x^\top \bm y<1-\e.
	\end{align*}
\end{assumption}

\begin{assumption}[Nondegeneracy]\label{ass:minK}
	Denote the smallest singular value of $\bm K_n$ as $\sigma_{\min}(\bm K_n)$. We assume that $\liminf_{n\to\infty}\sigma_{\min}(\bm K_n)/{\sqrt{\M}} >c > 0$ for some positive constant $c>0$. 
\end{assumption}

\checkedd{Assumption~\ref{ass:cosine} can be viewed as an asymptotic generalization of the model identifiability condition in Proposition~\ref{prop:1} for a sequence of models. Indeed, for fixed $n$, the model identifiability condition holds if and only if (I) $\rank(\Q^\top) = n-1$; (II) $\rank(\K) = d$; and (III) $\range(\Q^\top)\cap\range(\K) = \{\bm 0\}$. Condition (III) is equivalent to the existence of $\e_n>0$ such that 
	\begin{align*}
		\sup_{\substack{\bm x\in\range(\bm Q_n^\top),  \bm y\in\range(\bm K_n)\\ \|\bm x\|_2=\|\bm y\|_2 = 1}}\bm x^\top \bm y<1-\e_n.
	\end{align*}
		Assumption~\ref{ass:cosine} further requires $\e_n$ to be uniformly bounded away from $0$ so the sequence of models is asymptotically identifiable.} Assumption~\ref{ass:minK} is standard in the statistics literature. Both assumptions hold when sufficient randomness of $\bm K_n$ exists (Proposition~\ref{prop:dl}).

\begin{assumption}[Diameter of weakly admissible sequences]\label{ass:topology}
	For any fixed $\lambda>0$, the diameter of the $\lambda$-weakly admissible sequences in $\HH_n$ satisfies
	\begin{align*}
		\lim_{n\to\infty}\diam(\A_{\HH_n}(\lambda))\left(\frac{\log n}{h_{\HH_n}}\right)^{1/4}=0, 
	\end{align*}
	where $h_{\HH_n}$ is the modified Cheeger constant of $\HH_n$. 
\end{assumption}

Assumption~\ref{ass:topology} concerns the asymptotic connectivity of $\HH_n$. The required condition implies that $h_{\HH_n}$ is divergent, which can be roughly understood as a lower bound on the minimum degree of a graph. The controlled growth rate of $\diam(\A_{\HH_n}(\lambda))$ suggests that the neighborhood of an arbitrary vertex set should be reasonably large since a constant proportion of the boundary edges in $\partial A_j$ in a weakly admissible sequence intersects $A_{j+1}$. 

\begin{theorem}[Uniform consistency of the MLE]\label{thm:main}
	Under Assumptions~\ref{ass:edgesize}-\ref{ass:topology}, a.s., the MLE uniquely exists for all large $n$ with its components $\uu$ and $\vv$ satisfying
	\begin{equation}\label{rates}\small
		\|\uu-\u^*\|_\infty  \lesssim\diam(\mathcal A_{\HH_n}(\lambda))^2\sqrt{\frac{\log n}{h_{\HH_n}}}=o(1), \quad 
		\|\vv-\v^*\|_\infty  \lesssim \diam(\mathcal A_{\HH_n}(\lambda))\sqrt{\frac{\log n}{h_{\HH_n}}}=o(1),
	\end{equation}
	where $\lambda>0$ depends only on $\mM$ and $R$ in Assumptions~\ref{ass:edgesize} and \ref{ass:unif-bdd}. Consequently, $\tt$ is uniformly consistent, that is, $\|\tt-\t^*\|_\infty \to 0 $ as $n \to \infty$ a.s.  
\end{theorem}

Under Assumptions~\ref{ass:edgesize}-\ref{ass:topology}, the convergence rates for $\uu$ and $\vv$ are \emph{different} and characterized by the parameters $h_{\HH_n}$ and $\diam(\mathcal A_{\HH_n}(\lambda))$. In Section~\ref{sec:special case}, we will provide explicit upper bounds on both parameters when the graph sequence is drawn from appropriate random graph models.

\begin{figure}[t]
	\centering
	\begin{minipage}{0.442\linewidth}
		\centering
		\includegraphics[width=\linewidth, trim={0cm 0cm 0cm 0cm},clip]{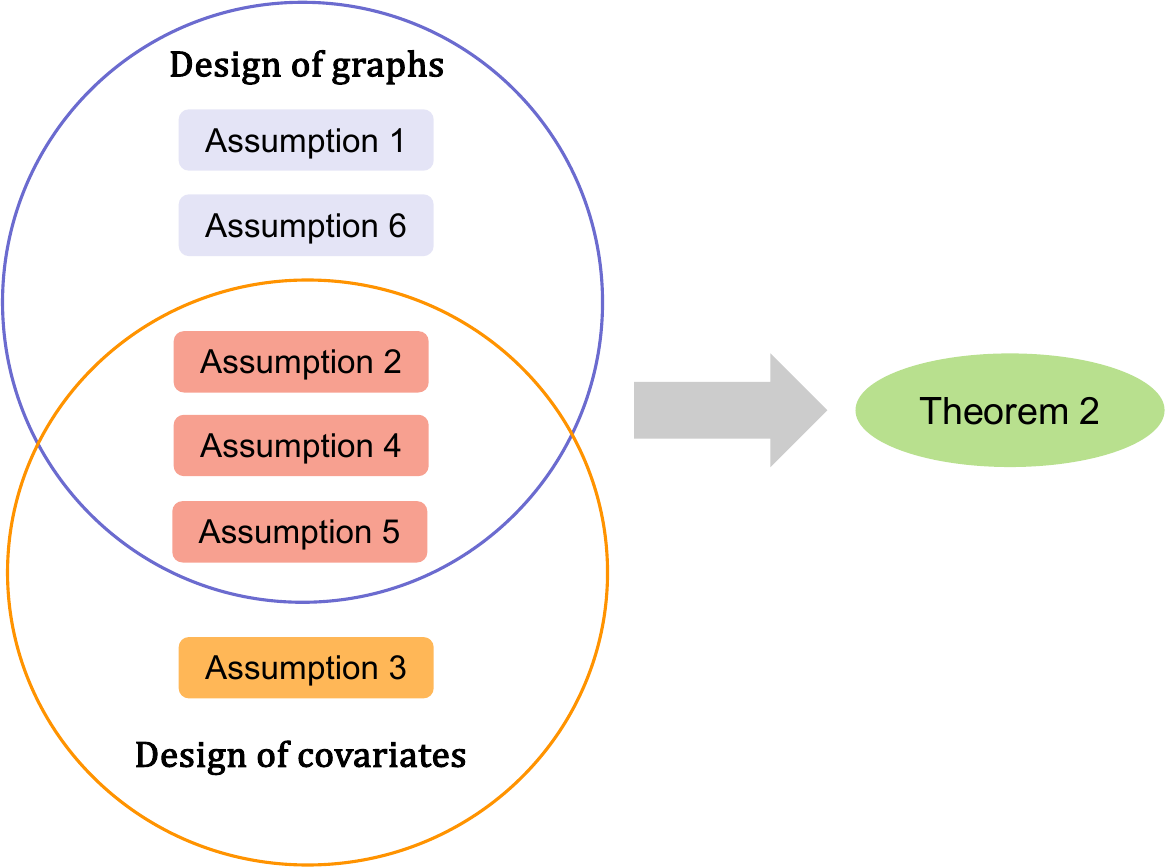}
		\caption*{(a)}\label{aaa1}
	\end{minipage}
	\hfill
	\begin{minipage}{0.55\linewidth}
		\centering
		\includegraphics[width=\linewidth, trim={0cm 0cm 0cm 0cm},clip]{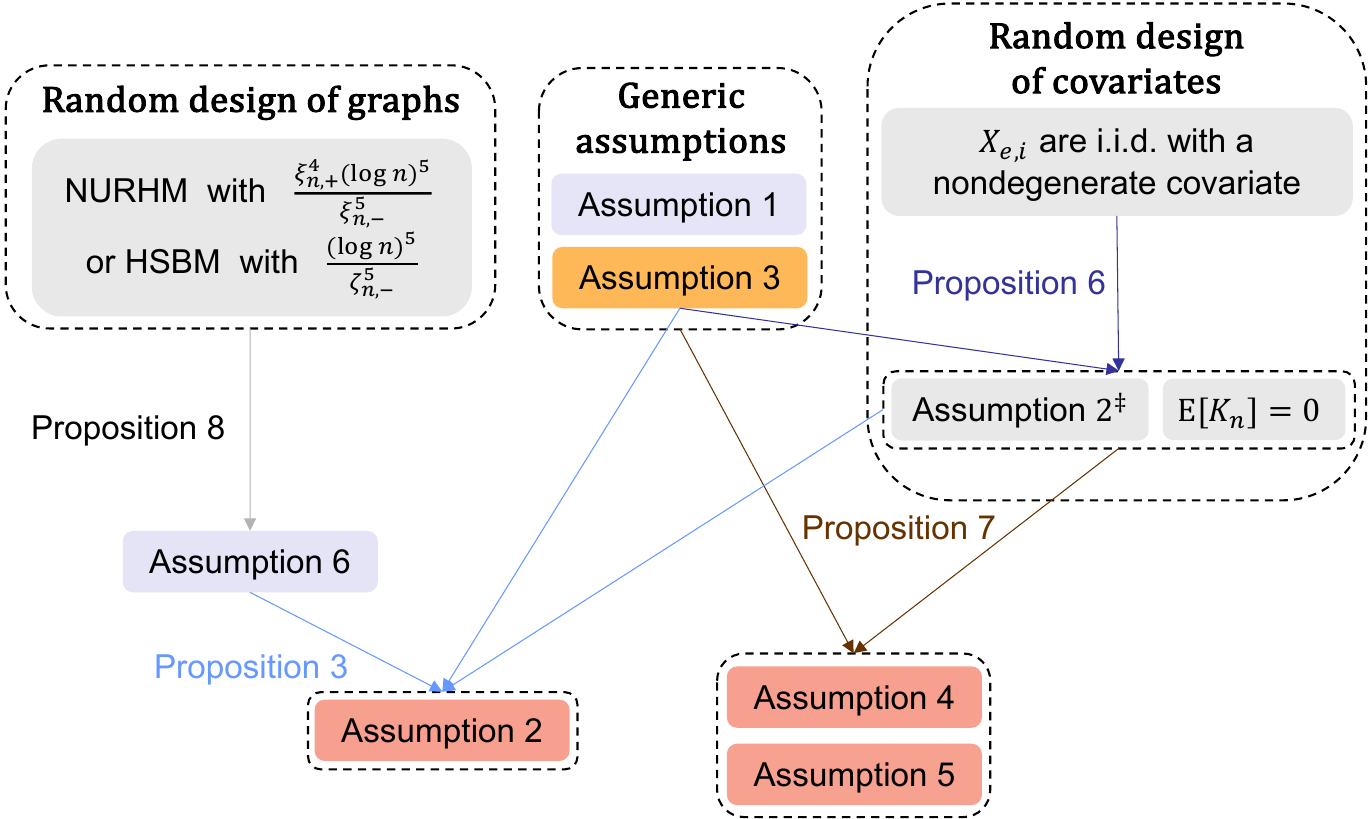}
		\caption*{(b)}\label{bbb1}
	\end{minipage}
	\caption{(a): Connections between various assumptions and the uniform consistency result in the deterministic design. (b): Random designs of graphs and covariates combined with generic assumptions imply the other assumptions needed for the uniform consistency result.} \label{fig:ass}
\end{figure}

\section{Uniform consistency under random designs}\label{sec:special case}

Theorem~\ref{thm:main} established a uniform consistency result for the MLE under deterministic assumptions of the design of both comparison graphs and covariates (see Figure~\ref{fig:ass}~(a)). It is not yet clear whether the required assumptions can hold simultaneously. This section addresses this question by considering random design models. For the comparison graphs, we consider the random hypergraph models in Section~\ref{sec:rg}, including both NURHM and HSBM. For the covariates, we leverage Assumption~\ref{ass:avoidpa}, which is a random design assumption on the covariates. A further specification of Assumption~\ref{ass:avoidpa} is given in Proposition~\ref{prop:1+}. 

\begin{proposition}\label{prop:1+}
	Suppose there exists an absolute constant $\delta'>0$ such that, for each $n$ and $i \in [N]$, the conditional distribution of each row of $\Delta X_i$ on $\HH_n$ has mean zero and covariance lower bounded by $\delta'^2 \bm I$ in the Loewner order. Under Assumption~\ref{ass:unif-bdd}, Assumption~\ref{ass:avoidpa} holds. 
\end{proposition}

\begin{remark}\label{rem:3}
	The conditions in Proposition~\ref{prop:1+} hold if $X_{\T, i}$ are i.i.d. random vectors independent of $\HH_n$ and $n$ with a nondegenerate covariance matrix. While more interesting applications may involve varying distributions of $X_{\T,i}$ depending on $\T$ and $i$, this simple case is sufficient to justify the random design considered in Assumption~\ref{ass:avoidpa} exists.  
\end{remark}

Under the random design assumption of both comparison graphs and covariates as well as Assumptions~\ref{ass:edgesize} and \ref{ass:unif-bdd}, the following results hold.

\begin{proposition}\label{prop:dl}
	Under Assumptions~\ref{ass:edgesize} and \ref{ass:unif-bdd}, if $\bm K_n\in\R^{N_\br\times d}$ has mean zero and Assumption~\ref{ass:avoidpa} holds, then for all large $n$, with probability at least $1-n^{-2}$, $\sigma_{\min}(\bm K_n)\gtrsim \sqrt{N}\gtrsim\sqrt{N_\br}$  and 
	\begin{align*}
		\sup_{\substack{\bm x\in  \range(\bm Q_n^\top),  \bm y\in\range(\bm K_n)\\ \|\bm x\|_2=\|\bm y\|_2 = 1}}\bm x^\top \bm y\lesssim \sqrt{\frac{n+d}{N}}.
	\end{align*}
	In particular, if $N/n\to\infty$ as $n \to \infty$, then Assumptions~\ref{ass:cosine}-\ref{ass:minK} hold a.s.  
\end{proposition}

\begin{proposition}\label{prop:rg}
	Let $\HH_n(\VV_n, \EE_n)$ be a hypergraph sequence with $\VV_n = [n]$. The following statements hold:
	\begin{enumerate}
		\item[(i)] If $\HH_n$ is sampled from a NURHM model with the corresponding $\xi_{n, \pm}$ in \eqref{myxis} satisfying
		${\xi^4_{n, +}(\log n)^5}/{\xi^5_{n,-}}\to 0$ as $ n\to\infty, \label{pj123}$
		then Assumption~\ref{ass:topology} holds a.s. Moreover, for any $\lambda>0$, $\diam(\A_{\HH_n}(\lambda))\lesssim (\log n)\xi_{n,+}/\xi_{n, -}$ and $h_{\HH_n}\gtrsim \xi_{n, -}$. 
		\item[(ii)] If $\HH_n$ is sampled from an HSBM model with the corresponding $\zeta_{n, -}$ in \eqref{myzeta} satisfying
		${(\log n)^5}/{\zeta_{n,-}}\to 0$ as $ n\to\infty, $
		then Assumption~\ref{ass:topology} holds a.s. Moreover, for any $\lambda>0$, $\diam(\A_{\HH_n}(\lambda))\lesssim\log n$ and $h_{\HH_n}\gtrsim \zeta_{n, -}$. 
	\end{enumerate}
\end{proposition}

\begin{remark}
	The sparsity conditions in Proposition~\ref{prop:rg} are optimal in the leading-order asymptotics. For either homogeneous NURHM ($\xi_{n, +} = \xi_{n, -}$) or HSBM, Assumption~\ref{ass:topology} holds as long as $(\log n)^5/\xi_{n, -}\to 0$ or $(\log n)^5/\zeta_{n, -}\to 0$. Meanwhile, the MLE exists only if the graph sequence is asymptotically connected, and for \ER graphs this occurs only if $\xi_{n,-}>\log n$.  
\end{remark}

Combining all the results, we arrive at Figure~\ref{fig:ass}~(b), which characterizes the implications of the assumptions used in the uniform consistency analysis under both deterministic and random designs. In particular, we have the following uniform consistency result under random designs. 

\begin{theorem}[Uniform consistency of the MLE under random designs]\label{thm:randdesign}
	Suppose the comparison graph sequence $\HH_n$ and the covariates $\{X_{\T, i}\}$ are generated from the random design as follows: (1). $\HH_n$ is sampled from NURHM with ${\xi^4_{n, +}(\log n)^5}/{\xi^5_{n,-}}\to 0$ or HSBM with $(\log n)^5/\zeta_{n, -}\to 0$; and (2). $\{X_{\T, i}\}_{\T\in\EE, i\in \T}$ are i.i.d. random vectors independent of $\HH_n$ and $n$ with a nondegenerate covariance matrix.
	Under Assumptions~\ref{ass:edgesize} and \ref{ass:unif-bdd}, a.s., the MLEs computed based on $\HH_n$ and $\{X_{\T, i}\}_{\T\in\EE, i\in \T}$ uniquely exist for all large $n$ and are uniformly consistent. 
\end{theorem}

\section{Numerical study: ATP tennis data}\label{sec:atp}

We apply the $\PLDC$ model to analyze an ATP tennis dataset collected over an extended period; more numerical results on synthetic and horse-racing data can be found in Sections B2-B3 of the supplementary file. To find the MLE in the $\PLDC$ model, we use Algorithm~\ref{alg:1} with $\e =  10^{-10}$ and initialization $\bm 0$. 

Tennis ranking data are frequently analyzed using \checkedd{pairwise comparison} models such as the BT model and its variants. One criticism of these analyses is that age is not taken into account while comparing different players; an older player's decline may be accompanied by a younger player's rise. To mitigate this, we incorporate covariates constructed based on players' ages to account for the age effects when comparing players' performances.

The dataset was collected by Jeff Sackmann and is openly accessible in this GitHub repository\footnote{\href{https://github.com/JeffSackmann/tennis_atp}{https://github.com/JeffSackmann/tennis-atp}}. Since the data in the early days is incomplete, we focus on the relatively recent period between {January 1980 and May 2024}. The dataset consists of tour-level main draw matches and tour-level qualifying and challenger main draw matches. We focus on the former to eliminate additional heterogeneous effects besides time, as its data size is sufficiently large to draw interesting conclusions. After preprocessing, we are left with $n=1,803$ players, each participating in at least $10$ matches in the dataset. The total number of matches is $N = 136,935$. 

We first identify appropriate covariates to represent the age effects. As a rule of thumb, a player's performance initially improves with experience over time and then declines due to aging. To model this phenomenon, we construct covariates using Gaussian radial basis functions $f(t; a, \lambdaf) = \exp\{{-\lambdaf (t-a)^2}\}$ with varying mean and scale parameters $(a, \lambdaf) \in \{25, 30, 35\}\times\{0.01, 0.03\}$, where $a$ denotes the potential peak time and $\lambdaf$ denotes the potential decreasing rate away from the peak. Although different individuals may have different peaks and rates, we only consider these effects at a population level. \checkedd{There are 64 candidate models and we apply the BIC criterion to select the best model; see Figure~\ref{fig:tennis} (a)-(b).}

We use the selected covariates to fit a $\PLDC$ model. The resulting normalized log-likelihood is $-0.594$, representing a $1.4\%$ improvement over the BT model, which has a log-likelihood of $-0.603$.  This improvement is not so small given the existing number of parameters in the BT model.  The resulting ranking statistics of the top 10 players are given in the table in Figure~\ref{fig:tennis} (c). Additionally, based on the estimated model parameters $\vv$, we plot the estimated age effect in Figure~\ref{fig:tennis} (b). Furthermore, we provide the complete log score dynamics of the top ten players throughout their careers in Figure~\ref{fig:dc}.

\begin{figure}[htbp]
	\centering
	\begin{minipage}{0.325\linewidth}
		\centering
		\includegraphics[width=\linewidth,  trim={0.5cm 0.1cm 2.5cm 2cm}, clip]{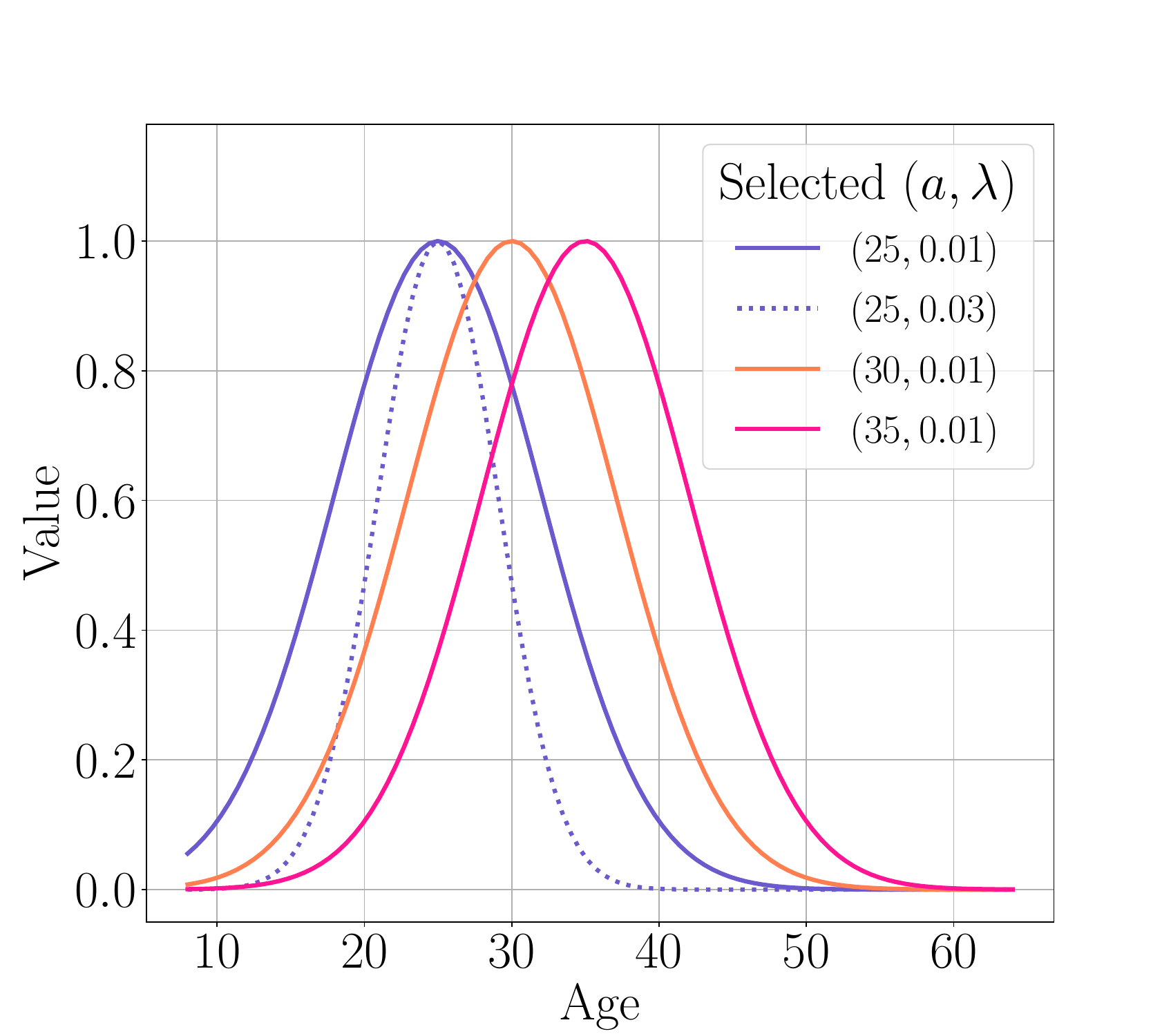}
		\caption*{(a)}\label{111}
	\end{minipage}\hfill 
	\begin{minipage}{0.325\linewidth}
		\centering
		\includegraphics[width=\linewidth,   trim={0.5cm 0.1cm 2.5cm 2cm}, clip]{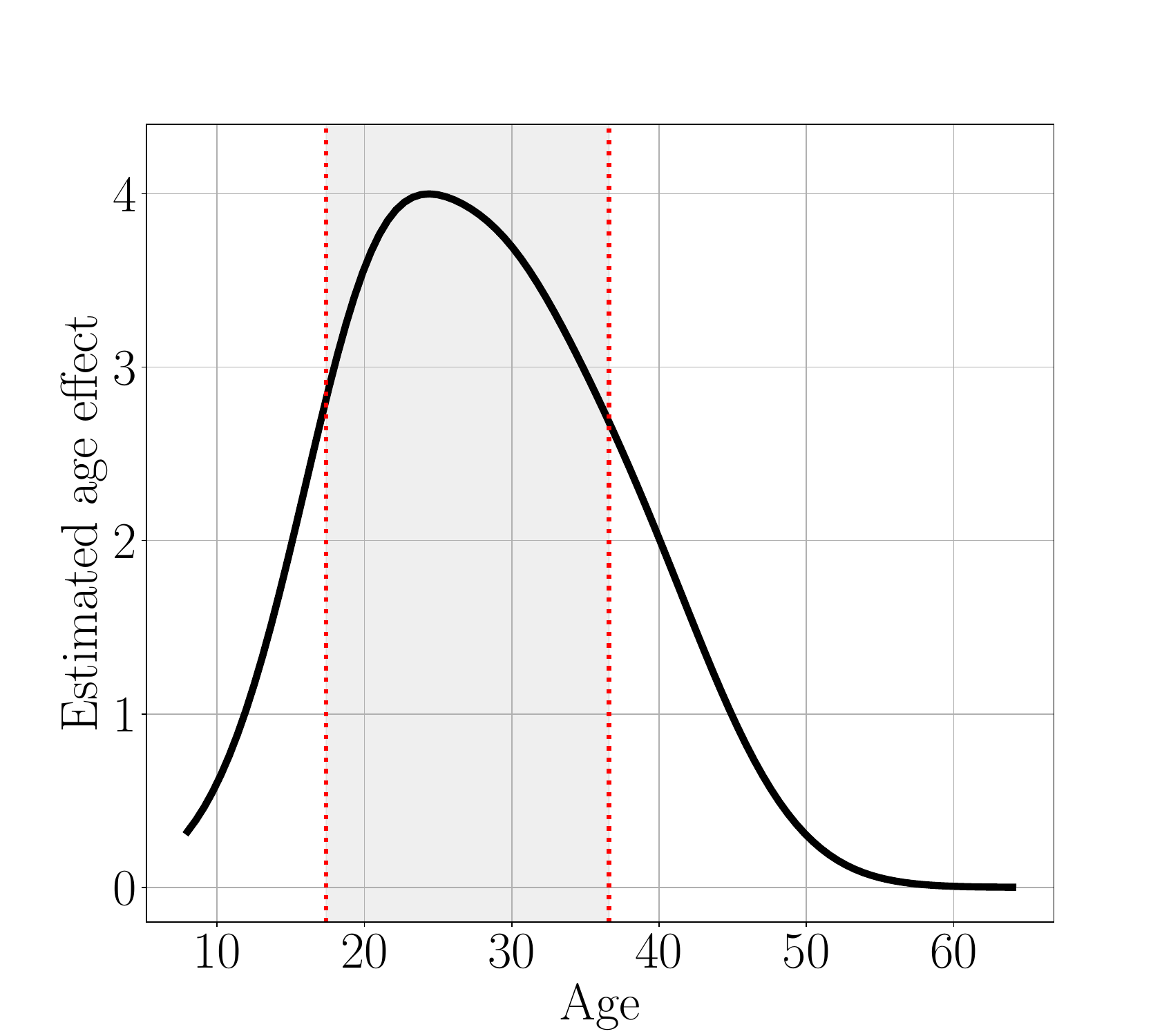}
		\caption*{(b)}
	\end{minipage}\hfill
	\begin{minipage}{0.325\textwidth}
		\vspace{3mm} 
		\resizebox{\textwidth}{!}{
			\begin{tabular}{ccc}
				Player & $\text{Rank}_{\text{BT}}$ & $\text{Rank}_\PLDC$  \\ \hline
				Roger Federer     & 3    & 1     \\
				Novak Djokovic     & 1   & 2  \\
				Rafael Nadal    & 2   & 3   \\
				Jimmy Connors     & 23   & 4     \\
				Andre Agassi      & 13   & 5   \\
				Ivan Lendl     & 7    & 6     \\
				Carlos Alcaraz     & 4   & 7   \\
				Pete Sampras      & 9   & 8    \\
				Andy Murray     & 5  & 9   \\
				John McEnroe     & 11    & 10    
			\end{tabular}
		}
		\caption*{~ \\ (c)}
	\end{minipage}
	\caption{
		(a): Four Gaussian radial basis functions in the best model selected by BIC. 
		(b): The estimated $\text{age effect} (t) =6.238\cdot\exp\{{-0.01\cdot (t-25)^2}\}-0.865\cdot\exp\{{-0.03\cdot (t-25)^2}\}-3.338\cdot\exp\{{-0.01\cdot (t-30)^2}\}+3.319\cdot\exp\{{-0.01\cdot (t-35)^2}\}$. 	The gray region bounded by the two vertical dashed lines captures $99\%$ of the age range of the players in the dataset.
		(c): Top ten players ranked by the $\PLDC$ model based on their estimated utilities and their ranks in the BT model. }
	\label{fig:tennis}
\end{figure}

The ranking results in Figure~\ref{fig:tennis} show a noticeable overlap in the top ten players selected by both BT and 
$\PLDC$. However, BT favors younger players while $\PLDC$ mitigates this bias, identifying legends like Jimmy Connors despite generational differences.  The estimated population age effects take an asymmetric hill shape, peaking around $24.3$ and slowly decreasing afterward. From this, we can discern the peak time for each player. For instance, in Figure~\ref{fig:dc}, the peak times for Federer and Djokovic are near 2006 and 2011, respectively, consistent with the information in Wikipedia. Even more surprisingly, our result successfully finds the peak of Jimmy Connors in the mid-1970s, which is out of range of the training data. These results suggest that the $\PLDC$ model provides a useful tool to compare players at different times on a fairer scale than the BT model. 

Figure~\ref{fig:dc} provides an interesting summary of history of tennis competition as learned from the data\footnote{\href{https://en.wikipedia.org/wiki/List_of_ATP_number_1_ranked_singles_tennis_players}{https://en.wikipedia.org/wiki/List\_of\_ATP\_number\_1\_ranked\_singles\_tennis\_players}}. In 1980, Jimmy Connors dominated until he was challenged by Ivan Lendl and John McEnroe. The turning point came in 1983 when Connors lost his No. 1 spot permanently. Lendl and McEnroe then led for seven years. In the early 1990s, Pete Sampras rose to the top, soon facing his longtime rival, Andre Agassi, in a highly competitive era. This lasted for more than a decade until the winter of 2004. The ``Big Four'': Roger Federer, Rafael Nadal, Novak Djokovic, and Andy Murray began their 18-year reign. Recently, Carlos Alcaraz has emerged as a potential successor, but time will tell his final place in history.

\section{Discussion}

In this paper, we introduced the $\PLDC$ statistical ranking model, which incorporates dynamic covariates within the PL framework. We formulated the model, characterized its identifiability conditions, and examined various aspects of the MLE, including its unique existence, uniform consistency, and algorithms. Numerical experiments are provided to validate our findings and demonstrate the model's application to real-world ranking datasets.

As a next step, one may consider conducting statistical inference for the MLE in $\PLDC$. In our case, the MLE is biased. Due to the interplay between local and global parameters, a debiasing procedure may be needed to establish an asymptotic normality result for inference. 
However, this procedure is quite technical and demands meticulous asymptotic analysis.  In fact, conducting inference in heterogeneous graphs with varying node degrees is an open problem in the statistical ranking literature.  For model interpretation, one might be interested in covariate selection and causal inference in the presence of unobserved confounders.  We leave these topics for future investigation.

	\medskip
	\noindent 
	\textbf{Supplementary Material}: The supplementary file provides a list of notations, additional numerical results,  and all the technical proofs.

\newpage

\begin{appendix}

\section*{Supplementary}
In this supplementary file, we provide additional details on the implementation and analysis of the alternating maximization algorithm, along with numerical experiments on synthetic and horse racing data. We also present detailed proofs of all technical results in the manuscript. For the reader's convenience, we list the important notation used in both the main article and this supplement in Section \ref{sec:notation}. The details of Algorithm \ref{alg:1} and additional numerical results are included in Section \ref{sec:addnum}. Technical proofs related to model identifiability are provided in Section \ref{proof:identifiability}, while those related to parameter estimation are included in Section \ref{proof:emle}. The proof of uniform consistency is given in Section \ref{sec:mainproof}. Additional illustrations of concepts related to graph topology are provided in Section~\ref{FFF}.
\section{Notation list}\label{sec:notation}
In this part, we summarize the important notation used throughout both the main article and the supplementary.\vspace{0.5cm}

{\bf Basic notation:}
\begin{itemize}
	\item $[n]: =\{1,\ldots,n\}$, where $ n\in\mathbb N $ is the number of subjects.
	\item $\mathbb I_{B}(\cdot)$ is the indicator function on a set $B$.
	\item $O(\cdot)$ and $o(\cdot)$ are the Bachmann--Landau asymptotic notation. 	$\lesssim$ and $\gtrsim$ represent the asymptotic inequality relations, and $\asymp$ if both $\lesssim$ and $\gtrsim$ hold.
	
	\item $\u\in\R^n$ is a local parameter in the $\PLDC$ model. $\v\in\R^d$ is a global parameter in the $\PLDC$ model. $\t = (\u^\top,\v^\top)^\top$.
	
	\item  $\bm\t^* = [(\u^*)^\top, (\bm v^*)^\top]^\top$ and $\tt = (\widehat{\bm u}^\top, \widehat{\bm v}^\top)^\top$ representing true value and maximum likelihood estimation respectively.
	\item $N$: the number of total comparison data.
	\item $\T_i, i\in[N]$: the set collecting the subjects that participate the $i$th comparison.
	\item $m_i$: the number of subjects that participate the $i$th comparison, namely $m_i = |\T_i|.$ 
	The upper bound of $m_i$ is $M$.
	\item  $\pi_i(j)$ represents the item ranked $j$ in $\T_i$ and $r_i(k)$ represents the rank of item $k$ in $\T_i$.
	\item $X_{\T, j}\in\R^d$ is the covariate of $j$th item in  $\T$; $s_j(\t;\T):= u_j + X^\top_{\T,j}\v$ is the score of item $j$ in $\T$.
	\item $s_{ij}:=s_{\pi_i(j)}(\t; \T_i) = u_{\pi_i(j)} + X_{\T_i, \pi_i(j)}^\top\bm\v$ is score of the item ranked $j$ in $\T_i$.
	\item $R:=\sup_{n\in\N}\max\left\{\|\t^*\|_\infty, \max_{i\in [N], j\in \T_i}\|X_{\T_i, j}\|_1\right\}$ is the upper bound.

\end{itemize}

\noindent{\bf Notation used in hypergraph:}
\begin{itemize}
	\item $\VV=[n] := \{1, \ldots, n\}$ is the set of items, and $\mathscr P(\VV)$ is the power set of $\VV$.
	\item  $\mathcal S(\T)$ is the set of permutation on $\T$, for $\T\subseteq\VV$.
	
	\item ${\V\choose m}= \{\T\in\mathscr P(\VV): |\T|=m\}$ is the set of $\muu$-edge.
	\item $\HH(\VV, \EE)$ is a comparison hypergraph with $N$ edges, where $\EE = \{\T_i\}_{i\in [N]}\subseteq \mathscr{P}(\VV)$.
	\item $\deg(k) := |\{\T_i\in \EE: k\in \T_i\}|$ is the degree of $k$th item;
	$\partial U := \{\T\in \EE: e\cap U\neq\emptyset, e\cap U^\complement\neq\emptyset\}$ is the boundary of subset of item $U$;
	$h_\HH$ is the modified Cheeger constant defined in Definition \ref{def:ch}. 	
	\item $\A_\HH(\lambda)$ is the set of $\lambda$-weakly admissible sequences and $\diam(\A_\HH(\lambda))$ is the corresponding diameter, that is, $\diam(\A_\HH(\lambda))=\max_{\{A_j\}_{j=1}^J\in\A_\HH(\lambda)}J$.
	\item $\HH_\br$, $\EE_\br$ and ${N_{\br}}$ are the graph, edges and the number of edges after breaking respectively.  
	\item $\Q $: the incidence matrix of $\HH_\br$; $\K$: the matrix of  the covariate differences between subjects after breaking.
	\item $\bm{W}_n = [\Q^\top,\K]\in\R^{\M\times (n+d)}$ is the augment matrix corresponding to $\HH_\br$.
	\item {The $\range(\Q^\top)$ and $\range(\K)$ are the linear space spanned by the columns of $\Q^\top$ and $\K$.}
	\item $p_n^{(m)}$: the lower bound of the probability of an $m$-size edge in NURHM; $q_n^{(m)}$: the upper bound of the probability of an $m$-size edge in NURHM.
	\item In NURHM, we define $\xi_{n,-}:=\sum_{m=2}^M n^{m-1}p_n^{(m)}$ and $ \xi_{n,+}:=\sum_{m=2}^M n^{m-1}q_n^{(m)}$.
	\item $\omega_{n, \ell}$, $\ell\in[L]$: the probability of an edge in $\ell$th community in HSBM.
	\item In HSBM, we define $\zeta_{n,-}:=n^{M-1}\min_{0\leq \ell \leq L}\omega_{n, \ell}$.
	\end{itemize}

\section{Additional numerical results}\label{sec:addnum}

In this section, we first provide additional details of Algorithm \ref{alg:1} concerning its implementation. Then, we verify and apply the results established in the main manuscript. This includes verifying the uniform consistency of the MLE using synthetic data and applying the $\PLDC$ model to a horse-racing dataset with in-field and off-field information.  In all experiments, we use Algorithm~\ref{alg:1} with stopping parameter {$\e =  10^{-10}$} and the initialization for $\bm\t^{(0)} =\bm 0$.

\subsection{Details of Algorithm \ref{alg:1}}

We now describe the details for implementing steps 2 and 3 in Algorithm \ref{alg:1}. Then, we discuss the overall computational complexity for its implementation. 

\smallskip
\underline{\bf \texorpdfstring{Fix $\v$ and update $\u$}{}.} 
We start by fixing $\v$ and optimize over $\u$. We utilize an explicit iteration method that generalizes the Minorize--Maximization (MM) algorithm in \cite{MR2051012}. Recall the likelihood function $l(\u, \v)$:
\begin{align*}
	l(\u, \v) = \sum_{i\in [N]}\sum_{j\in [m_i]}\left[s_{ij}(\u, \v)-\log(\sum_{t=j}^{m_i}\exp{s_{it}(\u, \v)})\right],
\end{align*}
where $s_{ij}(\u, \v) = u_{\T_i,\pi_i(j)} + X_{\T_i, \pi_i(j)}^\top\bm\v$. Here, we write $\u, \v$ in place of $\t$ to emphasize their dependence.

To solve the subproblem $\bm u^{(\tauf+1)}= \argmax_{\bm u: \bm 1^\top \bm u = 0}l(\bm u, \bm v^{(\tauf)})$, we consider a sequence of minorizing functions $\{Q(\u \mid\u^{(\tauf,\rhof)}, \v^{(\tauf)})\}_{\tauf, \rhof}$ constructed based on a sequence $\{\u^{(\tauf,\rhof)}\}_{\rhof\in\N}$: 
{\fontsize{10.8}{12} \begin{align*}
		&Q(\u \mid \u^{(\tauf,\rhof)} , \v^{(\tauf)}) \\
		=& \sum_{i\in [N]}\sum_{j\in [m_i]}\Bigg[s_{ij}(\u, \v^{(\tauf)})-\frac{\sum_{t=j}^{m_i}\exp{s_{it}(\u, \v^{(\tauf)})}}{\sum_{t=j}^{m_i}\exp{s_{it}(\u^{(\tauf,\rhof)}, \v^{(\tauf)})}}+1-\log\bigg(\sum_{t=j}^{m_i}\exp{s_{it}(\u^{(\tauf,\rhof)}, \v^{(\tauf)})}\bigg)\Bigg].
\end{align*}}
It is straightforward to check that $Q(\u \mid\u^{(\tauf,\rhof)} , \v^{(\tauf)})$ are strictly concave whenever $\HH$ is connected. Moreover, since 
$	-\log(y)\ge -\log(x)+1-(y/x)$ holds for any $ x,y >0$,  we obtain for all $\u\in\R^n$, 
$l(\u, \v^{(\tauf)})\ge Q(\u \mid\u^{(\tauf,\rhof)} , \v^{(\tauf)}), $
with equality if $\u = \u^{(\tauf,\rhof)}.$ If we set the initial $ \u^{(\tauf,0)} = \u^{(\tauf)}$
and define ${\u}^{(\tauf,\rhof)}$ recursively as 
\begin{align}\label{mm}
	{\u}^{(\tauf,\rhof+1)} = \argmax_{\bm 1^\top\u = 0}Q(\u \mid\u^{(\tauf,\rhof)} , \v^{(\tauf)}), 
\end{align} 
then we have the following result.
\begin{proposition}\label{prop:mm}
	Under Assumption~\ref{ass:ue}, for every $\tauf$ and $\u^{(\tauf,\rhof)}$ defined in \eqref{mm}, 
	$\u^{(\tauf+1)} = \lim_{\rhof\to\infty}\u^{(\tauf,\rhof)}.$
\end{proposition}
In practice, a more explicit update rule for ${\u}^{(\tauf,\rhof+1)}$ can be obtained by setting \\$\nabla_\u Q(\u \mid\u^{(\tauf,\rhof)}, \v^{(\tauf)})=\bm 0$:
\begin{align}
	{u}_k^{(\tauf,\rhof+1)} = \log\{\deg(k)\} - \log\left[\sum_{i:k\in \T_i}\sum_{j\in [r_i(k)]}\frac{\exp{X^\top_{\T_i,k}\v^{(\tauf)}}}{\sum_{t=j}^{m_i}\exp{s_{it}(\u^{(\tauf,\rhof)}, \v^{(\tauf)})}}\right] \quad \ k \in [n]\label{eq:alt_u},
\end{align}
followed by centering. An alternative way to obtain \eqref{eq:alt_u} is to directly set $\nabla_\u Q(\u \mid\u^{(\tauf,\rhof)}, \v^{(\tauf)})=\bm 0$ and rearrange terms into a set of fixed point equations, similar to the Zermelo's algorithm in the BT model with no covariates \citep{newman2023efficient}. 

\smallskip

\noindent
\underline{\bf\texorpdfstring{Fix $\u$ and update $\v$.}{}}
Fixing $\bm{u}$, the model for $\bm{v}$ is similar to the logistic regression. 
Since $d$ is assumed to be fixed and does not grow with $n$, the subproblem $\bm v^{(\tauf+1)}\gets \argmax_{\bm v}l(\bm u^{(\tauf+1)}, \bm v)$ is a standard convex optimization problem in a fixed dimension. For convenience, we consider a damped Newton--Raphson procedure to find the minimizer of $\v$ conditional on the current estimate $\bm u$:
\begin{align}
	\v^{(\tauf,\rhof+1)} = \v^{(\tauf,\rhof)} - \nu^{(\tauf, \rhof)}\{\nabla^2_\v l(\u^{(\tauf)}, \v^{(\tauf,\rhof)})\}^{-1}\nabla_\v l(\u^{(\tauf)}, \v^{(\tauf,\rhof)}), \label{nt}
\end{align}
where $\nu^{(\tauf, \rhof)}>0$ is the step size, 
{\fontsize{10.5}{12}\begin{align*}
		\nabla_\v l(\u^{(\tauf)}, \v^{(\tauf,\rhof)}) &= \sum_{i\in [N]}\sum_{j\in [m_i]}\left(X_{\T_i, \pi_i(j)}-\frac{\sum_{t=j}^{m_i}\exp{s_{it}(\u^{(\tauf)}, \v^{(\tauf,\rhof)})}X_{\T_i, \pi_i(t)}}{\sum_{t=j}^{m_i}\exp{s_{it}(\u^{(\tauf)}, \v^{(\tauf,\rhof)})}}\right)\\
		\nabla^2_\v l(\u^{(\tauf)}, \v^{(\tauf,\rhof)}) &= -\sum_{i\in [N]}\sum_{j\in [m_i]}[X_{\T_i, \pi_i(j)}, \ldots, X_{\T_i, \pi_i(m_i)}][\text{diag}(\bm b_{ij})-\bm b_{ij}\bm b_{ij}^\top][X_{\T_i, \pi_i(j)}, \ldots, X_{\T_i, \pi_i(m_i)}]^\top,
\end{align*}}
where
\begin{align*}
	\bm b_{ij} = \frac{1}{\sum_{t=j}^{m_i}\exp{s_{it}(\u^{(\tauf)}, \v^{(\tauf,\rhof)})}}\left(\exp{s_{ij}(\u^{(\tauf)}, \v^{(\tauf,\rhof)})}, \ldots, \exp{s_{im_i}(\u^{(\tauf)}, \v^{(\tauf,\rhof)})}\right)^\top
\end{align*}
and $\text{diag}(\bm b_{ij})$ stands for the diagonal matrix with diagonal entries $\bm b_{ij}$. 
When $\nu^{(\tauf, \rhof)}\equiv 1$, this is the Newton--Raphson method.
A more advanced method involves selecting $\nu^{(\tauf, \rhof)}$ using the backtracking search. For further details on convergence analysis, see \cite{nesterov1994interior, nocedal1999numerical}. In our implementation, the step size parameters in updating $\v$ are set to be equal to $1$. (We observe that such a simple choice is sufficient to lead to fast convergence in our experiments.)
\smallskip

\noindent
\underline{\bf \texorpdfstring{Computational complexity}{}.} 
To find the overall computational complexity of Algorithm \ref{alg:1}, we first analyze the computational complexity of a single iteration, which involves updating both $\u$ and $\v$. These updates are approximately achieved using iterative algorithms (MM the Newton--Raphson algorithms), each requiring $\mathcal{O}(N)$ operations per update; see \eqref{eq:alt_u}-\eqref{nt}. Our proof of Proposition~\ref{prop:mm} shows that the MM algorithm has a linear convergence rate, while the Newton--Raphson algorithm typically exhibits quadratic local convergence. Assuming their respective convergence rate parameters do not depend on $n$ (these parameters are related to the smoothness and strict concavity of the normalized log-likelihood), the total number of updates required to achieve a tolerance $\e'$ within a single iteration of Algorithm \ref{alg:1} is $\mathcal{O}(\log (1/\e'))$.

On the other hand, due to the linear convergence of Algorithm \ref{alg:1}, the total number of iterations required for Algorithm \ref{alg:1} to achieve a global tolerance $\e$ is $\mathcal{O}(\log(1/\e))$. Therefore, the overall complexity of implementing Algorithm \ref{alg:1} is $\mathcal{O}(N\log (1/\e)\log(1/\e'))$. Setting $\e' = \e$, this complexity simplifies to $\mathcal{O}(N(\log (1/\e))^2)$, which remains effectively linear in $N$.

\subsection{Synthetic data}

We verify the uniform consistency of the MLE in the $\PLDC$ model using simulated data. To this end, we generate comparison graphs using the random graph models introduced in Section~\ref{sec:rg} and the comparison outcomes on edges using the $\PLDC$ model with three dynamic covariates ($d=3$). The edge sizes are randomly chosen between $2$ and $7$ in NURHM and fixed as $5$ in HSBM. For dynamic covariates, we simulate them as independent three-dimensional standard multivariate normal vectors for each object per edge. The true utility vector $\u^*$ is generated with each component selected uniformly in $[-0.5, 0.5]$ and centered to satisfy ${\bm 1}^\top \u^*=0$; the true model coefficient vector $\v^*$ is set to $\v^*=(1, -0.5, 0)^\top$. The total size of edges $N$ is chosen as: $N = 0.1n(\text{log}n)^3$ in NURHM, and $N = 0.07n^2$ in HSBM. More simulation details are given below. We repeat the experiment 300 times across a range of different $n\in\{200, 400, 600, 800, 1000\}$ and calculate $\|\uu-\u^*\|_\infty$ and $\|\vv-\v^*\|_\infty$. The simulation results are reported in Figure~\ref{fig:syn}. 

\begin{table}[h]
	\centering
	\begin{tabular}{c c c}
		Numbers of items $n$ & $ N $ in NURHM & $N$ in HSBM \\ \hline
		200     & 2,974    &  2,800 \\
		400     & 8,603    & 11,200\\
		600     & 15,706   & 25,200\\
		800     & 23,895   & 44,800 \\
		1000    & 32,961   & 70,000\\
	\end{tabular}
	\caption{Number of items $n$ and edges $N$ used in NURHM and HSBM.}
	\label{tab:num}
\end{table}

\begin{itemize}
	\item[(a).] NURHM: six groups are considered with $p_n^{(\muu)} = q_n^{(\muu)}$ for $\muu\in\{2,3,4,5,6,7\}$. The edges are equally allocated to each group.
	\item[(b).] HSBM: We consider two blocks $\VV = \VV_1\cup \VV_2$ with $|\VV_1| = \lfloor {n}/{3}\rfloor$ and $|\VV_2| = \lceil {2n}/{3} \rceil$, where $\lfloor\cdot\rfloor$ and $\lceil\cdot\rceil$ are the flooring and ceiling function, respectively. For each $i\in [N]$, the sampling probability of the $i$th edge is given as
	\begin{equation*}
		\P\{e_i\subseteq \VV_1\} : \P\{e_i\subseteq \VV_2\} : \P\{e_i\in \partial \VV_1\} = 5n : 20n : 4({\log}n)^3,
	\end{equation*}
	and each edges size $\mM$ is set to $\mM=5$. Once determining the community of $e_k$, we uniformly select an edge among all the possible edges in the corresponding community.
	
\end{itemize} 

\begin{figure}[htbp]
	\centering
	\begin{minipage}{0.24\linewidth}
		\centering
		\includegraphics[width=\linewidth, trim={0.3cm 0.6cm 1cm 0.5cm}, clip]{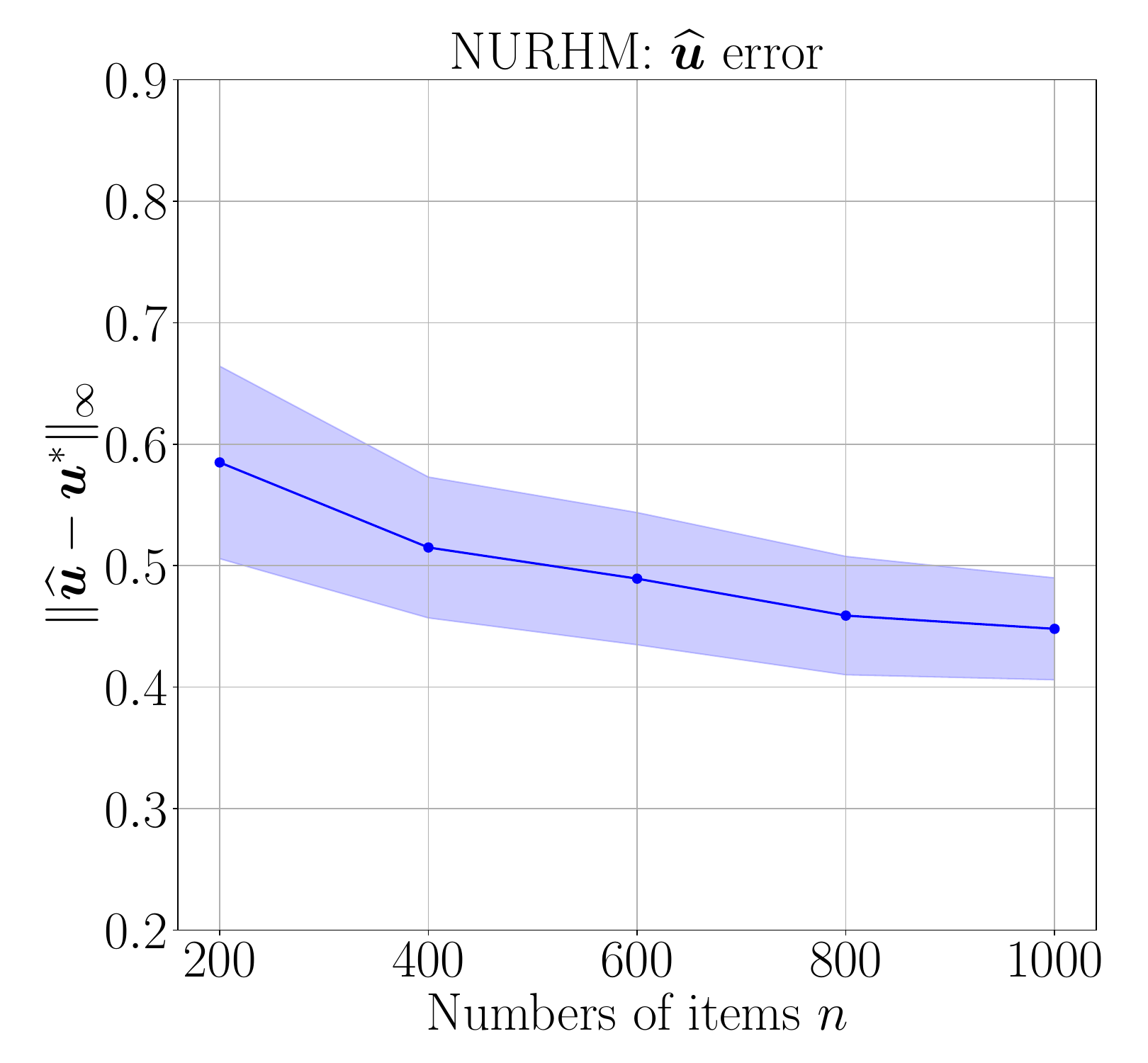}
		\caption*{(a)}
	\end{minipage}
	\begin{minipage}{0.24\linewidth}
		\centering
		\includegraphics[width=\linewidth, trim={0.0cm 0.6cm 1cm 0.5cm}, clip]{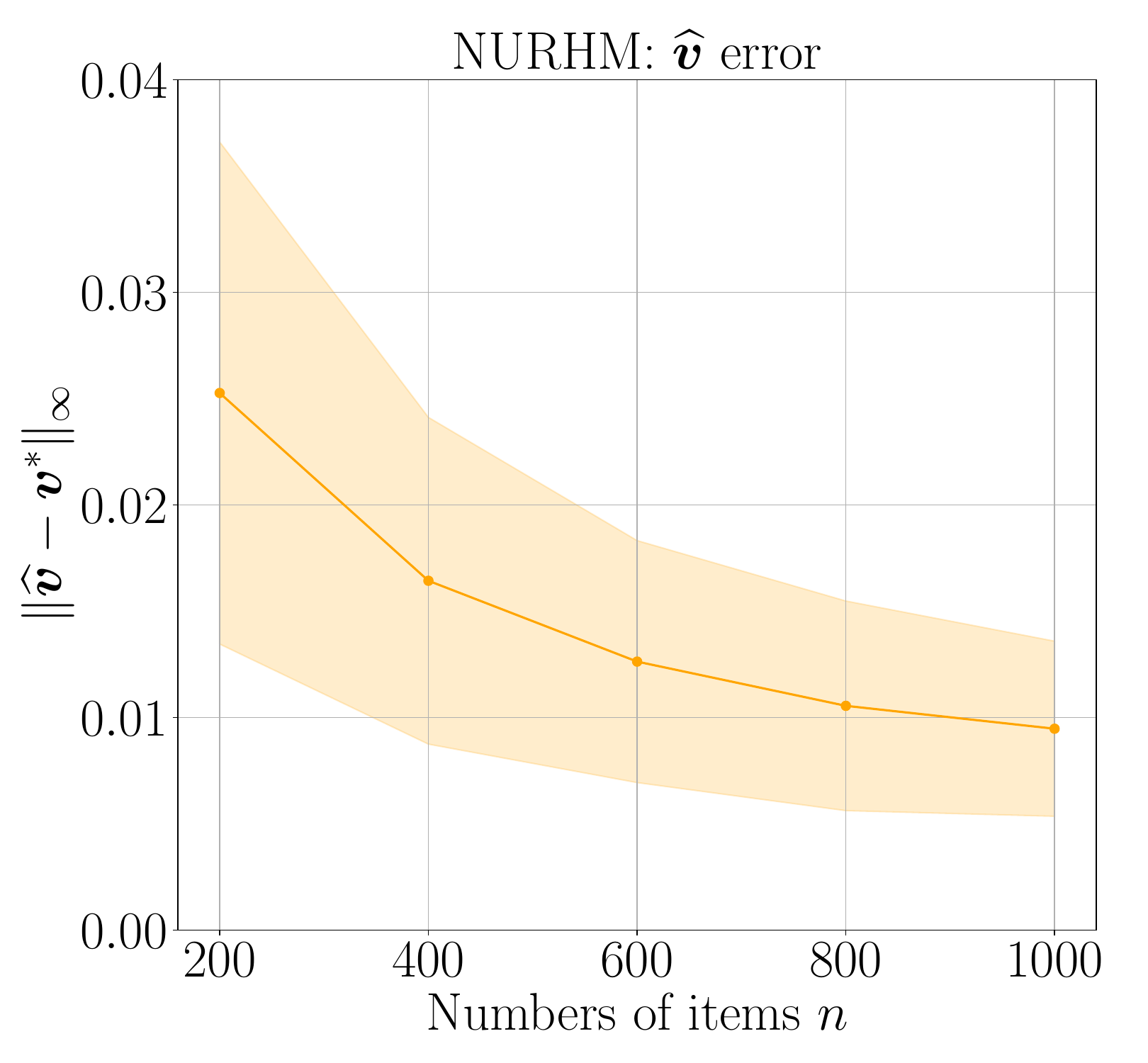}
		\caption*{(b)}
	\end{minipage}
	\begin{minipage}{0.24\linewidth}
		\centering
		\includegraphics[width=\linewidth, trim={0.3cm 0.6cm 1cm 0.5cm}, clip]{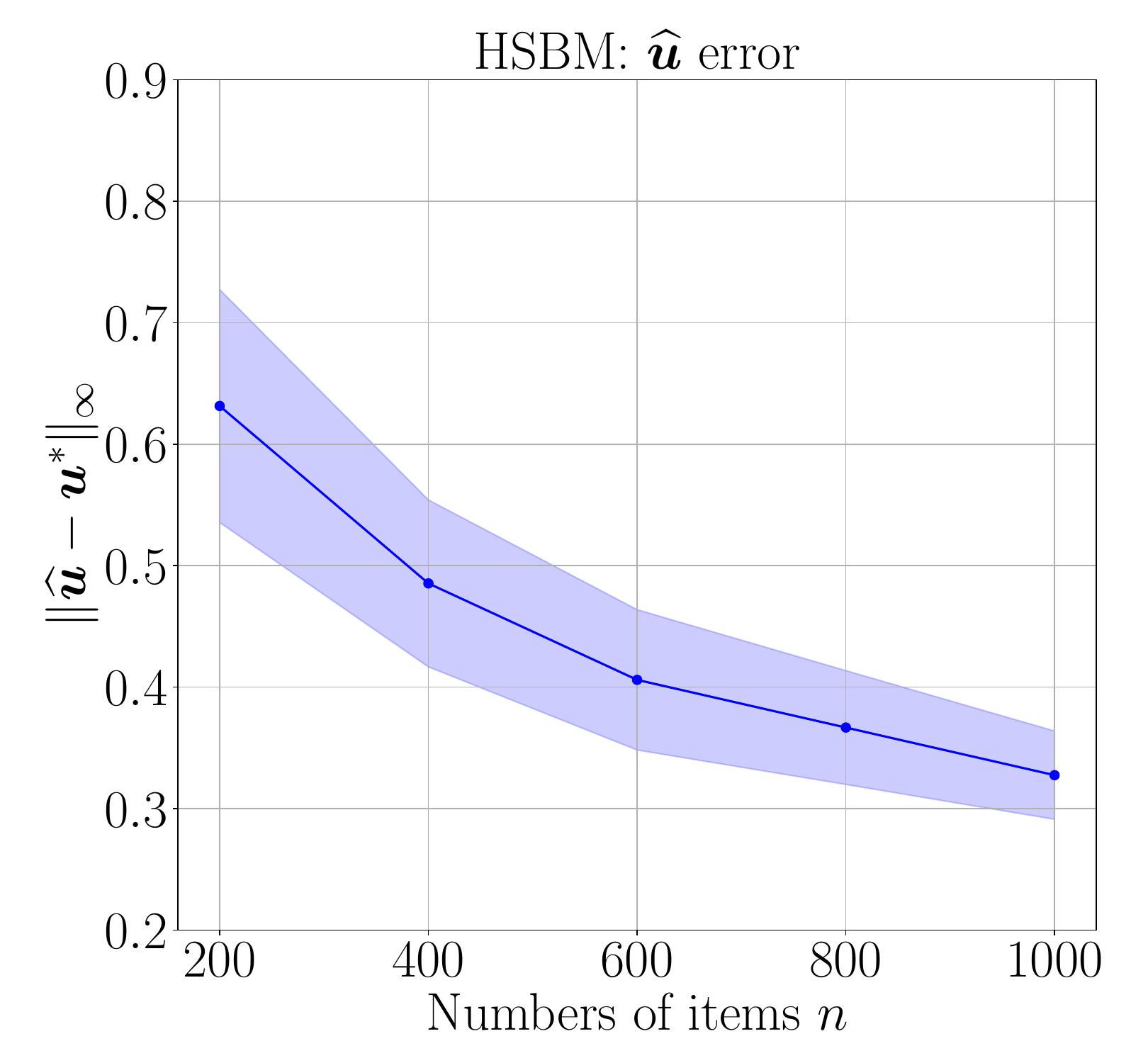}
		\caption*{(c)}
	\end{minipage}
	\begin{minipage}{0.24\linewidth}
		\centering
		\includegraphics[width=\linewidth, trim={0.0cm 0.6cm 1cm 0.5cm}, clip]{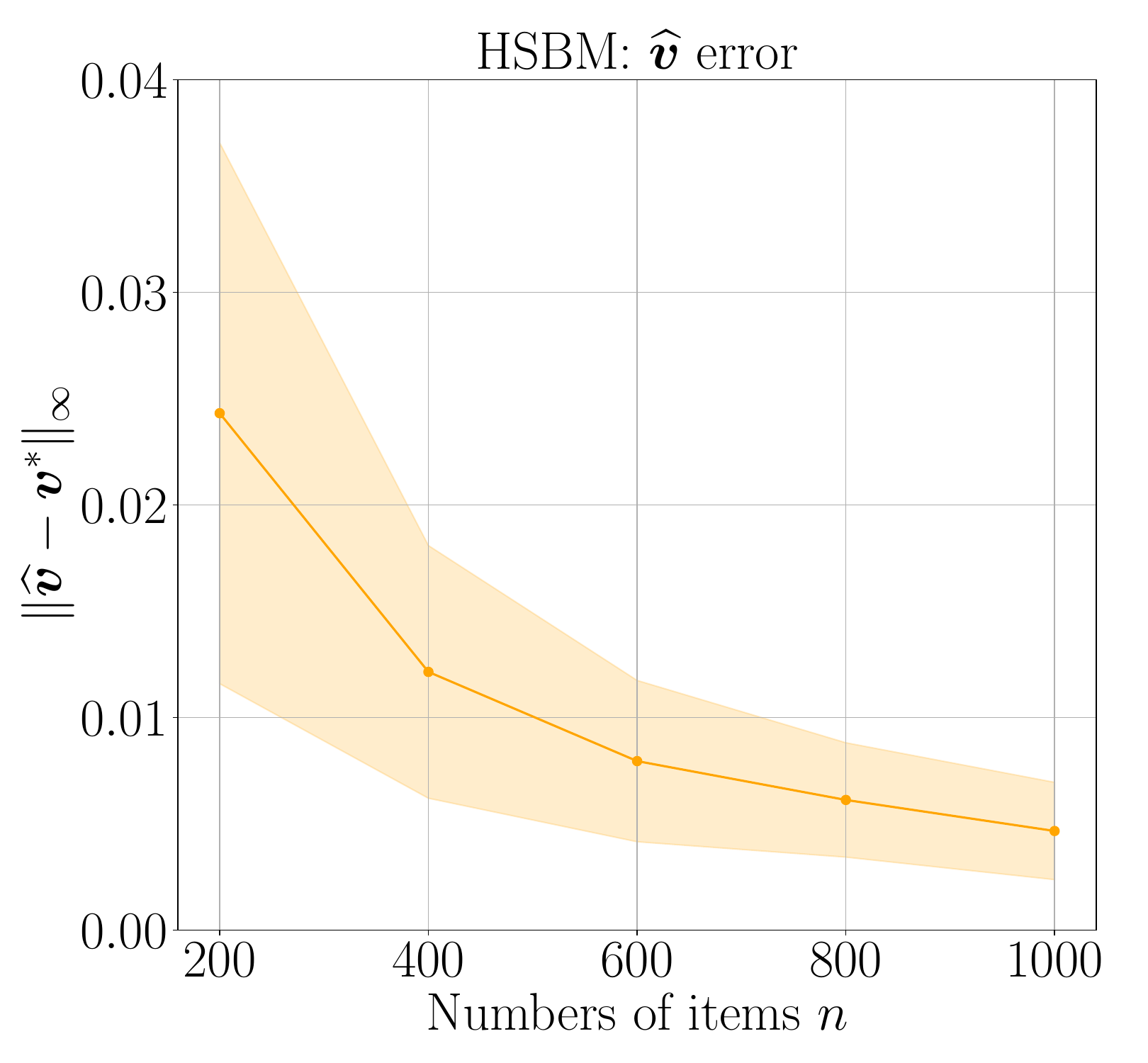}
		\caption*{(d)}
	\end{minipage}
	\caption{ Average value of $\|\uu-\u^*\|_\infty$ and $\|\vv-\v^*\|_\infty$ with the shadow area representing the sample standard deviation. (a)-(b): NURHM; (c)-(d): HSBM.}
	\label{fig:syn}
\end{figure}

The above simulation setup covers both cases of varying edge sizes (NURHM) and imbalanced data across objects (HSBM). The results in Figure~\ref{fig:syn} show that in both scenarios, the maximum entrywise errors of the joint MLE for $\u^*$ and $\v^*$ decay to zero as $n$ increases. This numerically verifies the uniform consistency result of the MLE in Theorem~\ref{thm:main}. Additionally, we notice that in both cases, the convergence rate of $\vv$ seems to be faster than $\u$, which provides further numerical evidence that the MLE of the global parameter may have a faster convergence rate, which is also consistent with the quantitative error bounds \eqref{rates} in Theorem~\ref{thm:main}.

\subsection{Horse-racing data}\label{sec:horse}

Horse racing is an equestrian sport with its associated betting industry constituting a pillar part of Hong Kong's economy. In this example, we consider a public Hong Kong horse-racing dataset\footnote{\href{https://www.kaggle.com/datasets/gdaley/hkracing}{https://www.kaggle.com/datasets/gdaley/hkracing}}. This dataset contains thousands of races of different horses from 1999 to 2005, along with relevant information. We first cleaned the dataset by removing horses that either participated in too few competitions or had won/lost in all the competitions they participated in. After preprocessing, there are $n=2,814$ horses and $N=6,328$ races. The size of comparisons is nonuniform and ranges between 4 and 14. The minimum degree of a horse is $10>\log n\approx 8$. 

To apply the $\PLDC$ model, we consider three factors to explain the internal scores of horses: the actual weight a horse carries in a competition, the draw (the post position number of a horse in a race), and the public belief (which is the winning probability derived from the last-minute win odds before the competition starts). The first two factors are in-field and reflect how the site conditions affect a horse's performance, while the third is an off-field factor summarizing unobserved information about the horses not explicitly observed from the dataset. We take the logarithm of the first and third factors (applying centering and normalization to the first factor as well) and rescale the draw information in each competition to a number between 0 and 1 based on the respective competition size. We apply the MLE to the full dataset and compute the MLE.

In this example, the estimated model coefficient vector is $\vv = (-0.139, -0.18\CR{7}, 0.502)^\top$. This suggests that either carrying a heavier weight or having a higher draw position negatively impacts a horse's performance. Conversely, public belief serves as a positive indicator of a horse's strength. To draw more convincing conclusions, we compare the log-likelihood, AIC, and BIC corresponding to the different models in Table~\ref{tab:model_choice}. The results indicate that $\PLDC$ performs better than PL in terms of the likelihood-based criteria, providing further evidence for incorporating dynamic covariates in the PL model to enhance prediction. 

\begin{table}[h]
	\centering
	\resizebox{\textwidth}{!}{
		\begin{tabular}{ccccccccc}
			& $({1,1,1})$ & $({0,1,1})$ & $({1,0,1})$ & $({1,1,0})$ & $({1,0,0})$ & $({0,1,0})$ & $({0,0,1})$ & $({0,0,0})$ \\ \hline
			\bf log-likelihood     & -16.985 & -17.033  & -16.996 
			& -17.573 & -17.650 & -17.594 & -17.046 & -17.671   \\
			\bf AIC                & 34.860  &  34.957  &  34.882 
			&  36.\CR{036} &  36.189 &  36.078 &  34.981 & 36.23\CR{0}\\
			\bf BIC                & 37.86\CR{5}  &  37.96\CR{1}  &  37.88\CR{6} 
			&  39.04\CR{0} &  39.19\CR{1} &  39.08\CR{0} &  37.98\CR{3} & 39.23\CR{2}\\
		\end{tabular}
	}
	\caption{The eight models above are generated from $\PLDC$ using different combinations of the three covariates. Here, $({i_1,i_2,i_3})$ represents the $\PLDC$ model incorporating the $t$th covariate if $i_t=1$. The PL model, namely $({0,0,0})$, is a special case of $\PLDC$. The log-likelihood, AIC, and BIC are normalized.}
	\label{tab:model_choice}
\end{table}

The estimated utility vector $\uu$ is plotted against the MLE obtained from the PL model in Figure~\ref{fig:struct}. The deviation of the data from the diagonal line is attributed to the covariates in the $\PLDC$ model. Moreover, we list the top ten horses ranked based on their estimated utilities in the $\PLDC$ model as well as other relevant statistics in the table in Figure~\ref{fig:struct}. It can be observed that the top-eight horses ranked by $\PLDC$ are shuffled compared to those in the PL model after accounting for the covariates.

\begin{figure}[htbp]
	\begin{minipage}{0.27\linewidth}
		\centering
		\includegraphics[width=\linewidth, trim={0cm 0.2cm 2cm 2cm}, clip]{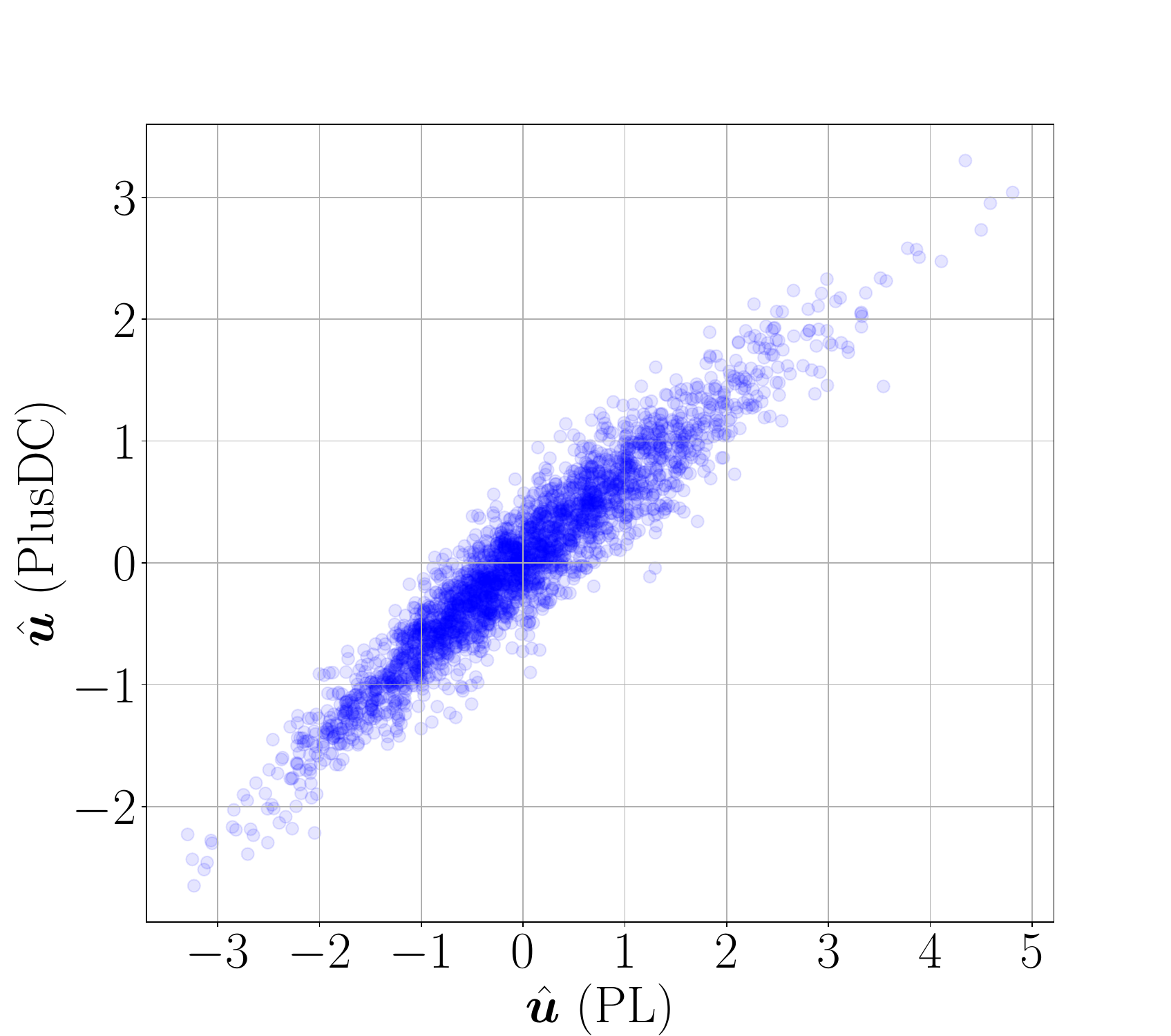}
	\end{minipage}
	\hfill
	\begin{minipage}{0.72\textwidth}
		\resizebox{\textwidth}{!}{
			\begin{tabular}{cccccccc}
				Horse id & Race & Mean place & Mean covariates (weight, draw, belief) & $\uu_{\PL}$ & $\uu_{\PLDC}$ & $\text{Rank}_{\PL}$  &  $\text{Rank}_{\PLDC}$ \\ \hline
				564     & 11    & 1.364  & (0.636, 0.506, 0.405)  & 4.344 & 3.302   &  4     & 1    \\
				1033     & 21   & 1.952   & (0.52\CR{8}, 0.67\CR{8}, 0.451)   &  4.809 & 3.040    & 1     & 2    \\
				2402     & 14   & 2.000   & (0.4\CR{46}, 0.5\CR{49}, 0.309)  &  4.590 & 2.952    & 2  & 3    \\
				1588     & 12   & 3.083  & (0.1\CR{89}, 0.4\CR{71}, 0.342)  &  4.500 & 2.733      &   3   & 4    \\
				160      & 11    & 3.182  &(0.\CR{015}, 0.5\CR{17}, 0.247)   &  3.778 & 2.582    &  8     & 5    \\
				2558     & 28    & 2.536   & (0.5\CR{34}, 0.52\CR{9}, 0.257)  &  3.864 & 2.571     & 7      & 6    \\
				2830     & 22   & 3.182  & (0.45\CR{9}, 0.5\CR{82}, 0.254)   &  3.891  & 2.509   & 6        & 7    \\
				218      & 16   & 2.438   & (0.54\CR{1}, 0.5\CR{30}, 0.273) &  4.109 & 2.475       & 5    & 8    \\
				1044     & 19   & 2.842   & (0.5\CR{32}, 0.\CR{500}, 0.309)   &   3.510 & 2.338   & 11         & 9    \\
				3044     & 34    & 3.647  & (0.52\CR{7}, 0.49\CR{8}, 0.162)  & 2.984 & 2.329    & 26      & 10   
			\end{tabular}
		}
	\end{minipage}
	\caption{Left: Scatter plot of the MLE estimate $\uu$ of the utility vector for both the PL and $\PLDC$ models. Right: Detailed information about the top ten horses ranked according to the $\PLDC$ model, including the number of races they participated in, their mean historical placement, the mean covariates information (normalized log actual weight, normalized draw, and log public belief of their chances of winning), the estimated utilities, and their estimated ranking results in the PL model.}
	\label{fig:struct}
\end{figure}

To further investigate whether the $\PLDC$ model improves prediction, we conduct a $\mathsf{k}$-fold cross-validation with $\mathsf{k}=63$. We randomly partition the dataset into 63 equal-sized subsamples (with the last one slightly larger than the rest) and use them as test data for the $\PLDC$ model estimated from the remaining data. We compute the mean cross-entropy (log-likelihood) for the top-1 observation, the top-3 observations, and all observations, and compare them with the corresponding performance of PL and public belief. The results are presented in Figure~\ref{fig:KFCV}.  
\begin{figure}[htbp]
	\centering
	\begin{minipage}{0.32\linewidth}
		\centering
		\includegraphics[width=\linewidth,  trim={0.5cm 1.5cm 0.5cm 0.5cm}, clip]{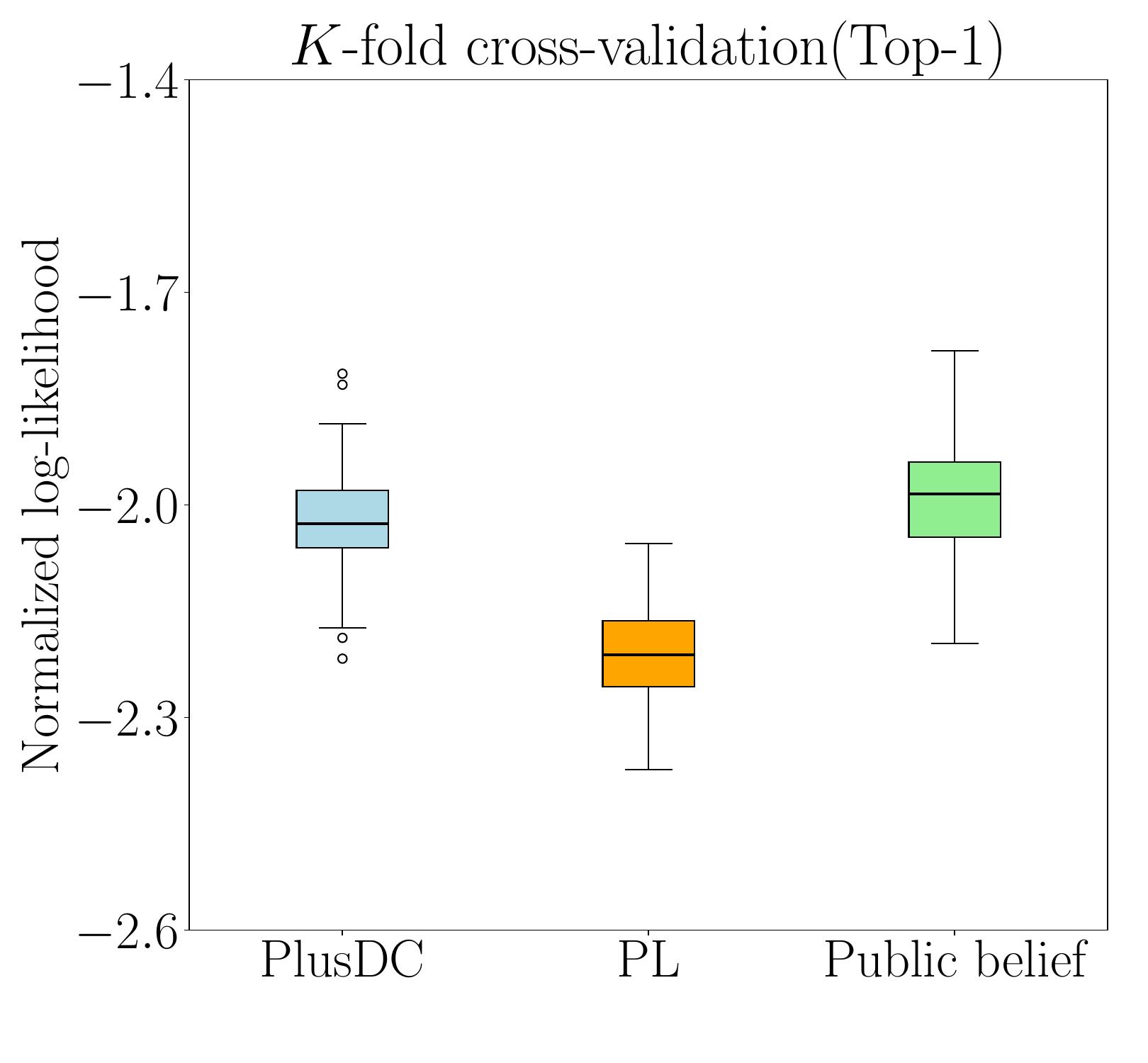}
		\caption*{(a)}
	\end{minipage}\hfill 
	\begin{minipage}{0.32\linewidth}
		\centering
		\includegraphics[width=\linewidth,   trim={0.5cm 1.5cm 0.5cm 0.5cm}, clip]{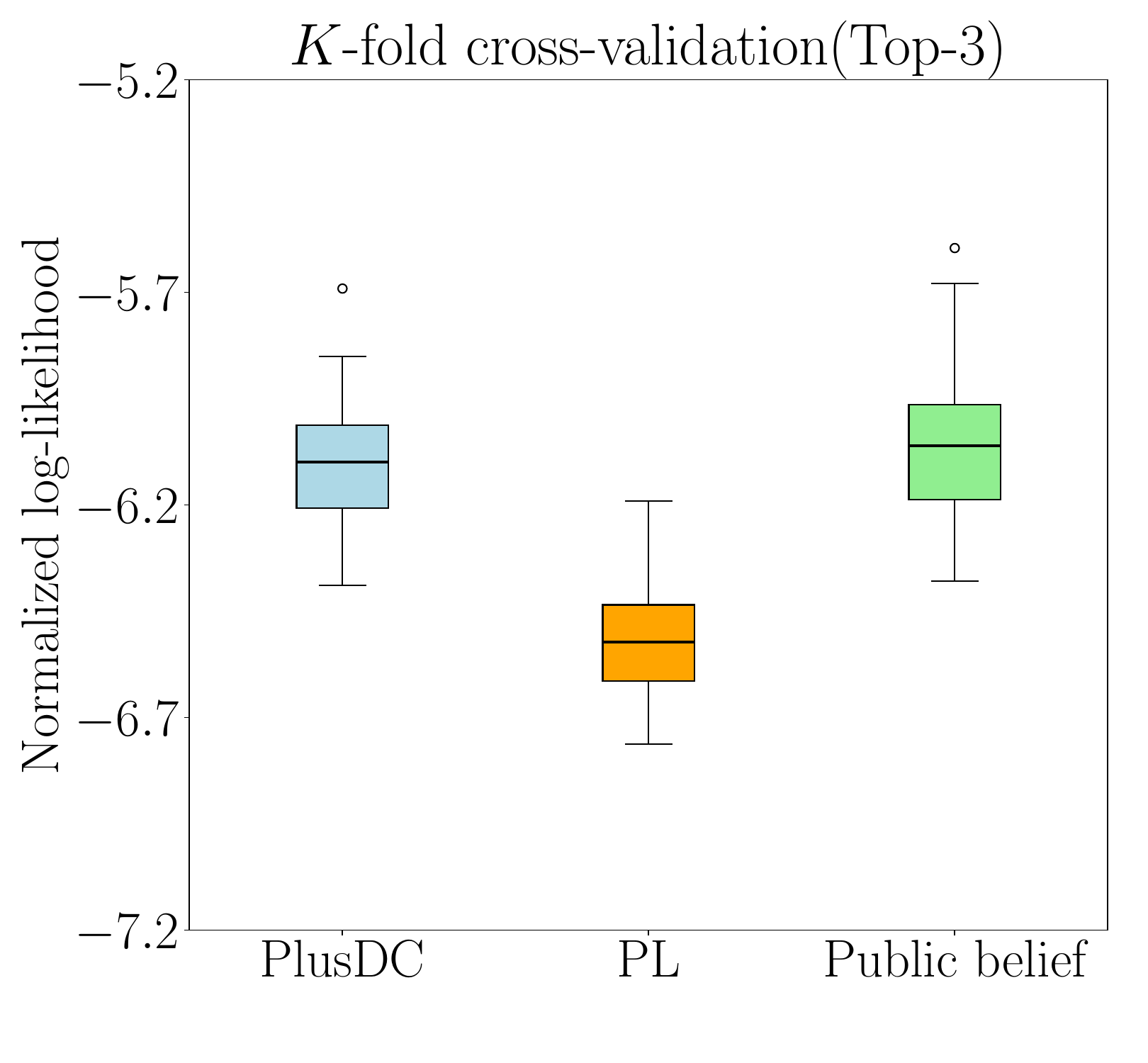}
		\caption*{(b)}
	\end{minipage}\hfill
	\begin{minipage}{0.32\linewidth}
		\centering
		\includegraphics[width=\linewidth,   trim={0.5cm 1.5cm 0.5cm 0.5cm}, clip]{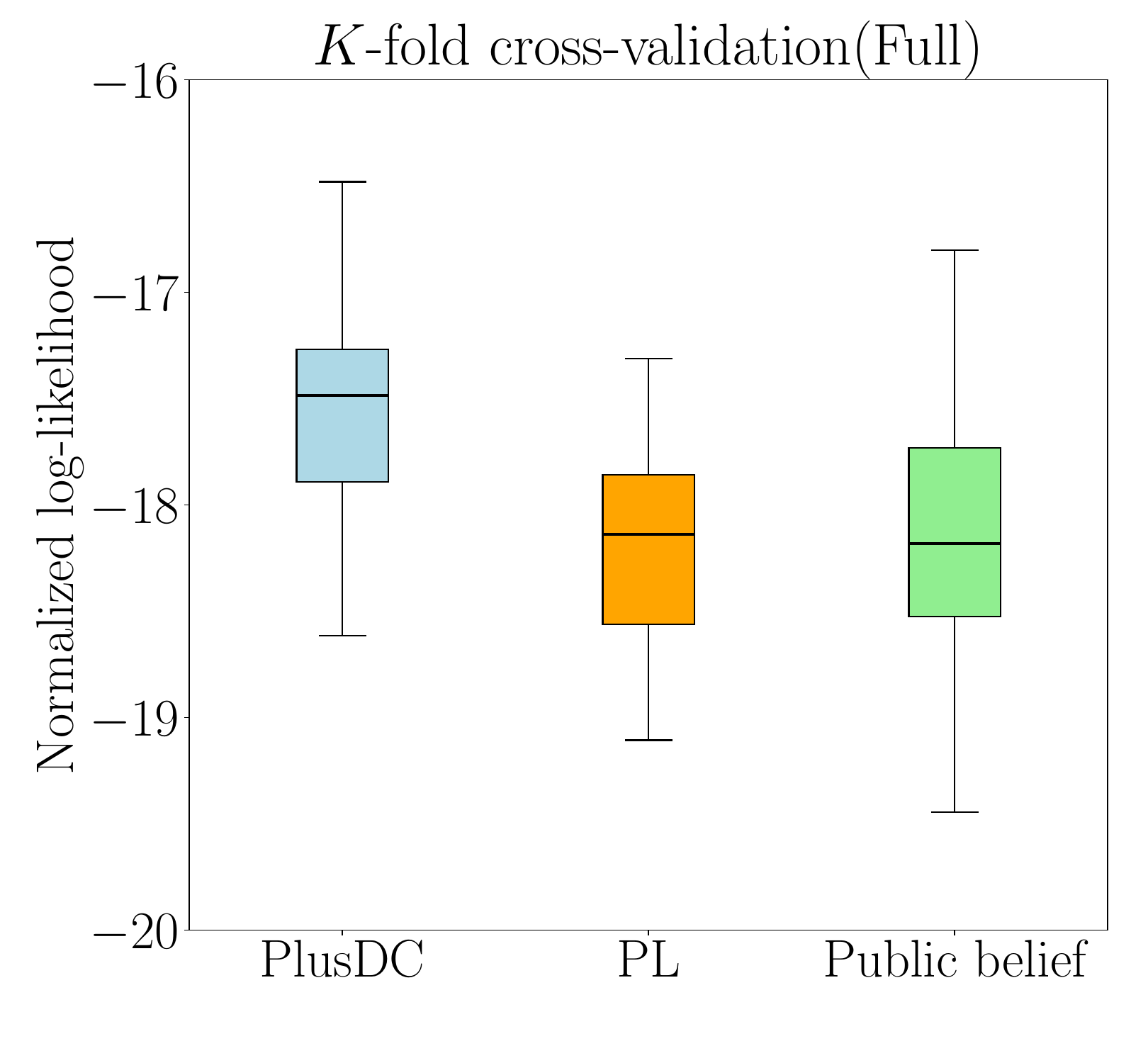}
		\caption*{(c)}
	\end{minipage}
	\caption{$\mathsf{k}$-fold cross-validation of $\PLDC$, PL, and Public belief. (a): Top-1 ranking; (b): Top-3 ranking; (c): Full ranking. }
	\label{fig:KFCV}
\end{figure}

In Figure~\ref{fig:KFCV}, $\PLDC$ consistently surpasses PL in all three scenarios, demonstrating evidence that integrating covariates enhances model prediction. Meanwhile, $\PLDC$ achieves comparable yet slightly lower performance compared to public belief in the top-1 and top-3 settings, but significantly outperforms it in predicting the full ranking. This outcome is not unexpected since the $\PLDC$ model is trained using the full likelihood.

\section{Model identifiability}\label{proof:identifiability}

We provide the proofs of Propositions~\ref{prop:1}--\ref{prop:curl} in Sections \ref{p:prop1}--\ref{p:prop2} respectively.

\subsection{Proof of Proposition~\ref{prop:1}}\label{p:prop1}
We establish the ``if'' part via a contradiction argument. Suppose that the $\PLDC$ model is not identifiable. Then, by definition, there exist $\t = [\u^\top, \v^\top]^\top\neq\t' = [(\u')^\top, (\v')^\top]^\top$ with $\bm 1^\top(\u-\u')=0$ such that they induce the same probability measures on each edge set $\EE$, which can be equivalently stated as $\Q^\top(\u-\u') + \K(\v-\v')=\bm 0$, or $\AA(\t-\t')=\bm 0$. Since $\bm a = (\bm 1^\top, \bm 0^\top)^\top\in\ker(\AA)$ and $\bm a^\top (\t-\t') = 0$, we have $\dim(\ker(\AA))\geq 2$. But this implies $\rank(\AA)\leq n+d-2$, which is a contradiction. The ``only if'' part is similar and thus omitted.  

\subsection{Proof of Proposition~\ref{prop:curl}}\label{p:prop2}

It suffices to check $\rank(\K) = d$ and $\range(\Q^\top)\cap\range(\K) = \{\bm 0\}$. For the former, by the definition of the curl operator in \eqref{curl}, there exists $\bm B\in\{\pm 1, 0\}^{d\times N}$ such that $\bm B \K = \bm T_{\triangle}$. As a result, $d=\rank(\bm T_{\triangle})\leq\rank(\K) \leq d$. For the latter, for any $\bm 0\neq\bm\beta\in\R^d$, the restriction of $\curl(\K\bm\beta)$ to $\mathcal T_{\triangle}$ is given by $\bm T_{\triangle}\bm\beta$, which is nonzero since $\bm T_{\triangle}$ has full column rank. Therefore, $\K\bm\beta\not\in\ker(\curl)$, hence $\range(\Q^\top)\cap\range(\K) = \{\bm 0\}$.

\section{Estimation using the MLE}\label{proof:emle} 

We provide the proof of Theorem~\ref{thm:existence} in Section \ref{p:thm1}. The proofs of Propositions \ref{prop:at_verify}, \ref{prop:equiv} and \ref{prop:mm} are presented in Sections \ref{p:prop4}--\ref{p:prop6} respectively.

\subsection{Proof of Theorem~\ref{thm:existence}}\label{p:thm1}
To reduce notational clutter, we introduce a few more compact notations below. 
For $i\in [N]$, $j\in [m_i]$, and $j\leq t< m_i$, denote 
\begin{align}
	Z_{(i,j,t)} = [\vec{\mathsf e}^{\top}_{\pi_i(t)\pi_i(j)} ,(X_{\T_i, \pi_i(t)}-X_{\T_i, \pi_i(j)})^\top]^\top\in\R^{n+d}, \label{myZZ}
\end{align}
where $\vec{\mathsf e}_{k\ell}$ denotes the difference between the $k$th and $\ell$th canonical basis in $\R^n$ for $k, \ell\in [n]$. Moreover, we define $\bm\Lambda_{ij} = \text{diag}(\bm g_{ij})-\bm g_{ij}\bm g_{ij}^\top$, where
\begin{align}
	\bm g_{ij} = \frac{1}{\sum_{t=j}^{m_i}\exp\{Z_{(i,j,t)}^\top\t\}}(\exp\{Z_{(i,j,j+1)}^\top\t\}, \ldots, \exp\{Z_{(i,j,m_i)}^\top\t\})^\top\in\R^{m_i-j},\label{mygg}
\end{align}
with $\text{diag}(\bm g_{ij})$ standing for the diagonal matrix with diagonal entries given by the vector $\bm g_{ij}$.
With these notations, Assumption \ref{ass:ue} can be restated as follows:
\begin{customthm}{2$^\mathsection$}\label{ass:ue+}
	Given $\HH = (\V, \EE)$, $\mathcal X = \{X_{\T, j}\}_{j\in \T, \T\in\EE}$, and the comparison outcomes $\{\pi_\T\}_{\T\in\EE}$, for any non-zero vector $\t=(\u^\top,\v^\top)^\top \in\R^{(n+d)}$ with $\bm{1}^\top \u = 0$, there exist $i\in[N]$ and $j<t\le m_i$ such that $Z^\top_{(i,j,t)}\t>0$.
\end{customthm}

\subsubsection*{\underline{Sufficiency}}

We first show that Assumption~\ref{ass:ue+} implies the unique existence of the MLE. 
Using the fact that $Z_{(i,j,j)}={\bf 0}$, we can write $l(\t)$ and its derivatives as  
\begin{align}
	l(\t) &= -\sum_{i\in [N]}\sum_{j\in [m_i-1]} \log\left(\sum_{t=j}^{m_i}\exp{Z_{(i,j,t)}^\top\t}\right),\nonumber\\
	\nabla l(\t) &= -\sum_{i\in [N]}\sum_{j\in [m_i-1]} \frac{ \sum^{m_i}_{t=j+1}\exp{Z^\top_{(i,j,t)}\t}Z_{(i,j,t)}}{1+\sum^{m_i}_{t=j+1}\exp{Z^\top_{(i,j,t)}\t}},\label{eq:gradient}\\
	\nabla^2 l(\t) &= -\sum_{i\in [N]}\sum_{j\in [m_i-1]}[Z_{(i, j, j+1)}, \ldots, Z_{(i,j,m_i)}]\bm\Lambda_{ij}[Z_{(i,j,j+1)}, \ldots, Z_{(i,j,m_i)}]^\top. \label{l:hess}
\end{align}
Since 
\begin{align}
	[\bm\Lambda_{ij}]_{kk}-\sum_{\ell\in [m_i-j], \ell\neq k}|[\bm\Lambda_{ij}]_{k\ell}| = \frac{\exp{Z_{(i,j,j+k)}^\top\t}}{\left(\sum_{t=j}^{m_i}\exp{Z_{(i,j,t)}^\top\t}\right)^2}>0,\label{l:hess1}
\end{align}
it follows from the Gershgorin circle theorem that $\bm\Lambda_{ij}$ is positive-definite. Consequently, the Hessian matrix is semi-negative definite and the log-likelihood function is concave. 

To establish the unique existence of the MLE, it suffices to show that on $\bm 1^\top\u = 0$, $-l(\t)$ is strictly convex and coercive, that is, $-l(\t)\to \infty$ if $\|\t\|_2\to\infty$, under Assumption~\ref{ass:ue+}.  For the former, for any $\widetilde{\t}=(\widetilde{\u}^\top, \widetilde{\v}^\top)^\top\neq\bm 0$ with $\bm 1^\top\widetilde{\u} = 0$, Assumption~\ref{ass:ue+} says there exist an edge $\T_{i}$ and $j<t\leq m_{i}$ such that $Z^\top_{(i,j,t)}\widetilde{\t}>0$. As a result, by the positivity of $\bm\Lambda_{ij}$, 	
\begin{align*}
	-\widetilde{\t}^\top\nabla^2 l(\t)\widetilde{\t} \geq \widetilde{\t}^\top[Z_{(i, j, j+1)}, \ldots, Z_{(i, j, m_{i})}]\bm\Lambda_{ij}[Z_{(i, j, j+1)}, \ldots, Z_{(i, j, m_{i})}]^\top\widetilde{\t}>0,
\end{align*}
proving the strict convexity of $-l(\t)$ on $\bm 1^\top\u = 0$. 

For coercivity, note that for any $\t\neq\bm 0$ with $\bm 1^\top\u = 0$, Assumption~\ref{ass:ue+} implies
\begin{align*}
	\mathW(\t):= \max_{i\in [N]}\max_{j\in[m_i]}\max_{j<t\le m_i} Z^\top_{(i,j,t)}\t>0.
\end{align*}
Since $\mathW(\t)$ is continuous in $\t$ and $\{\|\t\|_2=1: \bm 1^\top \u = 0\}$ is compact, we have  $\inf_{\|\t\|_2=1: \bm 1^\top \u = 0}\mathW(\t)>0$. Therefore, $\mathW(\t)$ is coercive on $\bm 1^\top\u = 0$ as a result of positive homogeneity ($\mathW(\mu\t) = \mu\mathW(\t)$ for $\mu>0$). The coercivity of $-l(\t)$ on $\bm 1^\top\u = 0$ follows by noting $-l(\t)\geq \mathW(\t)$. Indeed, for any tuple $(i, j, t)$ with $j<t\leq m_i$,  
\begin{align}
	-l(\t) \geq & \log(1+\sum_{ {\ell=j+1}}^{m_i}\exp\{Z_{(i,j,\ell)}^\top\t\})\geq Z_{(i,j,t)}^\top\t.
\end{align}
Taking the maximum over $(i, j, t)$ on the right-hand side yields $-l(\t)\geq \mathW(\t)$.

\subsubsection*{\underline{Necessity}}

For the necessity of Assumption~\ref{ass:ue+}, let $\tt$ be the unique MLE satisfying $\nabla l(\tt) = 0$. We show next that if there exists $\t_{0}\neq\bm 0$ with $\bm 1^\top\t_0 = 0$ such that $Z^\top_{(i,j,t)}\t_{0}\ge 0$ for all $i\in[N]$ and $j<t\le m_i$, then we will arrive at a contradiction.

We first consider the case where there exist $i_0$, $j_0$, $t_0$ such that $Z^\top_{(i_0,j_0,t_0)}\t_{0}>0$. 
Under such circumstances, 
\begin{align*}
	\left\{\nabla l(\tt)\right\}^\top \t_{0}&= -\sum_{i\in [N]}\sum_{j\in [m_i-1]} \frac{\sum^{m_i}_{t=j+1}\exp{Z_{(i,j,t)}^\top\tt}Z_{(i,j,t)}^\top\t_{0}}{1+\sum^{m_i}_{t=j+1}\exp{Z_{(i,j,t)}^\top\tt}}\\
	&\le-\frac{\exp{Z_{(i_0,j_0,t_0)}^\top\tt}Z_{(i_0,j_0,t_0)}^\top\t_{0}}{1+\sum^{m_i}_{t=j+1}\exp{Z_{(i,j,t)}^\top\tt}}<0.
\end{align*}

The only remaining case is when $Z^\top_{(i,j,t)}\t_{0} = 0$ for all $i\in [N], j<t\leq m_i$. From equation~\eqref{eq:gradient}, we can verify that
$\nabla l(\t+\t_{0})=\nabla l(\t)$. According to the Taylor expansion and the property of concave functions, $l(\tt+\t_{0})= l(\tt)=\max_{\t: \bm 1^\top\u = 0}l(\t)$, contradicting to the uniqueness of $\tt$. Hence, Assumption~\ref{ass:ue+} holds whenever the MLE uniquely exists.

\subsection{Proof of Proposition~\ref{prop:at_verify}}\label{p:prop4}

The proof consists of two steps. We first show a stronger assumption, that is, Assumption \ref{ass:ue_2} implies Assumption \ref{ass:ue+} (and also Assumption \ref{ass:ue} as shown in Section~\ref{p:thm1}). Then, we show that Assumption \ref{ass:ue_2} holds with probability at least $1-n^{-2}$ for all large $n$. 

\begin{customthm}{2$^*$}\label{ass:ue_2}
	For every partition of $\VV = \VV_1\cup \VV_2$ and non-zero vector $\v\in\R^d$, there exist $k_1\in \VV_1, k_2\in \VV_2, \T\supseteq\{k_1, k_2\}$ such that $k_1\succ k_2$ on $\T$ and $X^\top_{\T, k_1}\v< X_{\T, k_2}^\top\v$. 
\end{customthm}

\subsubsection*{\underline{Step I}: Assumption~\ref{ass:ue_2} implies Assumption~\ref{ass:ue+}}
Let $\t=(\u^\top,\v^\top)^\top\neq\bm 0$ with $\bm 1^\top\u = 0$. When $\u\neq 0$ and $\v\neq 0$, we can choose $\VV_1 = \{k\in[n]:u_k = \min_{t\in [n]}u_t\}$ and $\VV_2 = \VV\setminus\VV_1$. Since $\u\neq 0$ and $\bm{1}^\top \u = 0$, $\VV_1$ and $\VV_2$ are non-empty sets, and for any $k_1 \in\VV_1$ and $k_2 \in\VV_2$, $u_{k_1}<u_{k_2}$. Under Assumption~\ref{ass:ue_2}, there exist $k_1 \in \VV_1, k_2\in \VV_2, \T\supseteq\{k_1, k_2\}$ such that $i\succ j$ on $\T$ and $X^\top_{\T, k_1}\v< X_{\T, k_2}^\top\v$, which implies $u_{k_1}+X^\top_{\T, k_1}\v<u_{k_2}+X^\top_{\T,k_2}\v$. 
If $\u=0$ and $\v\neq 0$, we can use Assumption~\ref{ass:ue_2} directly to get the result. If $\v=0$ and $\u\neq 0$, we also have $u_{k_1}+X^\top_{\T,k_1}\v<u_{k_2}+X^\top_{\T,k_2}\v$ due to $u_{k_1}<u_{k_2}$.

\subsubsection*{\underline{Step II}: Assumption~\ref{ass:ue_2} holds with high probability}

The lower bound on the modified Cheeger constant implies that for all nonempty $U\subset \VV$, 
\begin{align}
	\frac{|\partial U|}{\min\{|U|, n-|U|\}\log n}\to\infty. \label{qiangzi}
\end{align}
Since the condition in Assumption~\ref{ass:ue_2} is scale-invariant in $\bm v$,  without loss of generality, fix $\v$ with $\|\bm v\|_\infty = 1$. For any $U\subset [n]$, define the following random variable $\Gamma(U; \v, \delta)$ counting the number of the pair of objects making Assumption~\ref{ass:ue_2} satisfied with margin at least $\delta>0$: 
\begin{align*}
	\Gamma(U; \v, \delta) = \sum_{\T\in\partial U}\mathbb I_{\left\{\text{there exist $i\in \T\cap U, j\in \T\cap U^\complement$, such that $i\succ j, (X_{\T, i}-X_{\T,j})^\top\bm v < -\delta$}\right\}},
\end{align*}
which is an independent sum of Bernoulli random variables.
Under Assumption~\ref{ass:unif-bdd} and \eqref{mass}, there exists an absolute constant $c_1>0$ (depending on $R$ and $c_0$) such that each of these Bernoulli random variables has mean at least $c_1$. By the Chernoff bound, for all large $n$,  
\begin{align*}
	\P\{\Gamma(U; \v, \delta)=0\} &= \P\left\{\Gamma(U; \v, \delta)-\E[\Gamma(U; \v, \delta)]\leq -\E[\Gamma(U; \v, \delta)]\right\}\\
	&\leq \exp{-\frac{\E[\Gamma(U; \v, \delta)]}{3}}\leq\exp{-\frac{c_1|\partial U|}{3}}\stackrel{\eqref{qiangzi}}{\leq} n^{-7\min\{|U|, n-|U|\}} .
\end{align*}
Taking a union bound over $U$ yields 
\begin{align}
	\P\left\{\text{$\Gamma(U; \v, \delta)\geq 1$ for all $U\subset \VV$}\right\}&\geq 1-\sum_{U\subset \VV: |U|\leq n/2}\P\{\Gamma(U; \v, \delta)=0\}\nonumber\\
	&\geq 1-\sum_{k=1}^{n/2}{n\choose k}n^{-6k}\nonumber\\
	&\geq 1-n^{-5}. \label{single-bdd}
\end{align}
Therefore, with probability at least $1-n^{-5}$, $\Gamma(U; \v, \delta)\geq 1$ for all partitions of $\VV = U\cup U^\complement$.

To extend the above result to arbitrary $\v$ with $\|\v\|_\infty=1$, we apply an approximation argument.
Let $\mathcal N(\delta/4R, \|\cdot \|_\infty)$ be a minimal $(\delta/4R)$-covering of $\{\v\in\R^d: \|\v\|_\infty=1\}$ in $\|\cdot \|_\infty$. 
Since $|\mathcal N(\delta/4R, \|\cdot\|_\infty)|$ is finite, $|\mathcal N(\delta/4R, \|\cdot\|_\infty)|\leq n$ for all large $n$. 
Taking a union bound for \eqref{single-bdd} over all $\v\in\mathcal N(\delta/4R, \|\cdot\|_\infty)$ yields that 
\begin{align}
	\P\left\{\text{$\Gamma(U; \v, \delta)\geq 1$ for all $U\subset \VV$ and $\v\in \mathcal N(\delta/4R, \|\cdot \|_\infty)$}\right\}&\geq 1-n^{-4}.\label{coverbound}
\end{align}

We now apply the triangle inequality to conclude that the event in \eqref{coverbound} implies Assumption~\ref{ass:ue_2}.
Indeed, for any $\v$ with $\|\v\|_\infty=1$, denote by $\v'$ its closest point in $\mathcal N(\delta/4R, \|\cdot\|_\infty)$ with respect to $\|\cdot\|_\infty$. It follows from the triangle inequality that for all $U\subset \VV$, $\Gamma(U; \v, \delta/2)\geq \Gamma(U; \v', \delta)>0$, implying Assumption~\ref{ass:ue_2}. 

Since the probability in \eqref{coverbound} concerns the randomness of both covariates and comparison outcomes, to finish the proof, we apply the conditional probability formula to obtain
\begin{align*}
	1-n^{-4}&\leq\P\left\{\text{Assumption~\ref{ass:ue_2} holds}\right\}\leq \P\left\{\text{Assumption~\ref{ass:ue} holds}\right\}\\
	&= 1-\P\left\{\text{Assumption~\ref{ass:ue} does not hold}\right\}\\
	&\leq 1-\P\left\{\text{$\HH$ and $\mathcal X$ are not AEP}\right\}\P\left\{\text{Assumption~\ref{ass:ue} does not hold}\mid \text{$\HH$ and $\mathcal X$ are not AEP}\right\}\\
	&\leq 1 - \P\left\{\text{$\HH$ and $\mathcal X$ are not AEP}\right\}n^{-2},
\end{align*}
which implies the desired result. 


\subsection{Proof of Proposition~\ref{prop:equiv}}

Let $\mathscr{C}_1$ and $\mathscr{C}_2$ denote the constrained parameter spaces in the PL and CARE models, respectively.
Under the assumptions in Proposition~\ref{prop:equiv}, $\dim(\mathscr{C}_1) = \dim(\mathscr{C}_2) = n-1$. 
Consider the linear map $\phi: \R^n\to\R^{n+d}$ defined as follows:
\begin{align*}
	\phi(\u) &= \begin{pmatrix}
		\phi_1(\u)\\
		\phi_2(\u)
	\end{pmatrix} = \begin{pmatrix}
		\u - \left(\bm I-\frac{\bm 1\bm 1^\top}{n}\right)\bm Z\left(\bm Z^\top\bm Z - \frac{\bm Z^\top\bm 1\bm 1^\top\bm Z}{n}\right)^{-1}\bm Z^\top\u\\
		\left(\bm Z^\top\bm Z - \frac{\bm Z^\top\bm 1\bm 1^\top\bm Z}{n}\right)^{-1}\bm Z^\top\u
	\end{pmatrix}.
\end{align*} 
The desired equivalence result in Proposition~\ref{prop:equiv} would follow if the following hold:
\begin{itemize}
	\item [(1)]$\phi$ is a bijection between $\mathscr{C}_1$ and $\mathscr{C}_2$;
	\item [(2)]$l(\phi(\u)) = l(\u, \bm 0)$.
\end{itemize}
For (1), note that when $\bm 1^\top\u = 0$, 
\begin{align*}
	\bm 1^\top \phi_1(\u) &= \bm 1^\top \u -\left[\bm 1^\top \left(\bm I-\frac{\bm 1\bm 1^\top}{n}\right)\right]\bm Z\left(\bm Z^\top\bm Z - \frac{\bm Z^\top\bm 1\bm 1^\top\bm Z}{n}\right)^{-1}\bm Z^\top\u = 0, \\
	\bm Z^\top\phi_1(\u) & = \bm Z^\top \u - \left(\bm Z^\top\bm Z - \frac{\bm Z^\top\bm 1\bm 1^\top\bm Z}{n}\right)\left(\bm Z^\top\bm Z - \frac{\bm Z^\top\bm 1\bm 1^\top\bm Z}{n}\right)^{-1}\bm Z^\top\u = \bm 0.
\end{align*}
Therefore, $\phi(\mathscr{C}_1)\subseteq \mathscr{C}_2$. Since $\phi^\dagger: (\u, \v)\mapsto \u + \bm Z\v - n^{-1}\bm 1\bm 1^\top\bm Z\v$ satisfies $\phi^\dagger\circ\phi = \text{id}$ and $\dim(\mathscr{C}_1) = \dim(\mathscr{C}_2)$, $\phi$ is a bijection between $\mathscr{C}_1$ and $\mathscr{C}_2$ with pseudoinverse  $\phi^\dagger$. For (2), it suffices to check for $\phi^\dagger$: $l(\u, \v) = l(\u+\bm Z\v, \bm 0) = l(\phi^\dagger(\u, \v), \bm 0)$.

\subsection{Proof of Proposition~\ref{prop:mm}}\label{p:prop6}

We first provide the explicit form of ${\u}^{(\tauf,\rhof+1)} $, which is computed by solving $\nabla_\u Q(\u \mid\u^{(\tauf,\rhof)}, \v^{(\tauf)})=\bm 0$ with centering. Specifically, denote the precentered update $\widetilde{\u}^{(\tauf, \rhof+1)}$ as 
\begin{align*}
	\widetilde{u}_k^{(\tauf,\rhof+1)} = \log\{\deg(k)\} - \log\left[\sum_{i:k\in \T_i}\sum_{j\in [r_i(k)]}\frac{\exp{X^\top_{\T_i,k}\v^{(\tauf)}}}{\sum_{t=j}^{m_i}\exp{s_{it}(\u^{(\tauf,\rhof)}, \v^{(\tauf)})}}\right] 
	\quad \quad \quad   \ k \in [n],
\end{align*}
so that $${\u}^{(\tauf,\rhof+1)} = \widetilde{\u}^{(\tauf,\rhof+1)} - \frac{1}{n}\bm 1\bm 1^\top \widetilde{\u}^{(\tauf,\rhof+1)}.$$
We first show that $\{\u^{(\tauf,\rhof)}\}_{\rhof\in\N}$ are uniformly bounded. To see this, note that
\begin{align*}
	&	l(\u^{(\tauf,\rhof+1)}, \v^{(\tauf)})=l(\widetilde{\u}^{(\tauf,\rhof+1)}, \v^{(\tauf)})\\
	&\geq Q(\widetilde{\u}^{(\tauf,\rhof+1)} \mid\u^{(\tauf,\rhof)} , \v^{(\tauf)})\geq Q({\u}^{(\tauf,\rhof)} \mid\u^{(\tauf,\rhof)} , \v^{(\tauf)}) = l(\u^{(\tauf,\rhof)}, \v^{(\tauf)}),
\end{align*}
so that $\{l(\u^{(\tauf,\rhof)}, \v^{(\tauf)})\}_{\rhof}$ are bounded from below. As was shown in Section~\ref{proof:emle}, Assumption~\ref{ass:ue} implies the coercivity of $-l(\t)$ on $\bm 1^\top\u = 0$. Consequently, $\{\u^{(\tauf,\rhof)}\}_{\rhof\in\N}$ is uniformly bounded.  

Now take a large $C>0$ so that $\max_{\rhof\in\N}\|\u^{(\tauf,\rhof)}\|_\infty\leq C$. 
Meanwhile, it follows from the direct computation that there exists $\sigma_1, \sigma_2>0$ (which may depend on $n$) such that for all $\u$ with $\|\u\|_\infty\leq C$ and $\rhof\in\N$, 
\begin{align}
	-\nabla_\u^2 l(\u,  \v^{(\tauf)})&\succ \sigma_1 \bm I, \label{sc2}\\
	-\nabla_\u^2Q(\u\mid\u^{(\tauf,\rhof)} , \v^{(\tauf)})&\prec \sigma_2\bm I\label{sc1}.
\end{align}
Condition \eqref{sc2} concerns strong concavity of $l(\u, \v^{(\tauf)})$ on $\|\u\|_\infty\leq C$ and implies the following Polyak--Lojasiewicz inequality: 
\begin{align}
	l(\u^{(\tauf+1)},  \v^{(\tauf)})-l(\u^{(\tauf, \rhof)},  \v^{(\tauf)})&\leq \frac{1}{2\sigma_1}\|\nabla_\u l(\u^{(\tauf, \rhof)},  \v^{(\tauf)})\|^2_2. \label{tll1}
\end{align}
Condition \eqref{sc1} is an $\sigma_2$-smoothness condition for $Q(\u \mid \u^{(\tauf,\rhof)} , \v^{(\tauf)})$ on $\|\u\|_\infty\leq C$ which implies
\begin{align}
	l(\u^{(\tauf, \rhof+1)},  \v^{(\tauf)})-l(\u^{(\tauf, \rhof)},  \v^{(\tauf)})&=l(\widetilde{\u}^{(\tauf, \rhof+1)},  \v^{(\tauf)})-l(\u^{(\tauf, \rhof)},  \v^{(\tauf)})\nonumber\\
	&\geq Q(\widetilde{\u}^{(\tauf, \rhof+1)} |\u^{(\tauf,\rhof)} , \v^{(\tauf)})-Q(\u^{(\tauf, \rhof)} \mid\u^{(\tauf,\rhof)} , \v^{(\tauf)})\nonumber\\
	&\geq \frac{\sigma_2}{2}\|\nabla_\u Q(\u^{(\tauf, \rhof)} \mid\u^{(\tauf,\rhof)} , \v^{(\tauf)})\|^2_2\nonumber\\
	& = \frac{\sigma_2}{2}\|\nabla_\u l(\u^{(\tauf, \rhof)},  \v^{(\tauf)})\|^2_2,\label{tll2}
\end{align}
where the second inequality follows from a standard gradient descent estimate \cite[Eq. (3.5)]{bubeck2015convex}. 
Combining \eqref{tll1} and \eqref{tll2} together yields
\begin{align*}
	l(\u^{(\tauf+1)},  \v^{(\tauf)})-l(\u^{(\tauf, \rhof+1)},  \v^{(\tauf)})&\leq \left(1-\frac{\sigma_2}{4\sigma_1}\right)\left(l(\u^{(\tauf+1)},  \v^{(\tauf)})-l(\u^{(\tauf, \rhof)},  \v^{(\tauf)})\right)\\
	&\leq \left(1-\frac{\sigma_2}{4\sigma_1}\right)^{\rhof+1}\left(l(\u^{(\tauf+1)},  \v^{(\tauf)})-l(\u^{(\tauf)},  \v^{(\tauf)})\right).
\end{align*}
Therefore, $l(\u^{(\tauf, \rhof+1)},  \v^{(\tauf)})\to l(\u^{(\tauf+1)},  \v^{(\tauf)})$ as $\rhof\to\infty$. Combined with condition \eqref{sc2}, this implies the desired result \eqref{nt}. 

\section{Uniform consistency}\label{sec:mainproof}

We provide the proof of Theorem~\ref{thm:main} in Section \ref{p:thm2} while its supporting lemmas are demonstrated in Section \ref{sec:hanproof}. In Sections \ref{p:prop5}-\ref{p:prop8}, we prove Propositions \ref{prop:1+}-\ref{prop:rg} respectively. In Section \ref{p:thm3}, we prove Theorem \ref{thm:randdesign}.

It is worth mentioning that \cite[Remark 5.2]{singh2024graphical} briefly discusses the case of infinite objects with covariates in the context of cardinal paired comparisons. Due to the least-squares formulation and the asymptotic orthogonality between the incidence matrix and the covariate matrix, the vector of merits can be consistently estimated without the covariates. In contrast, Theorem~\ref{thm:main} holds under more general assumptions, encompassing cases where the covariates are non-random and not asymptotically orthogonal to the incidence matrix. Our analysis significantly differs from theirs due to the interplay between the global parameter $\v$ and the local parameter $\bu$, which is further complicated by the lack of an explicit solution. 

Before proving the uniform consistency result, we introduce an additional notation. For a smooth function $f(\bx): \mathbb{R}^n \to \mathbb{R}$ and $k_1, \ldots, k_t \in [n]$, we let $\partial_{k_1\ldots k_t} f$ denote the $t$th order partial derivative of $f$ with respect to the $k_1, \ldots, k_t$ components. 

\subsection{Proof of Theorem~\ref{thm:main}}\label{p:thm2}

The unique existence of the MLE holds a.s. follows from Assumption~\ref{ass:aep} and the Borel--Cantelli lemma. Therefore, we focus on the uniform consistency part only, which consists of two main stages. First, we will show that $\|\widehat{\t}-\t^*\|_\infty=o_p(1)$. This stage relies on an error-comparison lemma between $\|\uu-\u^*\|_\infty$ and $\|\vv-\v^*\|_\infty$ (Lemma~\ref{lm:han}) and an empirical process type estimate. The resulting convergence rates for both parameters in this stage are suboptimal. To address this, we perform additional asymptotic analysis to boost the convergence rates closer to optimal. This step requires that both $\uu$ and $\vv$ are close to their respective optima which is guaranteed by the results in the first step. The following elementary lemma will be used in the analysis. 

\begin{lemma}\label{lm:ce}
	Let $\bm p^* = (p^*_1, \ldots, p^*_k)^\top\in\R^k$ be a probability vector with $\mu:=\min_{i\in [k]}p^*_i>0$. Define the cross-entropy function with respect to $\bm p^*$ as $F(\bm p; \bm p^*) = -\sum_{i\in [k]}p^*_i\log p_i$ where $\bm p = (p_1, \ldots, p_k)^\top\in\R^k_+$. Then $F$ is $\mu$-strongly convex on the simplex $\{\bm p\in\R_+^k: \|\bm p\|_1 =1\}$. 
\end{lemma}
\begin{proof}[Proof of Lemma~\ref{lm:ce}]
	The proof follows from a direct computation. For any $\bm p\in\R_+^k$ with $\|\bm p\|_1=1$, $\nabla^2 F(\bm p; \bm p^*) = \mathrm{diag}([p_1^*/p_1^2, \ldots, p_k^*/p_k^2]^\top)\succeq\mu\bm I$. 
\end{proof}

We will also need a technical lemma that upper bounds $\|\uu-\u^*\|_\infty$ by $\|\vv-\v^*\|_\infty$.

\begin{lemma}[Upper bound]\label{lm:han}
	Under Assumptions~\ref{ass:edgesize}--\ref{ass:topology}, there exists an absolute constant $\lambda>0$ such that 
	\begin{equation}\label{han}
		\|\uu - \u^*\|_\infty\lesssim \diam(\A_{\HH_n}(\lambda))\cdot\max\left\{\sqrt{\frac{\log n}{h_{\HH_n}}}, \|\vv-\v^*\|_\infty\right\}.
	\end{equation}
\end{lemma}

Now, we provide the proof of Theorem \ref{thm:main} as follows, which consists of two stages.

\subsubsection*{\underline{Stage I}: \texorpdfstring{$\|\widehat{\t}-\t^*\|_\infty = o_p(1)$}{}}
Assumption~\ref{ass:ue} ensures that the MLE $\tt$ uniquely exists. 
We first provide a lower bound of $\|\uu-\u^*\|_\infty$ via $\|\vv-\v^*\|_\infty$.
For $i\in [N]$ and $T_i = (j_{i1}, \ldots, j_{im_i})$ with $1\leq j_{i1}<\cdots<j_{im_i}\leq n$, define
\begin{align*}
	&\Delta {\bm s}^*_i = (s^*_{ir_i(j_{i2})}-s^*_{ir_i(j_{i1})}, \ldots, s^*_{ir_i(j_{im_i})}-s^*_{ir_i(j_{i1})})^\top & s^*_{ij}& = s_{\pi_i(j)}(\t^*; T_i)= u^*_{\pi_i(j)} + X_{T_i, \pi_i(j)}^\top\v^*,\\
	&\Delta {\bm s}_i = (s_{ir_i(j_{i2})}-s_{ir_i(j_{i1})}, \ldots, s_{ir_i(j_{im_i})}-s_{ir_i(j_{i1})})^\top&s_{ij} &  = s_{\pi_i(j)}(\t; T_i) = u_{\pi_i(j)} + X_{T_i, \pi_i(j)}^\top\v,
\end{align*}
where $r_i(\cdot)$ represents the rank of an element on $T_i$. 
Note that both $\Delta {\bm s}^*_i$ and $\Delta {\bm s}_i$ are independent of the comparison outcome $\pi_i$ as the difference is taken based on the ordered tuple $(j_{i1}, \ldots, j_{im_i})$. With these notions, we may write the {\it normalized} likelihood function and its population version as follows:
\begin{align}
	l(\t) &= \frac{1}{N}\sum_{i\in [N]}\sum_{j\in [m_i]}\left[s_{i j}-\log(\sum_{t=j}^{m_i}\exp{s_{it}})\right],\nonumber\\
	\bar{l}(\t) & = \frac{1}{N}\sum_{i\in [N]}\sum_{\pi\in\S(T_i)}\P_{\t^*}\{\pi |~T_i\}\log \P_{\t}\{\pi |~T_i\}\label{cetp},
\end{align} 
where $\P_\t\{\pi |~T\}$ is defined in \eqref{eq:PLDC}. 
Here we have slightly abused the notation $l$ to denote {\it normalized log-likelihood} for convenience. 

Let $\B = \{\t\in\R^{n+d}: \|\t-\t^*\|_\infty\leq 1\}$. As a first step, we show that $l$ and $\bar{l}$ are close to each other on $\B$.  
Under Assumptions~\ref{ass:edgesize} and \ref{ass:unif-bdd}, let $c_i$, $i\geq 1$ denote the absolute constants depending only on $\mM$ and $R$.
We have the following consequences:
\begin{itemize}
	\item Summands of $l(\t)$ are independent random variables uniformly bounded by some constant $c_1>0$ for all $\t\in\B$;
	\item There exists constant $c_2>0$ for all $\t\in\B$,
	\begin{align}
		\max_{i\in [N]}\max_{j\in [m_i]}|s_{ij}|\leq c_2;\label{pangziya}
	\end{align}
	\item Both $l$ and $\bar{l}$ are $c_3$-Lipschitz continuous on $\B$ with respect to $\|\cdot\|_\infty$ for some $c_3>0$. (This step requires the normalization in $l$, otherwise the Lipschitz constant will depend on $N$.)
\end{itemize}
Therefore, for fixed $\t\in\B$, by Hoeffding's inequality, for any $\ttt>0$, 
\begin{align*}
	\P\left\{|l(\t)-\bar{l}(\t)|>\ttt\right\}=\P\left\{|l(\t)-\E[l(\t)]|>\ttt\right\}\leq 2\exp{-\frac{N\ttt^2}{2c^2_1}}. 
\end{align*}
Meanwhile, the $(\ttt/c_3)$-covering number of $\B$ under $\|\cdot\|_\infty$ is bounded by $(c_4/\ttt)^{n+d}\leq (c_4/\ttt)^{2n}$ for some $c_4>0$. Taking $\ttt = c_5\sqrt{(n\log n)/N}$ and applying a union bound yield that, with probability at least $1-n^{-3}$, 
\begin{align}
	\sup_{\t\in\B}|l(\t)-\bar{l}(\t)|\leq 2c_5\sqrt{\frac{n\log n}{N}}.\label{sup:bdd}
\end{align}
Note that $N\geq ({nh_{\HH_n}}/{\mM})$. Under Assumption~\ref{ass:topology}, $n\log n = o(N).$
Therefore $\sup_{\t\in\B}|l(\t)-\bar{l}(\t)| \to 0$ as $n \to \infty$ almost surely. This suggests we may analyze the analytic properties of $\bar{l}$ rather than those of $l$. 

To analyze the population log-likelihood, observe from \eqref{cetp} that $-\bar{l}$ is an average of cross-entropy functions of discrete measures supported on at most $(M!)$ elements. In particular, there exists an absolute constant $c_6>0$ such that $\min_{i\in [N]}\min_{\pi\in\S(\T_i)}\P_{\t^*}\{\pi\mid \T\}>c_6$. Since $\nabla \bar{l}(\t^*)=\bm 0$, it follows from the mean-value theorem and Lemma~\ref{lm:ce} that, for all large $n$,  
\begin{align}
	\bar{l}(\t^*)-\bar{l}(\t) &\geq \frac{c_6}{2N}\sum_{i\in [N]}\sum_{\pi\in\S(\T_i)}(\P_{\t}\{\pi\mid \T_i\}-\P_{\t^*}\{\pi\mid \T_i\})^2&\text{(Lemma~\ref{lm:ce})}&\nonumber\\
	&\gtrsim \frac{c_6}{N}\sum_{i\in [N]}\|\Delta {\bm s}^*_i  - \Delta {\bm s}_i \|_2^2&\text{($**$)}&\nonumber\\
	& = \frac{c_6}{N} \|\bm Q_n^\top (\u-\u^*) + \bm K_n (\v-\v^*)\|^2_2\nonumber\\
	&\geq \frac{c_6\e}{N}( \|\bm Q_n^\top (\u-\u^*)\|_2^2 + \|\bm K_n (\v-\v^*)\|^2_2)&\text{(Assumption~\ref{ass:cosine})}&\nonumber\\
	&\geq c_{7}\|\v-\v^*\|_2^2.&\text{(Assumption~\ref{ass:minK})}&\label{low:bdd}
\end{align}
A detailed derivation of the step ($**$) is as follows: 
\begin{align*}
	&\frac{1}{N}\sum_{i\in [N]}\sum_{\pi\in\S(\T_i)}(\P_{\t}\{\pi\mid \T_i\}-\P_{\t^*}\{\pi\mid \T_i\})^2\\
	\geq&\ \frac{1}{MN}\sum_{i\in [N]}m_i\sum_{\pi\in\S(\T_i)}(\P_{\t}\{\pi\mid \T_i\}-\P_{\t^*}\{\pi\mid \T_i\})^2\ \ \ \ \ \  \text{(Assumption~\ref{ass:edgesize})}\\
	\geq&\ \frac{1}{MN}\sum_{i\in [N]}\sum_{\ell=2}^{m_i}\sum_{\substack{\pi\in\S(T_i)\\ r_\pi(j_{i1})<r_\pi(j_{i\ell})}}(\P_{\t}\{\pi\mid \T_i\}-\P_{\t^*}\{\pi\mid \T_i\})^2\\
	\gtrsim&\ \frac{1}{N}\sum_{i\in [N]}\sum_{\ell=2}^{m_i}\left(\sum_{\substack{\pi\in\S(\T_i)\\ r_\pi(j_{i1})<r_\pi(j_{i\ell})}}\P_{\t}\{\pi\mid \T_i\}-\sum_{\substack{\pi\in\S(\T_i)\\  r_\pi(j_{i1})<r_\pi(j_{i\ell})}}\P_{\t^*}\{\pi\mid \T_i\}\right)^2 \ \ \ \ \ \ \text{(Cauchy--Schwarz)} \\
	=&\ \frac{1}{N}\sum_{i\in [N]}\sum_{\ell=2}^{m_i}\left(\P_{\t}\{j_{i1}\succ j_{i\ell}\mid \{j_{i1}, j_{i\ell}\}_{\T_i}\}-\P_{\t^*}\{j_{i1}\succ j_{i\ell}\mid \{j_{i1}, j_{i\ell}\}_{\T_i}\}\right)^2 \ \ \ \ \text{(Luce's choice axiom)} \\
	=&\ \frac{1}{N}\sum_{i\in [N]}\sum_{\ell=2}^{m_i}\left(\frac{1}{1 + \exp{s_{ir_{i}(j_{i\ell})}-s_{ir_{i}(j_{i1})}}}-\frac{1}{1 + \exp{s^*_{ir_{i}(j_{i\ell})}-s^*_{ir_{i}(j_{i1})}}}\right)^2\\
	\gtrsim&\ \frac{c_6}{N}\sum_{i\in [N]}\|\Delta {\bm s}^*_i  - \Delta {\bm s}_i \|_2^2, \ \ \ \ \ \ \text{(\eqref{pangziya} + mean-value theorem)}
\end{align*}
where $\{j_{i1}, j_{i\ell}\}_{\T_i}$ represents a hypothetical edge containing $j_{i1}$ and $j_{i\ell}$ only and with the same covariates as $\T_i$.  
In particular, the fourth step follows from Luce's choice axiom which states that the probability of $j_{i1}$ beating $j_{i\ell}$ is independent of the presence or absence of other objects under the same covariates. 

\begin{figure}[t]
	\centering 
	\includegraphics[width=0.55\linewidth, trim={0.1cm 0cm 1cm 0cm},clip]{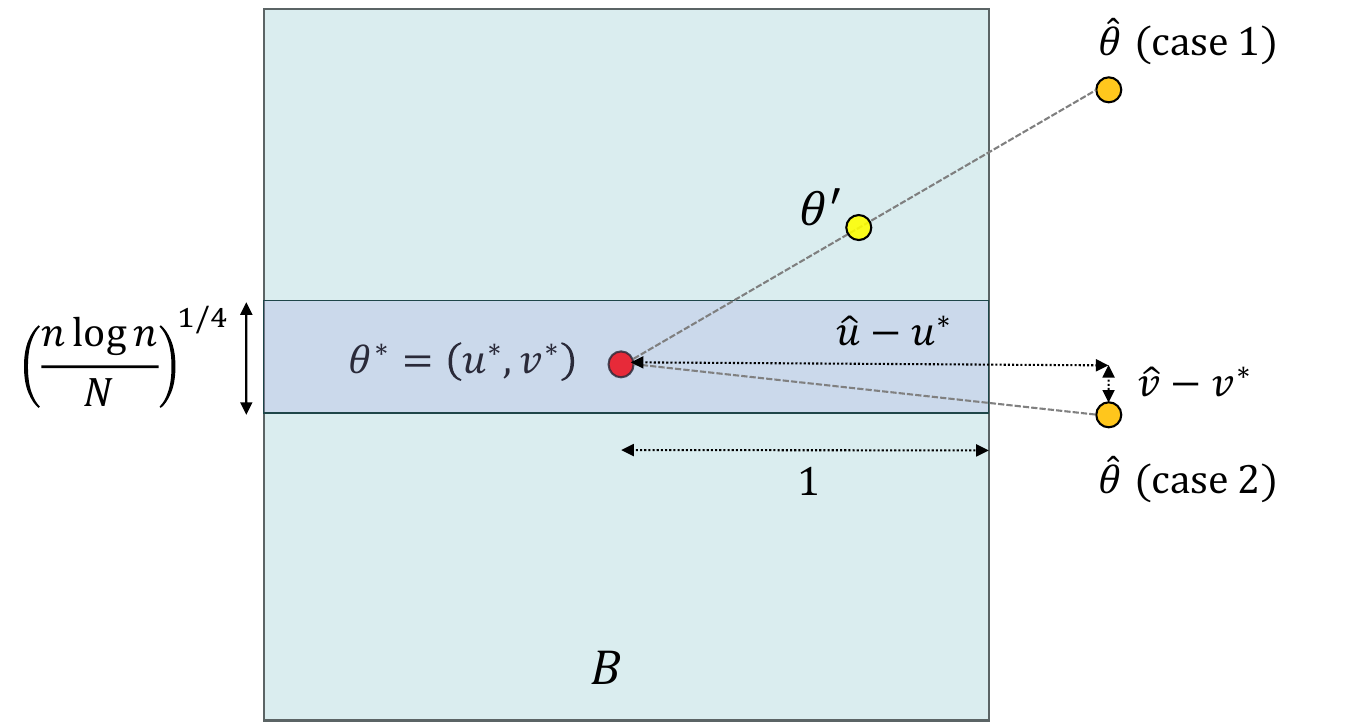}
	\caption{Illustration of the two possible cases when $\widehat{\t}\in\B^\complement$.} \label{fig:1}
\end{figure}

If $\widehat{\t}\in\B$ and $\|\widehat{\v}-\v^*\|^2_2> 4c_5c^{-1}_{7}\sqrt{(n\log n)/N}$, then
\begin{align}
	l(\widehat{\t})\stackrel{\eqref{sup:bdd}}{\leq} \bar{l}(\widehat{\t}) + 2c_5\sqrt{\frac{n\log n}{N}}\stackrel{\eqref{low:bdd}}{\leq}\bar{l}(\t^*) - 2c_5\sqrt{\frac{n\log n}{N}}<l(\t^*),\label{middle}
\end{align}
which is a contradiction. 

If $\widehat{\t}\in\B^\complement$ and  $\|\widehat{\v}-\v^*\|_2^{ {2} }> 4c_5c_{7}^{-1}\sqrt{ (n\log n)/N}$, then consider the line segment $[\widehat{\t}, \t^*]$ between $\widehat{\t}$ and $\t^*$. There are two possible cases, see Figure \ref{fig:1}.
The first is when $[\widehat{\t}, \t^*]\cap\{\t\in\B: \|\v-\v^*\|_2^2>4c_5c_{7}^{ -1} \sqrt{ (n\log n)/N}\} = \t'\neq\emptyset$. But this would lead to $l(\t^*)>l(\t')$ and $l(\widehat{\t})>l(\t')$ by a similar calculation as \eqref{middle}, contradicting the fact that $l(\cdot)$ is concave. 

Hence, we must have $[\widehat{\t}, \t^*]\cap\{\t\in\B: \|\v-\v^*\|_2^{2 }>4c_5c_{7}^{-1}\sqrt{ (n\log n)/N}\}=\emptyset$.
However, this would imply that 
\begin{align*}
	\|\widehat{\v}-\v^*\|_\infty&\lesssim \left(\frac{ n\log n}{N}\right)^{1/4}\|\widehat{\u}-\u^*\|_\infty.
\end{align*}
Meanwhile, via the upper bound given in Lemma~\ref{lm:han}, we obtain
\begin{align*}
	\|\widehat{\v}-\v^*\|_\infty	&\lesssim  \left(\frac{\log n}{h_{\HH_n}}\right)^{1/4}\diam(\A_{\HH_n}(\lambda))\cdot\max\left\{\sqrt{\frac{\log n}{h_{\HH_n}}},\|\widehat{\v}-\v^*\|_\infty\right\}\\
	&\lesssim \left(\frac{\log n}{h_{\HH_n}}\right)^{1/4}\diam(\A_{\HH_n}(\lambda))\cdot\|\widehat{\v}-\v^*\|_\infty\\
	&< \frac{1}{2}\|\widehat{\v}-\v^*\|_\infty.&\text{(Assumption~\ref{ass:topology})}
\end{align*}
This is again a contradiction. Hence, we have shown that $\|\widehat{\v}-\v^*\|_2\lesssim \{(\log n)/h_{\HH_n}\}^{1/4} = o(1)$. 
The error bound for $\widehat{\u}$ follows from \eqref{han} and Assumption~\ref{ass:topology}. Therefore, we have proved that $\|\widehat{\t}-\t^*\|_\infty = o_p(1)$ with respective convergence rates for $\uu$ and $\vv$ as 
\begin{align}
	&\|\widehat{\u}-\u^*\|_\infty\lesssim \diam(\mathcal A_{\HH_n}(\lambda))\left(\frac{\log n}{h_{\HH_n}}\right)^{1/4}, \ \ \ \|\widehat{\v}-\v^*\|_\infty\lesssim \left(\frac{\log n}{h_{\HH_n}}\right)^{1/4}.\label{sub-rates}
\end{align} 

\subsubsection*{\underline{Stage II}: Boosting convergence rates}

In the first step, we have proved that $\|\widehat{\t}-\t^*\|_\infty = o_p(1)$. The proximity between $\tt$ and $\t^*$ allows us to apply a Taylor expansion argument to improve the convergence rate exponent $1/4$ to $1/2$ in \eqref{sub-rates}.  
Without loss of generality, we assume $\|\vv-\v^*\|_\infty\geq\sqrt{(\log n)/h_{\HH_n}}$, since otherwise the desired result follows immediately by appealing to \eqref{han}. 

By the concavity of $l$, the mean-value theorem, and the Cauchy--Schwarz inequality, 
\begin{align}
	-(\tt - \t^*)\nabla^2 l(\bar{\t})(\tt - \t^*) = l(\tt) - l(\t^*) \leq \nabla l(\t^*)^\top(\tt - \t^*)\leq \|\nabla l(\t^*)\|_2\|\tt - \t^*\|_2,\label{bridge}
\end{align}
where $\bar\t$ is some point on the segment connecting $\t^*$ and $\tt$. Since $\|\widehat{\t}-\t^*\|_\infty = o_p(1)$, for all large $n$ we have $\|\t^*-\bar{\t}\|_\infty\leq 1$. Therefore, we can further bound the left-hand side of \eqref{bridge} from below as follows for all large $n$:
\begin{eqnarray}
	&&-(\tt - \t^*)\nabla^2 l(\bar{\t})(\tt - \t^*)\label{pj1}\\
	&\stackrel{\eqref{l:hess},\eqref{l:hess1}}{\geq}&\ \frac{1}{N}\sum_{i\in [N]}\sum_{j\in [m_i-1]}c_8\|[Z_{(i,j,j+1)}, \ldots, Z_{(i,j,m_i)}]^\top(\tt - \t^*)\|_2^2\ \ \ \ \ \ \ \ \ \ \  \ \  \ \ \text{($\|\t^*-\bar{\t}\|_\infty\leq 1$)}\nonumber\\
	&\stackrel{\eqref{myK}}{\geq}&\ \frac{c_8}{N}\|[\bm Q_n^\top, \bm K_n](\tt - \t^*)\|_2^2\nonumber\\
	&\geq &\ \frac{c_8\e}{N}(\|\bm Q_n^\top(\uu-\u^*)\|_2^2 + \|\bm K_n(\vv-\v^*)\|_2^2)    \ \ \ \ \ \ \ \ \ \ \ \ \ \ \ \  \ \ \ \ \  \ \ \ \ \   \ \ \ \ \  \ \text{(Assumption~\ref{ass:cosine})}\nonumber\\
	&\geq&\ c_{9}\|\vv-\v^*\|_2^2, \ \ \ \ \ \ \ \ \ \ \ \ \ \ \  \ \ \ \ \ \ \ \  \ \ \ \ \ \ \ \ \ \ \ \ \ \ \ \ \ \  \ \ \ \ \ \ \ \ \ \ \ \ \  \ \ \ \ \ \   \ \ \ \ \ \ \text{(Assumption~\ref{ass:minK})}\nonumber
\end{eqnarray}
where the appearance of the $1/N$ is due to the overriding of the notation $l$, the notation $Z_{(i,j,t)}$  are defined in \eqref{myZZ}, and the second step uses the fact that the rows of $[\bm{Q}_n^\top, \bm{K}_n]$ are a subset of the rows of the collection of matrices $\{[Z_{(i,j,j+1)}, \ldots, Z_{(i,j,m_i)}]^\top\}_{i\in [N], j\in [m_i-1]}$. 

For the upper bound in \eqref{bridge}, it holds with probability at least $1-n^{-3}$ that
\begin{align}
	\|\nabla_\u l(\t^*)\|_2^2 = \frac{1}{N^2}\sum_{k\in [n]}|\partial_k l(\t^*)|^2\stackrel{\eqref{needthis}}{\lesssim}\frac{1}{N^2}\sum_{k\in [n]}\deg(k)\log n\lesssim \frac{N\log n}{N^2} = \frac{\log n}{N},\label{pj2}
\end{align} 
where the statement in \eqref{needthis} is proved in Section~\ref{sec:hanproof} using standard concentration inequalities. 
By an almost identical argument (noting $\E[\nabla_{\v}l(\v^*)] = \nabla_{\v}l(\vv) = \bm 0$ where the expectation is taken with respect to $\pi_i$),  

\begin{spacing}{0.25}
\begin{align}
	\|\nabla_\v l(\t^*)\|_2^2 \lesssim \frac{d\log n}{N}.\label{pj3}
\end{align}
Substituting \eqref{pj1}, \eqref{pj2}, and \eqref{pj3} into \eqref{bridge} yields
\begin{align*}
	&\|\vv-\v^*\|_2^2 \\
	\lesssim &\ \sqrt{\frac{\log n}{N}}\cdot\sqrt{\|\uu-\u^*\|_2^2+\|\vv-\v^*\|_2^2}\\
	\lesssim &\ \sqrt{\frac{\log n}{N}}\cdot\sqrt{n\|\uu-\u^*\|_\infty^2+\|\vv-\v^*\|_\infty^2}\\
	\lesssim  &\ \diam(\mathcal A_{\HH_n}(\lambda))\sqrt{\frac{n\log n}{N}}\cdot\|\vv-\v^*\|_\infty &\text{(Lemma~\ref{lm:han} \& $\|\vv-\v^*\|_\infty\geq\sqrt{(\log n)/h_{\HH_n}}$)}\\
	\lesssim &\ \diam(\mathcal A_{\HH_n}(\lambda))\sqrt{\frac{\log n}{h_{\HH_n}}}\cdot\|\vv-\v^*\|_2, 
\end{align*} 
which implies $\|\vv-\v^*\|_\infty\lesssim\|\vv-\v^*\|_2\lesssim\diam(\mathcal A_{\HH_n}(\lambda))\sqrt{(\log n)/h_{\HH_n}}$. 
The proof is finished by applying \eqref{han}. 
\end{spacing}

\subsection{Proof of Lemma~\ref{lm:han}}\label{sec:hanproof}

The proof is based on a modification of the chaining technique. 
As we will see shortly, the construction of the chaining neighborhood requires a different choice of radius, which causes additional challenges in analysis. The following elementary lemma is needed in the proof. 
\begin{lemma}\label{lm:fraction}
	Let $a_1, a_2, b_1, b_2>0$. Suppose there exist $0<\e_1, \e_2\leq 1$ such that $a_2\geq \e_1a_1$ and
	\begin{align*}
		\frac{a_2}{b_2}\geq (1+\e_2)\frac{a_1}{b_1}.
	\end{align*}
	Then, 
	\begin{align*}
		\frac{a_1+a_2}{b_1+b_2}\geq \left(1+\frac{\e_1\e_2}{3}\right)\frac{a_1}{b_1}.
	\end{align*}
\end{lemma}
\begin{proof}[Proof of Lemma~\ref{lm:fraction}]
	It suffices to consider the case where $a_1= a_2/\e_1$ and $a_2/b_2 = (1+\e_2)a_1/b_1$. 
	In this case, 
	\begin{align*}
		\frac{a_1+a_2}{b_1+b_2}&= \frac{a_1 + \e_1a_1}{b_1 + \frac{\e_1 b_1}{(1+\e_2)}}\geq \left(1+\frac{\e_1\e_2}{3}\right)\frac{a_1}{b_1}.
	\end{align*}
	We finish the proof of Lemma~\ref{lm:fraction}.
\end{proof}

To begin with, note that under our choice of identifiability constraint, $\bm 1^\top(\uu-\u^*)=0$. 
Therefore, the largest and smallest components of $\uu-\u^*$ have opposite signs, that is,
\begin{align*}
	\|\uu-\u^*\|_\infty \leq  \max_{k\in [n]}(\widehat{u}_k-u^*_k) -  \min_{k\in [n]}(\widehat{u}_k-u^*_k). 
\end{align*}
To further bound the right-hand side, define
\begin{equation*}
	\alpha\in \arg\max_{k\in [n]}(\widehat{u}_k-u^*_k), \ \beta \in \arg\min_{k\in [n]}(\widehat{u}_k-u^*_k).
\end{equation*}
We will construct two nested sequences of the neighborhood centered at $\alpha$ and $\beta$ respectively and stop when they intersect, thereby chaining $\alpha$ and $\beta$ together. We will see that with an appropriate choice of neighborhood radius, the resulting neighborhood sequences are $\lambda$-weakly admissible for some absolute constant $\lambda>0$ so that their diameter can be bounded using $\diam(\A_{\HH_n}(\lambda))$. 

To proceed, let $c>0$ be an absolute constant that will be specified later.  
Consider the two sequences of neighbors started at $\alpha$ and $\beta$ respectively defined recursively as follows. 
Recall $R$ defined in Assumption~\ref{ass:unif-bdd} and $h_{\HH_n}$ defined in Definition \ref{def:ch}. 
For $\fgamma \geq 1$,  
\begin{equation}\label{olkui}
	\begin{aligned}
		&B_\fgamma = \{k: (\widehat{u}_\alpha-u_\alpha^*) - (\widehat{u}_k-u_k^*) \leq (\fgamma-1)\cdot\eta_n \},\\
		&D_\fgamma = \{k: (\widehat{u}_k-u_k^*) - (\widehat{u}_\beta-u_\beta^*) \leq (\fgamma-1)\cdot\eta_n \},
	\end{aligned}
\end{equation}
where 
\begin{align}
	\eta_n:= \max\left\{c\sqrt{\frac{\log n}{h_{\HH_n}}}, 4R\|\vv-\v^*\|_\infty\right\}.\label{myeta}
\end{align} 
Let $\fGamma_{n,1}$ and $\fGamma_{n,2}$ be two random variables defined as follows:
\begin{align*}
	&\fGamma_{n,1} = \min\left\{\fgamma: |B_\fgamma|>\frac{n}{2}\right\}& \fGamma_{n,2} = \min\left\{\fgamma: |D_\fgamma|>\frac{n}{2}\right\}.
\end{align*}
Under the above construction, $\alpha$ and $\beta$ can be chained together using $\{B_\fgamma\}_{\fgamma\in [\fGamma_{n,1}]}\cup \{D_\fgamma\}_{\fgamma\in [\fGamma_{n,2}]}$.
If we can show that both $\{B_\fgamma\}_{\fgamma\in [\fGamma_{n,1}]}$ and $\{D_\fgamma\}_{\fgamma\in [\fGamma_{n,2}]}$ are $\lambda$-weakly admissible sequences for some $\lambda>0$, then we can bound $ \max_{k\in [n]}(\widehat{u}_k-u^*_k) -  \min_{k\in [n]}(\widehat{u}_k-u^*_k)$ as follows to complete the proof:
\begin{align*}
	\max_{k\in [n]}(\widehat{u}_k-u^*_k) -  \min_{k\in [n]}(\widehat{u}_k-u^*_k)\leq 2\diam(\A_{\HH_n}(\lambda))\cdot\max\left\{c\sqrt{\frac{\log n}{h_{\HH_n}}}, 4R\|\vv-\v^*\|_\infty\right\}.
\end{align*}
Therefore, it remains to verify that both $\{B_\fgamma\}_{\fgamma\in [\fGamma_{n,1}]}$ and $\{D_\fgamma\}_{\fgamma\in [\fGamma_{n,2}]}$ are weakly admissible. For convenience, we only prove for $\{B_\fgamma\}_{\fgamma\in [\fGamma_{n,1}]}$ as the other is similar. In the subsequent analysis, we use $c_i, i\geq 1$ to denote absolute constants that depend only on $m$ and $R$. 

Recall the log-likelihood function is given by 
\begin{align*}
	&l(\t) = \sum_{i\in [N]}\sum_{j\in [m_i]}\left[s_{i j}-\log(\sum_{t=j}^{m_i}\exp{s_{it}})\right]&s_{ij}=s_{\pi_i(j)}(\t; \T_i),
\end{align*}
where the indices in $s_{ij}$ refer to the comparison index and the comparison outcome's rank, respectively. 
For notational convenience, we denote the log scores under $\tt$ and $\t^*$ as 
\begin{align*}
	&\widehat{s}_{ij}=s_{\pi_i(j)}(\tt; \T_i)&s^*_{ij}=s_{\pi_i(j)}(\t^*; \T_i). 
\end{align*}

By the first-order optimality condition of the MLE, the partial derivative with respect to the $k$th component of $\u$ is 
\begin{align}
	&\partial_{k}l(\tt) = \sum_{i: k\in \T_i}\psi(k; \T_i, \pi_i, \tt) = 0& k\in [n],\label{eq:1}
\end{align}
where  
\begin{align*}
	\psi(k; \T_i, \pi_i, \t) = 1 - \sum_{j\in [r_i(k)]}\frac{\exp\{s_{ir_i(k)}\}}{\sum_{t=j}^{m_i} \exp{s_{it}}},
\end{align*} 
for every $k\in \T_i$, $\pi_i\in\mathcal S(\T_i)$, $\t = [\bm u^\top, \v^\top]^\top\in\mathbb R^{n+d}$, and $r_i(k)$ is the rank of $k$ in $\pi_i$. 
Fixing $\T_i$, $\pi_i$, and $\t$, a crucial property of $\psi$ is that the sum over objects in a particular edge is zero:
\begin{align}
	\sum_{k\in \T_i}\psi(k; \T_i, \pi_i, \t) &= m_i - \sum_{k\in \T_i}\sum_{j\in [r_i(k)]}\frac{\exp\{s_{ir_i(k)}\}}{\sum_{t=j}^{m_i} \exp{s_{it}}}=0,\label{sm}\end{align}
where the last step follows by changing the order of summation. 
On the other hand, the true parameter $\t^*$ satisfies 
\begin{align}
	&\mathbb E[\psi(k; \T_i, \pi_i, \t^*)] = 0\label{eq:2},
\end{align}
for every $k, \T_i$.
Combining \eqref{eq:1}--\eqref{eq:2} yields the following probabilistic upper bounds for all $U\subseteq [n]$:
\begin{align}
	\sum_{k\in U}\partial_{k}l(\tt)-\sum_{k\in U}\partial_{k}l(\t^*)&\stackrel{\eqref{eq:1}, \eqref{eq:2}}{=} \sum_{k\in U}(\mathbb E[\partial_{k}l(\t^*)]-\partial_{k}l(\t^*))\nonumber\\
	& = \sum_{k\in U}\sum_{i: k\in \T_i}(\mathbb E[\psi(k; \T_i, \pi_i, \t^*)]-\psi(k; \T_i, \pi_i, \t^*))\nonumber\\
	&\stackrel{\eqref{sm}}{=}\sum_{i: \T_i\in\partial U}\sum_{k\in U\cap \T_i}(\mathbb E[\psi(k; \T_i, \pi_i, \t^*)]-\psi(k; \T_i, \pi_i, \t^*)).\label{kjl}
\end{align}
Under Assumption~\ref{ass:unif-bdd}, $\{\sum_{k\in U\cap \T_i}\psi(k; \T_i, \pi_i, \t^*)\}_{i: \T_i\in\partial U}$ are independent random variables uniformly bounded by some $c_1>0$, {where the randomness only comes from the comparison outcomes $\{\pi_i\}_{i\in [N]}$.}
By Hoeffding's inequality, with probability at least $1-n^{-4|U|}$, \eqref{kjl} is bounded by $c_2\sqrt{|U||\partial U|\log n}$.
To apply the result to $B_\fgamma$ with all $\fgamma< \fGamma_{n,1}$, which is random and satisfies $|B_\fgamma|\leq n/2$, we need to apply a union bound over $U$ with $|U|\leq n/2$ to obtain
\begin{align}
	&\mathbb P\left\{\text{$\sum_{k\in U}\partial_{k}l(\tt)-\sum_{k\in U}\partial_{k}l(\t^*)\leq c_2\sqrt{|U||\partial U|\log n}, \text{for any } U\subset [n]$ with $|U|\leq n/2$}\right\}\nonumber\\
	&\geq 1-\sum_{t=1}^{n/2}{n\choose t}n^{-4t}\geq 1-n^{-2}. \label{needthis}
\end{align}
Specifically, with probability at least $1-n^{-2}$,  
\begin{align}
	&\sum_{k\in B_\fgamma}\partial_{k}l(\tt)-\sum_{k\in B_\fgamma}\partial_{k}l(\t^*)\leq c_2\sqrt{|B_\fgamma||\partial B_\fgamma|\log n}& \fgamma< \fGamma_{n,1}.\label{upper}
\end{align}
We now derive a lower bound for the left-hand side of \eqref{upper}, from which we deduce the desired weakly admissible property. We start by rewriting \eqref{kjl} as 
\begin{align}
	&\sum_{k\in B_\fgamma}\partial_{k}l(\tt)-\sum_{k\in B_\fgamma}\partial_{k}l(\t^*)\nonumber\\
	\stackrel{\eqref{sm}}{=}&\ \sum_{i: \T_i\in\partial B_\fgamma}\sum_{k\in B_\fgamma\cap \T_i}\sum_{j\in [r_i(k)]}\left(\frac{\exp\{\widehat{s}_{ir_i(k)}\}}{\sum_{t=j}^{m_i} \exp{\widehat{s}_{it}}}-\frac{\exp\{s^*_{ir_i(k)}\}}{\sum_{t=j}^{m_i} \exp{s^*_{it}}}\right)\nonumber\\
	=&\ \sum_{i: \T_i\in\partial B_\fgamma}\sum_{j = 1}^{\max_{k\in B_\fgamma \cap \T_i}r_i(k)}\left(\frac{\sum_{t\geq j: \pi_i(t)\in B_\fgamma}\exp{\widehat{s}_{it}}}{\sum_{t=j}^{m_i} \exp{\widehat{s}_{it}}}-\frac{\sum_{t\geq j: \pi_i(t)\in B_\fgamma}\exp{s^*_{it}}}{\sum_{t=j}^{m_i} \exp{s^*_{it}}}\right). \label{zhaohan}
\end{align}

So far, the proof is similar to \cite[Theorem 5.1]{han2023unified}. However, the summand in \eqref{zhaohan} is no longer guaranteed to be non-negative based on the unilateral chaining on $\uu$, which is different from their situation thus making obtaining a useful lower bound more involved. To address this issue, we consider two cases depending on whether $\T_i\cap (B_{\fgamma+1}\setminus B_\fgamma)=\emptyset$ or not. 

\noindent\underline{\bf Case 1}: $\T_i\cap (B_{\fgamma+1}\setminus B_\fgamma)=\emptyset$. 
By definition, for all objects $k\in \T_i\cap B_\fgamma^\complement\subseteq B_{\fgamma+1}^\complement$,  
\begin{align*}
	\widehat{u}_{\alpha}-u^*_\alpha-(\widehat{u}_{k}-u^*_k)\geq \fgamma\eta_n. 
\end{align*}
Meanwhile, for the remaining vertices $j\in \T_i\cap B_\fgamma$, 
\begin{align*}
	\widehat{u}_{\alpha}-u^*_\alpha-(\widehat{u}_{j}-u^*_j)\leq (\fgamma-1)\eta_n.
\end{align*}
Combining both yields $\widehat{u}_{j}-u^*_j-(\widehat{u}_{k}-u^*_k)\geq\eta_n$. Therefore,  
\begin{align*}
	\widehat{s}_{ir_i(j)}-s^*_{ir_i(j)} - (\widehat{s}_{ir_i(k)}-s^*_{ir_i(k)})&=\widehat{u}_{j}-u^*_j-(\widehat{u}_{k}-u^*_k) + (X_{\T_i, j}-X_{\T_i, k})^\top (\vv -\v^*)\\
	&\geq \eta_n - 2R\|\vv -\v^*\|_\infty
	\stackrel{\eqref{myeta}}{\geq} \frac{1}{2}\eta_n,
\end{align*}
which implies 
\begin{align}
	\frac{\exp\{s^*_{ir_i(k)}\}}{\exp\{\widehat{s}_{ir_i(k)}\}}\geq
	\exp{\frac{\eta_n}{2}}\times\frac{\exp\{s^*_{ir_i(j)}\}}{\exp\{\widehat{s}_{ir_i(j)}\}}\geq \left(1+\frac{\eta_n}{2}\right)\times\frac{\exp\{s^*_{ir_i(j)}\}}{\exp\{\widehat{s}_{ir_i(j)}\}}.\label{gap}
\end{align}
Consequently, taking the ratio of the corresponding summand in \eqref{zhaohan} (assuming it is non-empty which occurs if all objects ranked after $j$ are from $B_\fgamma$) yields:
\begin{align}
	&\frac{\sum_{t\geq j: \pi_i(t)\in B_\fgamma}\exp{\widehat{s}_{it}}}{\sum_{t\geq j: \pi_i(t)\in B_\fgamma}\exp{s^*_{it}}}\times\frac{\sum_{t=j}^{m_i} \exp{s^*_{it}}}{\sum_{t=j}^{m_i} \exp{\widehat{s}_{it}}}\nonumber\\
	=&\ \frac{\sum_{t\geq j: \pi_i(t)\in B_\fgamma}\exp{\widehat{s}_{it}}}{\sum_{t\geq j: \pi_i(t)\in B_\fgamma}\exp{s^*_{it}}}\times\frac{\sum_{t\geq j: \pi_i(t)\in B_\fgamma} \exp{s^*_{it}}+\sum_{t\geq j: \pi_i(t)\in B^\complement_{\fgamma+1}} \exp{s^*_{it}}}{\sum_{t\geq j: \pi_i(t)\in B_\fgamma} \exp{\widehat{s}_{it}}+\sum_{t\geq j: \pi_i(t)\in B^\complement_{\fgamma+1}} \exp{\widehat{s}_{it}}}\label{aer1}.
\end{align}
To analyze the last term, we appeal to Lemma~\ref{lm:fraction} to obtain
\begin{align}
	&\frac{\sum_{t\geq j: \pi_i(t)\in B_\fgamma} \exp{s^*_{it}}+\sum_{t\geq j: \pi_i(t)\in B^\complement_{\fgamma+1}} \exp{s^*_{it}}}{\sum_{t\geq j: \pi_i(t)\in B_\fgamma} \exp{\widehat{s}_{it}}+\sum_{t\geq j: \pi_i(t)\in B^\complement_{\fgamma+1}} \exp{\widehat{s}_{it}}}\nonumber\\
	&\stackrel{\eqref{gap}}{\geq} (1+c_3\min\{\eta_n, 1\})\times\frac{\sum_{t\geq j: \pi_i(t)\in B_\fgamma} \exp{s^*_{it}}}{\sum_{t\geq j: \pi_i(t)\in B_\fgamma} \exp{\widehat{s}_{it}}},\label{aer2}
\end{align}
where the truncation is used to comply with the assumptions in Lemma~\ref{lm:fraction}. 
Hence, 
\begin{align}
	&\frac{\sum_{t\geq j: \pi_i(t)\in B_\fgamma}\exp{\widehat{s}_{it}}}{\sum_{t=j}^{m_i} \exp{\widehat{s}_{it}}}-\frac{\sum_{t\geq j: \pi_i(t)\in B_\fgamma}\exp{s^*_{it}}}{\sum_{t=j}^{m_i} \exp{s^*_{it}}}\nonumber\\
	&\stackrel{\eqref{aer1}, \eqref{aer2}}{\geq}c_3\min\{\eta_n, 1\}\times\frac{\sum_{t\geq j: \pi_i(t)\in B_\fgamma}\exp{s^*_{it}}}{\sum_{t=j}^{m_i} \exp{s^*_{it}}}
	\geq c_4\min\{\eta_n, 1\}.\label{key1} 
\end{align}
\underline{\bf Case 2}: $\T_i\cap (B_{\fgamma+1}\setminus B_\fgamma)\neq \emptyset$. 
In this case, we do not know whether the summand is positive or negative; however, we can also bound it from below using $\min\{\eta_n, 1\}$. To see this, note for any $k\in \T_i\cap (B_{\fgamma+1}\setminus B_\fgamma)$ and $j\in \T_i\cap B_{\fgamma}$,  
\begin{align*}
	\widehat{s}_{ir_i(j)}-s^*_{ir_i(j)} - (\widehat{s}_{ir_i(k)}-s^*_{ir_i(k)})\geq -2R\|\vv-\v^*\|_\infty\geq -\frac{\eta_n}{2},
\end{align*}
which implies 
\begin{align}
	\frac{\exp\{s^*_{ir_i(k)}\}}{\exp{\widehat{s}_{ir_i(k)}}}\geq \exp{-\frac{\eta_n}{2}}\times\frac{\exp\{s^*_{ir_i(j)}\}}{\exp{\widehat{s}_{ir_i(j)}}}\stackrel{\text{if $\eta_n\leq 1$}}{\geq} \left(1-\frac{\eta_n}{4}\right)\times\frac{\exp\{s^*_{ir_i(j)}\}}{\exp{\widehat{s}_{ir_i(j)}}}. \label{gap1}
\end{align}
Thus, if $\eta_n\leq 1$, taking the ratio of the corresponding summand in \eqref{zhaohan} as before yields
{\fontsize{9}{12}
	\begin{align*}
		&\frac{\sum_{t\geq j: \pi_i(t)\in B_\fgamma}\exp{\widehat{s}_{it}}}{\sum_{t\geq j: \pi_i(t)\in B_\fgamma}\exp{s^*_{it}}}\times\frac{\sum_{t=j}^{m_i} \exp{s^*_{it}}}{\sum_{t=j}^{m_i} \exp{\widehat{s}_{it}}}\nonumber\\
		=&\frac{\sum_{t\geq j: \pi_i(t)\in B_\fgamma}\exp{\widehat{s}_{it}}}{\sum_{t\geq j: \pi_i(t)\in B_\fgamma}\exp{s^*_{it}}}\times\frac{\sum_{t\geq j: \pi_i(t)\in B_\fgamma} \exp{s^*_{it}}+\sum_{t\geq j: \pi_i(t)\in B_{\fgamma+1}\setminus B_\fgamma} \exp{s^*_{it}}+\sum_{t\geq j: \pi_i(t)\in B^\complement_{\fgamma+1}} \exp{s^*_{it}}}{\sum_{t\geq j: \pi_i(t)\in B_\fgamma} \exp{\widehat{s}_{it}}+\sum_{t\geq j: \pi_i(t)\in B_{\fgamma+1}\setminus B_\fgamma} \exp{\widehat{s}_{it}}+\sum_{t\geq j: \pi_i(t)\in B^\complement_{\fgamma+1}} \exp{\widehat{s}_{it}}}\\
		\geq&\  \frac{\sum_{t\geq j: \pi_i(t)\in B_\fgamma}\exp{\widehat{s}_{it}}}{\sum_{t\geq j: \pi_i(t)\in B_\fgamma}\exp{s^*_{it}}}\times\min\left\{\frac{\sum_{t\geq j: \pi_i(t)\in B_\fgamma} \exp{s^*_{it}}}{\sum_{t\geq j: \pi_i(t)\in B_\fgamma} \exp{\widehat{s}_{it}}}, \frac{\sum_{t\geq j: \pi_i(t)\in B_{\fgamma+1}\setminus B_\fgamma} \exp{s^*_{it}}}{\sum_{t\geq j: \pi_i(t)\in B_{\fgamma+1}\setminus B_\fgamma} \exp{\widehat{s}_{it}}}, \frac{\sum_{t\geq j: \pi_i(t)\in B^\complement_{\fgamma+1}} \exp{s^*_{it}}}{\sum_{t\geq j: \pi_i(t)\in B^\complement_{\fgamma+1}} \exp{\widehat{s}_{it}}}\right\}\\
		=&\ \min\left\{1, \frac{\sum_{t\geq j: \pi_i(t)\in B_\fgamma}\exp{\widehat{s}_{it}}}{\sum_{t\geq j: \pi_i(t)\in B_\fgamma}\exp{s^*_{it}}}\times\frac{\sum_{t\geq j: \pi_i(t)\in B_{\fgamma+1}\setminus B_\fgamma} \exp{s^*_{it}}}{\sum_{t\geq j: \pi_i(t)\in B_{\fgamma+1}\setminus B_\fgamma} \exp{\widehat{s}_{it}}}\right\}\stackrel{\eqref{gap1}}{\geq}1-\frac{\eta_n}{4}.
\end{align*}}
As a result, 
\begin{align}
	&\frac{\sum_{t\geq j: \pi_i(t)\in B_\fgamma}\exp{\widehat{s}_{it}}}{\sum_{t=j}^{m_i} \exp{\widehat{s}_{it}}}-\frac{\sum_{t\geq j: \pi_i(t)\in B_\fgamma}\exp{s^*_{it}}}{\sum_{t=j}^{m_i} \exp{s^*_{it}}}\nonumber\\
	&\geq -\frac{\eta_n}{4}\times\frac{\sum_{t\geq j: \pi_i(t)\in B_\fgamma}\exp{s^*_{it}}}{\sum_{t=j}^{m_i} \exp{s^*_{it}}}\geq-\eta_n = -\min\{\eta_n, 1\}. \label{key2}
\end{align}
If $\eta_n>1$, the last step still holds.

Now define $|\{e\in\partial B_\fgamma: e\cap (B_{\fgamma+1}\setminus B_\fgamma)\neq\emptyset\}|/|\partial B_\fgamma| = \kappa$. 
Substituting \eqref{key1} and \eqref{key2} into \eqref{zhaohan} and combining with \eqref{upper}, we have 
\begin{align*}
	[c_4(1-\kappa)-\mM\kappa]\min\{1, \eta_n\}|\partial B_\fgamma|\leq c_2\sqrt{|B_\fgamma||\partial B_\fgamma|\log n}. 
\end{align*}
If $\kappa\leq \lambda:=c_4/(2c_4+2M)$, then 
\begin{align*}
	\frac{c_4}{2}\min\{1, \eta_n\}\leq c_2\sqrt{\frac{\log n}{h_{\HH_n}}}\xrightarrow{n\to\infty} 0,
\end{align*}
which implies $\eta_n<1$, so that 
\begin{align*}
	\max\left\{c\sqrt{\frac{\log n}{h_{\HH_n}}}, 4R\|\vv-\v^*\|_\infty\right\}\leq \frac{2c_2}{c_4}\sqrt{\frac{\log n}{h_{\HH_n}}}.
\end{align*}
For all $c>2c_2/c_4$, this is a contradiction. (Note that such a choice can be made at the beginning as both $c_2$ and $c_4$ are absolute constants.) Hence, $\kappa > \lambda$.
This implies that $\{B_\fgamma\}_{\fgamma\in [\fGamma_{n,1}]}$ is a $\lambda$-weakly admissible sequence.

\subsection{Proof of Proposition~\ref{prop:1+}}\label{p:prop5}
For each fixed $n, i, j$ and $\v$ with $\|\v\|_\infty=1$, 

\begin{align}
	\E[|\Delta X_i[j,:]\v|_2^2] = \v^\top\mathrm{Cov}[\Delta X_i[j,:]]\v\geq\min_{\|\bm w\|_2=1}\bm w^\top\mathrm{Cov}[\Delta X_i[j,:]]\bm w\geq\delta'^2,\label{hgyes}
\end{align}
where the last step follows from the assumption $\mathrm{Cov}[\Delta X_i[j,:]]\geq\delta'^2\bm I$. Meanwhile, under Assumption~\ref{ass:unif-bdd}, it follows from H\"{o}lder's inequality that $|\Delta X_i[j,:]\v|\leq \|\Delta X_i[j,:]\|_1\|\v\|_\infty< 2R$, so that 
\begin{align}
	\E[|\Delta X_i[j,:]\v|_2^2]\leq \frac{\delta'^2}{4} + 4R^2\P\left\{|\Delta X_i[j,:]\v|>\frac{\delta'}{2}\right\}. \label{hgyes1}
\end{align}
Combining \eqref{hgyes} and \eqref{hgyes1} yields
\begin{align*}
	\P\left\{|\Delta X_i[j,:]\v|>\frac{\delta'}{2}\right\}\geq\frac{3\delta'^2}{16R^2}. 
\end{align*}
Consequently, either 
\begin{equation}
	\P\left\{\Delta X_i[j,:]\v>\frac{\delta'}{2}\right\}\geq\frac{3\delta'^2}{32R^2} \ \ \text{or} \ \ \P\left\{\Delta X_i[j,:]\v<-\frac{\delta'}{2}\right\}\geq\frac{3\delta'^2}{32R^2}.\label{pangzai1}
\end{equation}
Without loss of generality, we may assume $\delta'<1$ and $R\geq 1$.
The first part of \eqref{pangzai1} is exactly what we wish to obtain. On the other hand, if the latter part of \eqref{pangzai1} occurs, then 
\begin{align*}
	\E[\Delta X_i[j,:]\v \cdot \mathbb I_{\{\Delta X_i[j,:]\v>0 \} }] &= -\E[\Delta X_i[j,:]\v\cdot \mathbb I_{\{\Delta X_i[j,:]\v\leq0\}}] \quad\quad (\text{$\E[\Delta X_i[j,:]] = \bm 0$})\\
	&\geq \frac{\delta'}{2}\cdot\P\left\{\Delta X_i[j,:]\v<-\frac{\delta'}{2}\right\}\\
	&\geq\frac{3\delta'^3}{64R^2}, 
\end{align*}
Considering the random variable $\Delta X_i[j,:]\v \cdot \mathbb I_{\{\Delta X_i[j,:]\v>0 \} }$ and applying a similar argument as {\eqref{hgyes1}}, one shows that 
\begin{align}
	\P\left\{\Delta X_i[j,:]\v>\frac{3\delta'^3}{128R^2}\right\}\geq \frac{3\delta'^3}{256R^3}. \label{pangzai2}
\end{align}
Combining the first part of \eqref{pangzai1} and  \eqref{pangzai2} yields $\P\{\Delta X_i[j,:]\v>\delta\}>c_0$, where
\begin{align*}
	c_0 = \frac{3\delta'^3}{256R^3} \ \ \text{and} \ \ \delta = \frac{3\delta'^3}{128R^2}. 
\end{align*}
Assumption~\ref{ass:avoidpa} follows by first taking minimum over $j$ and $i$ followed by sending $n\to\infty$. 

\subsection{Proof of Proposition~\ref{prop:dl}}

We first note that $\bm K_n^\top\bm K_n = \sum_{i\in [N]}\Delta X_i^\top\Delta X_i$ is a sum of independent self-adjoint random matrices, where $\Delta X_i$ are defined in \eqref{DeltaX} and uniformly bounded under Assumption~\ref{ass:unif-bdd}. Under Assumption~\ref{ass:avoidpa}, there exist $\delta>0$ such that for all large $n$ and $i\in [N]$, 
\begin{align*}
	\sigma_{\min}(\E[\Delta X_i^\top\Delta X_i]) = \min_{\|\v\|_2=1}\E[\|\Delta X_i\bm v\|_2^2]\stackrel{\eqref{mass}}{\geq}\frac{c_0\delta^2}{2d},
\end{align*}
where $c_0$ is defined in \eqref{mass} and the last step follows from the Cauchy--Schwarz inequality and Markov's inequality. Consequently, $\sigma_{\min}(\E[\bm K_n^\top\bm K_n])\gtrsim N\geq n$. It then follows from the matrix Chernoff bound \cite[Theorem 1.1]{tropp2012user} that, with probability at least $1-n^{-3}$,  
\begin{align*}
	\sigma_{\min}(\K) = \sqrt{\sigma_{\min}(\bm K_n^\top\bm K_n)}\geq \sqrt{\frac{1}{2}\sigma_{\min}(\E[\bm K_n^\top\bm K_n])}\gtrsim \sqrt{N},
\end{align*}
proving the first statement. 

To prove the second statement, we first consider the setting where $M=2$. In this case, all edges of $\HH_n$ are pairwise, so $\bm K_n$ has independent rows. Now we consider the random process $G(\bm x, \bm z) = \bm x^\top\K\bm z$ indexed by $(\bm x, \bm z)\in\Xi:=\mathbb S(\range(\bm Q_n^\top))\times \mathbb S(\R^d)$, where $\mathbb S(\cdot)$ denotes the unit sphere in the corresponding (sub)space. Endowing $\Xi$ with the metric ${\mathsf D}((\bm x, \bm z), (\bm x', \bm z')) = \|\bm x-\bm x'\|_2 + \|\bm z-\bm z'\|_2$, $G(\bm x, \bm z)$ has subgaussian increments with respect to ${\mathsf D}$; see \cite[Lemma 3.4.2, Definition 8.1.1]{vershynin2018high}. Let $\mathcal N(\Xi, {\mathsf D}, t)$ be the $t$-covering number of $\Xi$ under ${\mathsf D}$, and we have $\log\mathcal N(\Xi, {\mathsf D}, t)\lesssim (n+d)\log (1/t)$ and $\log\mathcal N(\Xi, {\mathsf D}, t) = 0$ if $t>\diam(\Xi)$. By Dudley's inequality \cite[Theorem 8.1.6]{vershynin2018high}, with probability at least $1-n^{-3}$, 
\begin{align}
	\sup_{(\bm x, \bm z)\in\Xi}|G(\bm x, \bm z)|\lesssim\int_0^\infty\sqrt{\log\mathcal N(\Xi, {\mathsf D}, t)} \d t + \sqrt{\log n}\lesssim\sqrt{n+d}. \label{dudley}
\end{align}
Conditioning on the intersection of the events $\sigma_{\min}(\K)\gtrsim\sqrt{N}$ and \eqref{dudley}, which holds with probability at least $1-2n^{-3}$,  
\begin{align}
	\sup_{\substack{\bm x\in \range(\bm Q_n^\top), \bm y\in\range(\K)\\ \|\bm x\|_2=\|\bm y\|_2 = 1}}\bm x^\top \bm y = \sup_{(\bm x, \bm z)\in\Xi}\frac{|G(\bm x, \bm z)|}{\|\K\bm z\|_2}\lesssim  \frac{1}{\sqrt{N}}\sup_{(\bm x, \bm z)\in\Xi}|G(\bm x, \bm z)|\lesssim\sqrt{\frac{n+d}{N}}.\label{jsheu}
\end{align}

For the general case where $M>2$, we claim that $\bm K_n$ can be written as $\bm K_n=\sum_{i\in [M-1]}\bm K_{n, i}$, where $\bm K_{n, i}$ satisfy the following properties. For $i\in [M-1]$, 
\begin{enumerate}
	\item [(i)] each row of $\bm K_{n, i}$ is either a row of $\K$ or a zero vector;
	\item [(ii)] the rows of $\bm K_{n, i}$ are independent;
	\item [(iii)] $\bm K_{n ,i}$ have mutually orthogonal column space across $i\in [M-1]$; 
	\item [(iv)] the number nonzero rows of $\bm K_{n,i}$ is at least $N/M$. 
\end{enumerate}
For instance, one way to obtain such a decomposition is to first initialize $\bm K_{n, i}$ as zero matrices and then copy each row of $\bm K$ to some $\bm K_{n, i}$ at a time, where $\bm{K}_{n,i}$  is randomly selected from $\{\bm K_{n, 1}, \ldots, \bm K_{n, M-1}\}$ that are independent of the copied row and have the most zero rows.

Based on the properties (i), (ii), and (iv) described above, one can apply a similar argument as the case $M=2$ together with a union bound to conclude that for all $i\in [M-1]$,  with probability at least $1-2(M-1)n^{-3}$, 
\begin{align*}
	\sup_{\substack{\bm x\in \range(\bm Q_n^\top), \bm y\in\range(\bm K_{n, i})\\ \|\bm x\|_2=\|\bm y\|_2 = 1}}\bm x^\top \bm y\lesssim\sqrt{\frac{M(n+d)}{N}}.
\end{align*}
Consequently, 
\begin{align*}
	\sup_{\substack{\bm x\in \range(\bm Q_n^\top), \bm y\in\range(\K)\\ \|\bm x\|_2=\|\bm y\|_2 = 1}}\bm x^\top \bm y &= \sup_{\substack{\bm x\in \range(\bm Q_n^\top), \z\in\R^d\\ \|\bm x\|_2=\|\K\z\|_2 = 1}}\bm x^\top \K\z\\
	&\leq \sum_{j\in [M-1]}\sup_{\substack{\bm x\in \range(\bm Q_n^\top), \z\in\R^d\\ \|\bm x\|_2=\|\K\z\|_2 = 1}}\bm x^\top \bm K_{n, j}\z\\
	&\leq \sum_{j\in [M-1]}\sup_{\substack{\bm x\in \range(\bm Q_n^\top), \z\in\R^d\\ \|\bm x\|_2=\|\bm K_{n, j}\z\|_2 = 1}}\bm x^\top \bm K_{n, j}\z\\
	& = \sum_{j\in [M-1]}\sup_{\substack{\bm x\in \range(\bm Q_n^\top), \bm y\in\range(\bm K_{n, j})\\ \|\bm x\|_2=\|\bm y\|_2 = 1}}\bm x^\top \bm y\\
	&\lesssim \sqrt{\frac{M^3(n+d)}{N}},
\end{align*}
where the second inequality uses the property (iii). 

\subsection{Proof of Proposition~\ref{prop:rg}}\label{p:prop8}

The proof is rather technical and is based on \cite[Lemma 5.1]{han2023unified}; we only outline the difference and explain how their argument can be applied in our setting.

In \cite[Lemma 5.1]{han2023unified}, the authors studied the diameter of admissible sequences, which consist of strictly increasing sequence of vertices $\{A_j\}_{j=1}^J$ with
\begin{align}
	&\frac{|\{e\in \partial A_{j}: e\subseteq A_{j+1}\}|}{|\partial A_j|}\geq\frac{1}{2}& 1\leq j <J. \label{tempad}
\end{align}
They showed that the diameter of all admissible sequences in $\HH_n$ is of order $(\xi_{n, +}\log n)/\xi_{n, -}$ and $\log n$ almost surely if $\HH_n$ is sampled from NURHM and HSBM, respectively. The parameter $1/2$ is unimportant and can be generalized to any positive constant $\lambda>0$. The key idea of their approach to establishing these results lies in obtaining explicit lower and upper bound for $|\partial A_j|$ and $|\{e\in \partial A_{j}: e\subseteq A_{j+1}\}|$ based on $\xi_{n, \pm}$ and $\zeta_{n,-}$ for any admissible sequence, and use \eqref{tempad} to deduce the desired growth rate of $|A_j|$. For the upper bound $|\{e\in \partial A_{j}: e\subseteq A_{j+1}\}|$, they used $|\{e\in \partial A_{j}: e\subseteq A_{j+1}\}|\leq |\partial (A_{j+1}\setminus A_j)|$ and estimate the latter using standard concentration inequalities. 

In our setting, we need to control the diameter of all weakly $\lambda$-admissible sequences, which are a superset of admissible sequences $\{A_j\}_{j=1}^J$ satisfying
\begin{align}
	&\frac{|\{e\in \partial A_{j}: e\cap (A_{j+1}\setminus A_j)\neq\emptyset\}|}{|\partial A_j|}\geq\frac{1}{2}& 1\leq j <J, \label{tempwad}
\end{align}
where we take $\lambda = 1/2$ for convenience. 
At first sight, this involves bounding the diameter of a set consisting of more elements than in the admissible setting.
If we attempt to follow a similar route as before, then we need to find an upper bound for $|\{e\in \partial A_{j}: e\cap (A_{j+1}\setminus A_j)\neq\emptyset\}|$. We may use the same bound $|\partial (A_{j+1}\setminus A_j)|$ thanks to the observation $\{e\in \partial A_{j}: e\cap (A_{j+1}\setminus A_j)\neq\emptyset\}\subseteq \partial (A_{j+1}\setminus A_j)$. Consequently, their results on the diameter of admissible sequences can also be used in our setting. On the other hand, they showed that $h_{\HH_n}$ is lower bounded by the order of $\xi_{n, -}$ and $\zeta_{n, -}$ and this result is independent of the weakly admissible sequence. Hence, the desired results follow.

\subsection{Proof of Theorem \ref{thm:randdesign}}\label{p:thm3}

According to Theorem~\ref{thm:main}, it suffices to show that Assumptions~\ref{ass:aep}, \ref{ass:cosine}, \ref{ass:minK}, and \ref{ass:topology} hold a.s. This is straightforward based on the results obtained so far. In particular, under the assumptions in Theorem~\ref{thm:randdesign}:
\begin{itemize}
	\item Assumption~\ref{ass:aep} holds a.s due to Proposition~\ref{prop:at_verify}, Proposition~\ref{prop:1+}, and Remark~\ref{rem:3}.
	\item Assumptions~\ref{ass:cosine}-\ref{ass:minK} hold a.s. due to Proposition~\ref{prop:1+}, Proposition~\ref{prop:dl}, and Remark~\ref{rem:3}.
	\item Assumption~\ref{ass:topology} holds a.s. due to Proposition~\ref{prop:rg}.
\end{itemize}

\section{Additional illustrations}\label{FFF}
\subsection{Illustration of identifiability}
In this section, we provide a toy example to  illustrate the concepts  related to model identifiability.  Consider a scenario with $N=4$ multiple comparisons among $n=4$ objects, as shown in Table~\ref{tab:ex1}. The corresponding graphs $\HH(\VV,\EE)$  and $\HH_\br(\VV,\EE_\br)$ are presented in Figure~\ref{fig:exa1_1}.

\begin{table}[htbp!]
	\centering
	\begin{spacing}{1}
		\begin{tabular}{c|c|c|c}
			Comparison & Object & Covariates & Rank  \\ \hline\hline
			Match 1 ($\T_1$)& 1& $X_{\T_1,1}=(-1,1)^\top$& 1\\
			Match 1 ($\T_1$)& 2& $X_{\T_1,2}=(3,-1)^\top$& 4\\
			Match 1 ($\T_1$)& 3& $X_{\T_1,3}=(-4,2)^\top$& 2\\
			Match 1 ($\T_1$)& 4& $X_{\T_1,4}=(0,-3)^\top$& 3\\
			Match 2 ($\T_2$)& 3& $X_{\T_2,3}=(-2,1)^\top$& 2\\
			Match 2 ($\T_2$)& 4& $X_{\T_2,4}=(-3,4)^\top$& 1\\
			Match 3 ($\T_3$)& 2& $X_{\T_3,2}=(0,-1)^\top$& 1\\
			Match 3 ($\T_3$)& 4& $X_{\T_3,4}=(-3,4)^\top$& 2
		\end{tabular}
		\caption{A toy example of multiple comparisons.}\label{table:ex1}
		\label{tab:ex1}
	\end{spacing}
\end{table}

\begin{figure}[htbp!]
	\centering
	\includegraphics[width=0.9\linewidth]{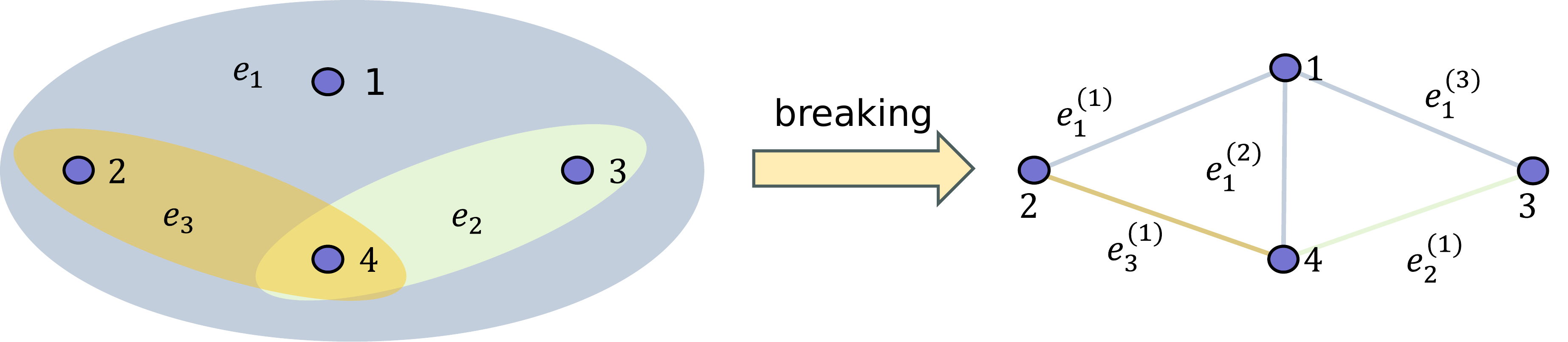}
	\caption{The comparison graph is derived from Table~\ref{table:ex1}. In this example, each hyperedge $e_i$ on the left, with $|e_i|=m_i$, is decomposed into $m_i-1$ pairwise edges $e_i^{(j)}$ on the right.}
	\label{fig:exa1_1}
\end{figure}

According to Figure~\ref{fig:exa1_1} and \eqref{DeltaX}, we have
\begin{align*}
	\Delta X_1 = \begin{bmatrix}
		4  & -2 \\
		-3 &  1 \\
		1  & -4 \\
	\end{bmatrix},
	\quad
	\Delta X_2 = \begin{bmatrix}
		-1 & 3 \\
	\end{bmatrix},
	\quad
	\Delta X_3 = \begin{bmatrix}
		-3  & 5 \\
	\end{bmatrix}.
\end{align*}
The corresponding $\Q$ and $\K$ can be computed as 
\begin{align*}
	\Q = 
	\begin{bmatrix}
		-1 & -1 & -1  & 0  &  0 \\
		1  & 0  &  0  & 0  & -1 \\
		0  & 1  &  0  & -1 &  0 \\
		0  & 0  &  1  & 1  &  1 \\
	\end{bmatrix},
	\quad
	\K=\begin{bmatrix}
		4  & -2 \\
		-3 &  1 \\
		1  & -4 \\
		-1 & 3 \\
		-3  & 5 \\
	\end{bmatrix}.
\end{align*}

We define $\triangle_1= (1,2,4)$ and $\triangle_2=(1,3,4)$ as two triangles in the graph, and let $f_1, f_2$ be the first and second columns of $\K$, respectively. We calculate the curl of $f_1, f_2$ over the triangles $\triangle_1,\triangle_2$: 
\begin{align*}
	\big[\curl(f_1)(\triangle_1),\curl(f_2)(\triangle_1)\big] &= \Delta X_1[1,:] + \Delta X_3[1,:] - \Delta X_1[2,:] \\
	&= [4,-2]+[-3,5] - [-3,1] = [4,2],\\
	\big[\curl(f_1)(\triangle_2),\curl(f_2)(\triangle_2)\big] &= \Delta X_1[3,:] + \Delta X_2[1,:] - \Delta X_1[2,:]  \\
	&=  [1,-4]+[-1,3] - [-3,1] = [3,-2].
\end{align*}
In this case, the PlusDC is identifiable since there exists  $\mathcal T_{\triangle}=\{\triangle_1,\triangle_2\}$ such that 
\begin{equation*}
	{\rm det}(\bm T_{\triangle})
	={\rm det}\left(
	\begin{bmatrix}
		\curl(f_1)(\triangle_1),\curl(f_2)(\triangle_1)\\
		\curl(f_1)(\triangle_2),\curl(f_2)(\triangle_2)\\
	\end{bmatrix}\right)
	={\rm det}\left(
	\begin{bmatrix}
		4  & 2 \\
		3  & -2\\
	\end{bmatrix}
	\right)=-14\neq0.
\end{equation*}

\subsection{Illustration of Assumption~\ref{ass:topology}}
\subsubsection*{\underline{Illustration of $h_\HH$}:}
The modified Cheeger constant $h_\HH$ generalizes the  Cheeger constant in spectral graph theory \cite[Section 2.2]{MR1421568}  to hypergraphs. Fixing a nonempty set $U\subset\V =[n]$, 
\begin{align*}
	&h_\HH(U) = \frac{|\partial U|}{\min\{|U|, |U^\complement|\}}
\end{align*}
reflects the connectivity between $U$ and $U^\complement$. The modified Cheeger constant $h_\HH = \min_{U\subset[n]}h_\HH(U)$ represents the minimum connectivity across any partition of $\V$. Thus, $h_\HH$ is upper bounded by the minimum vertex degree in the hypergraph, which measures the minimum effective sample size for an arbitrary vertex. We provide two examples on how to compute \( h_\HH \) in Figure~\ref{ima:cc}.

\begin{figure}[htbp]
	\begin{minipage}{0.4\linewidth}
		\centering		
		\includegraphics[width= 0.9\linewidth,   clip]{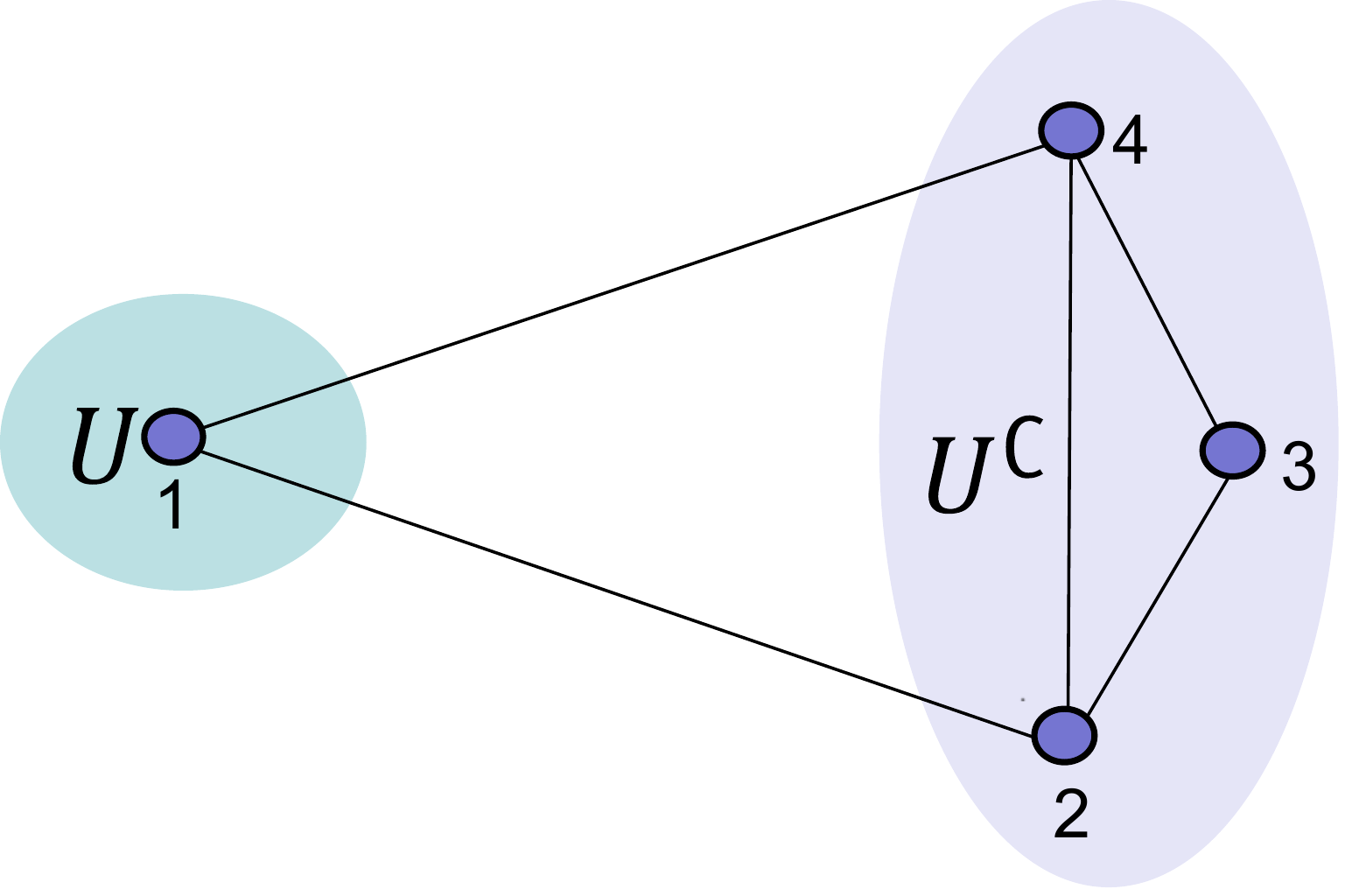}
		\caption*{ $h_\HH(U)=2$}
	\end{minipage}
	\hfill
	\begin{minipage}{0.42\linewidth}
		\centering		
		\includegraphics[width= 0.9\linewidth,   clip]{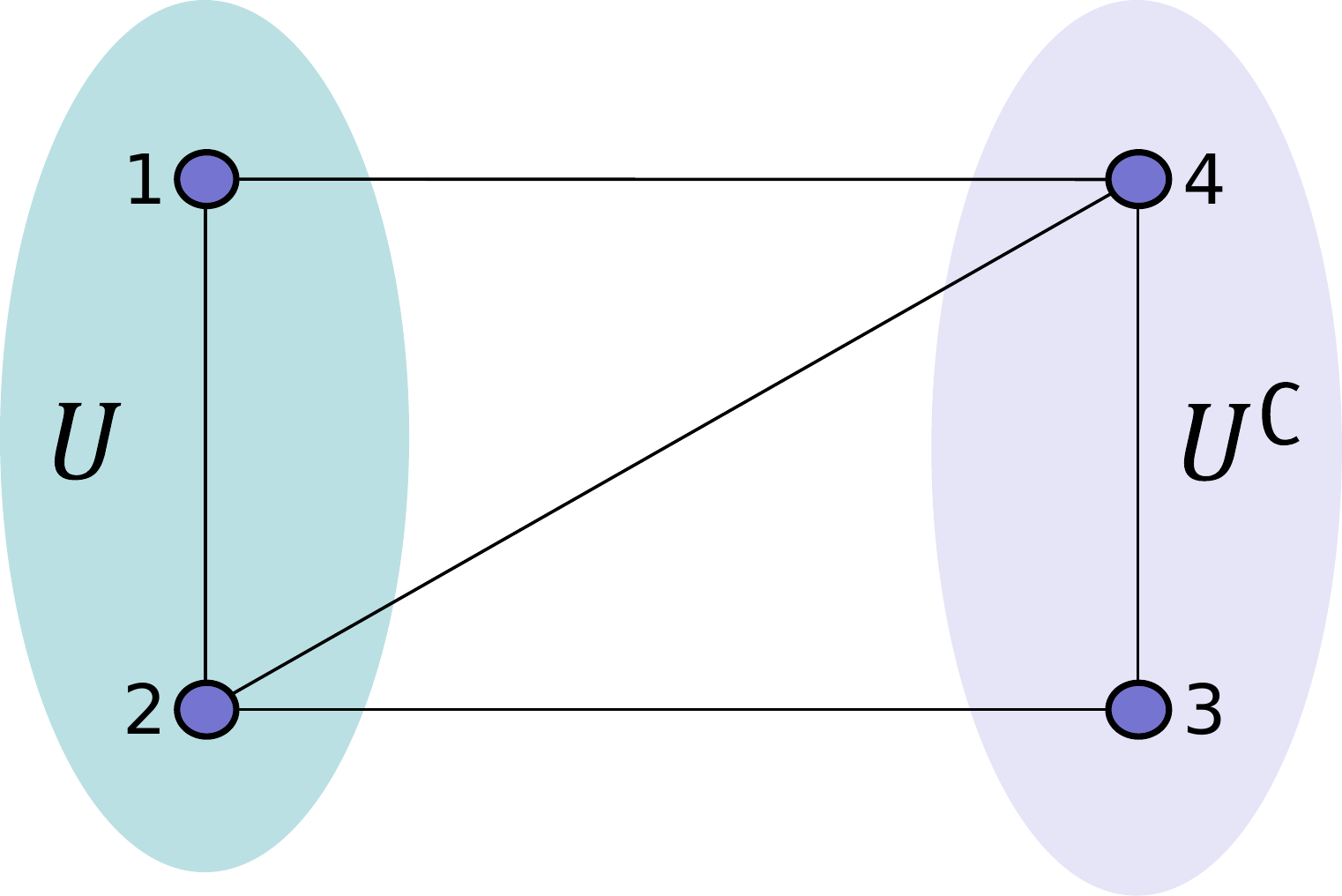}
		\caption*{$h_\HH(U)=1.5$}
	\end{minipage}
	\caption{In these examples, we consider a pairwise graph. In the left panel, the set $U$ is defined as $U= \{1\} $, resulting $|\partial U| = 2$. In the right panel, the set $U$ is defined as $U= \{1, 2\} $, yielding $|\partial U| = 3$.}\label{ima:cc}
\end{figure}
\subsubsection*{\underline{Illustration of $\diam(\A_\HH(\lambda))$}:}
The quantity $\diam(\A_\HH(\lambda))$ is the maximum cardinality of a $\lambda$-weakly admissible sequence $\{A_j\}_{j\in [J]}\in\A_\HH(\lambda)$, that is, 
\begin{align*}
	&\frac{|\{e\in \partial A_{j}: e\cap (A_{j+1}\setminus A_j)\neq\emptyset\}|}{|\partial A_j|}\geq\lambda& 1\leq j <J.
\end{align*}
It is easy to check from the definition that $\diam(\A_\HH(\lambda))$ is a nonincreasing function in $\lambda$. 
We provide one examples on how to compute $\diam(\A_\HH(\lambda))$ of the graph presented in Figure~\ref{ima:was}. 
\begin{figure}[htbp]
	\centering		
	\includegraphics[width= 0.46\linewidth,   clip]{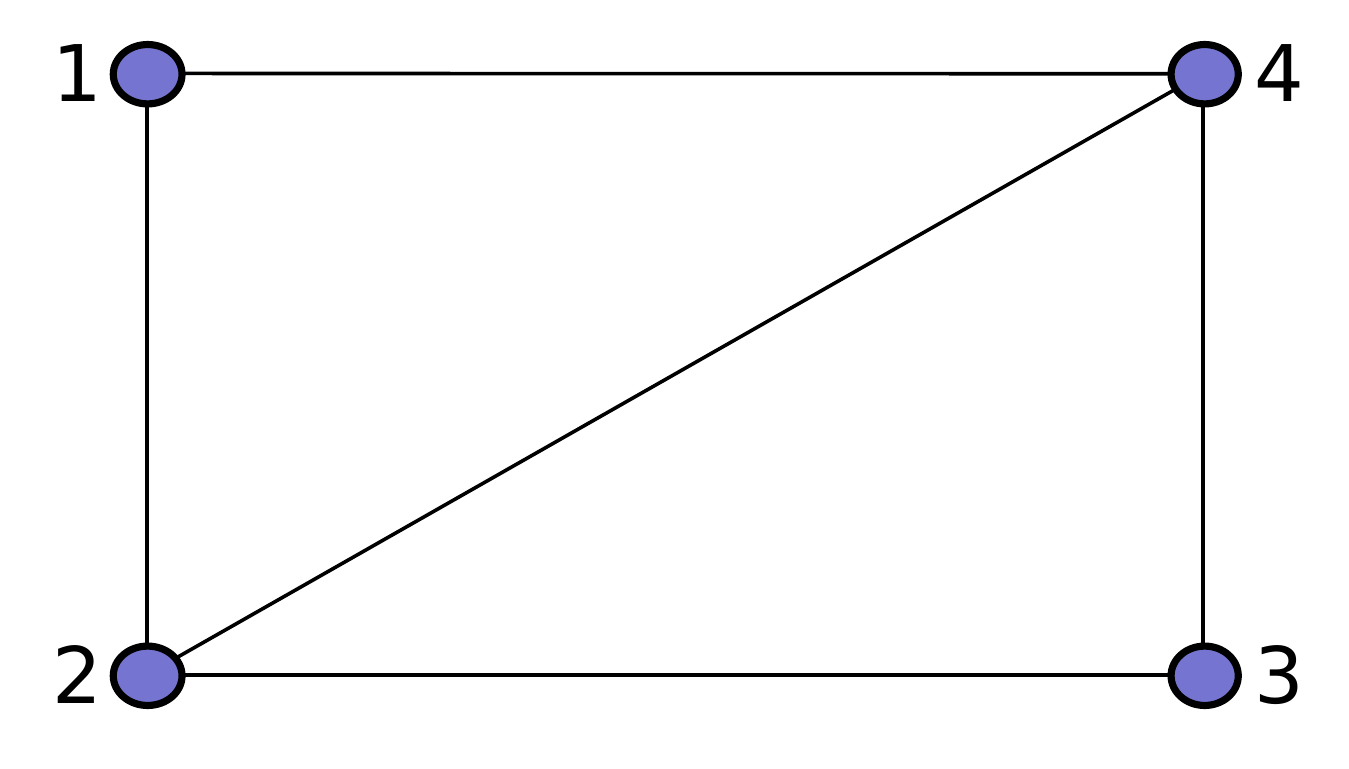}
	\caption{An example to demonstrate the $\lambda$-weakly admissible sequences.}
\label{ima:was}
\end{figure}

In this case, all the  $\lambda$-weakly admissible sequences (with $\lambda=0.6$) that contain the vertex set $\{1,2,3,4\}$ are listed as follows:
\begin{itemize}
	\item $A_1 = \{1\}$, $A_2 =  \{1,2,3,4\}$
	\item $A_1 = \{1\}$, $A_2 = \{1, 2, 4\}$, $A_3 =\{1, 2, 3, 4\}$
	\item $A_1 = \{2\}$, $A_2 = \{1,2,3,4\}$
	\item $A_1 = \{2\}$, $A_2= \{1, 2, 3\}$, $A_3 = \{1, 2, 3, 4\}$
	\item $A_1 = \{2\}$, $A_2= \{1, 2, 4\}$, $A_3 = \{1, 2, 3, 4\}$
	\item $A_1 = \{2\}$, $A_2= \{2, 3, 4\}$, $A_3 = \{1, 2, 3, 4\}$
	\item $A_1 = \{3\}$, $A_2 = \{1, 2, 3, 4\}$
	\item $A_1 = \{3\}$, $A_2 = \{2, 3, 4\}$, $A_3 = \{1, 2, 3, 4\}$
	\item $A_1 = \{4\}$, $A_2 = \{1, 2, 3, 4\}$
	\item $A_1 = \{4\}$, $A_2 = \{1, 2, 4\}$, $A_3 = \{1, 2, 3, 4\}$
	\item $A_1 = \{4\}$, $A_2 = \{1, 3, 4\}$, $A_3 = \{1, 2, 3, 4\}$
	\item $A_1 = \{4\}$, $A_2 = \{2, 3, 4\}$, $A_3 = \{1, 2, 3, 4\}$
	\item $A_1 = \{1, 2\}$, $A_2 = \{1, 2, 3, 4\}$
	\item $A_1 = \{1, 2\}$, $A_2=\{1, 2, 4\}$, $A_3 = \{1, 2, 3, 4\}$
	\item $A_1 = \{1, 3\}$, $A_2 = \{1, 2, 3, 4\}$
	\item $A_1 = \{1, 4\}$, $A_2 = \{1, 2, 3, 4\}$
	\item $A_1 = \{1, 4\}$, $A_2=\{1, 2, 4\}$, $A_3 = \{1, 2, 3, 4\}$
	\item $A_1 = \{2, 3\}$, $A_2 = \{1, 2, 3, 4\}$
	\item $A_1 = \{2, 3\}$, $A_2=\{2, 3, 4\}$, $A_3 = \{1, 2, 3, 4\}$
	\item $A_1 = \{2, 4\}$, $A_2 = \{1, 2, 3, 4\}$
	\item $A_1 = \{3, 4\}$, $A_2 = \{1, 2, 3, 4\}$
	\item $A_1 = \{3, 4\}$, $A_2 = \{2, 3, 4\}$, $A_3 = \{1, 2, 3, 4\}$
	\item $A_1=\{1, 2, 3\}$, $A_2=\{1, 2, 3, 4\}$
	\item $A_1=\{1, 2, 4\}$, $A_2=\{1, 2, 3, 4\}$
	\item $A_1=\{1, 3, 4\}$, $A_2=\{1, 2, 3, 4\}$
	\item $A_1=\{2, 3, 4\}$, $A_2=\{1, 2, 3, 4\}$
\end{itemize}

As a result, $\diam(\A_\HH(0.6)) = 3$ in Figure~\ref{ima:was}. More generally, one can check that 
\begin{align*}
\diam(\A_\HH(\lambda)) = \begin{cases}
	3 & 0.5 < \lambda\leq 1,\\
	4 & 0<\lambda\leq 0.5.  
\end{cases} 
\end{align*}
\color{black}
\end{appendix}

\bibliographystyle{apalike}
\bibliography{ref}

\end{document}